\documentclass[12pt,singlespace]{mitthesis}
\usepackage{geometry}

\usepackage{lgrind}
\usepackage{pdfsync}

\usepackage{amsmath, amsthm, amssymb}
\usepackage{graphicx}
\usepackage{epsfig,psfrag}
\usepackage{MnSymbol}
\usepackage{setspace}
\usepackage{subfigure}
\usepackage{epstopdf}
\usepackage{algorithmic}
\usepackage{verbatim} 
\usepackage{enumerate}
\usepackage{url}
\usepackage{color}

\newtheorem{thm}{Theorem}
\newtheorem{lemma}[thm]{Lemma}
\newtheorem{prop}[thm]{Proposition}
\newtheorem{cor}[thm]{Corollary}
\newtheorem{defn}[thm]{Definition}

\newtheorem{clm}[thm]{Claim}

\newcommand{\mcal}{\mathcal}
\newcommand{\mb}{\mathbb}
\newcommand{\mbf}{\mathbf}

\newcommand{\gm}{\gamma}
\newcommand{\ep}{\epsilon}

\newcommand{\beq}{\begin{equation}}
\newcommand{\eeq}{\end{equation}}
\newcommand{\beqn}{\begin{eqnarray}}
\newcommand{\eeqn}{\end{eqnarray}}
\newcommand{\benum}{\begin{enumerate}}
\newcommand{\eenum}{\end{enumerate}}

\newcommand{\etmz}{\end{itemize}}
\newcommand{\rar}{\rightarrow}

\newcommand{\bthm}{\begin{thm}}
\newcommand{\ethm}{\end{thm}}
\newcommand{\bdefn}{\begin{defn}}
\newcommand{\edefn}{\end{defn}}
\newcommand{\lt}{\left}
\newcommand{\rt}{\right}
\newcommand{\nnb}{\nonumber}
\newcommand{\lbl}{\label}
\newcommand{\ity}{\infty}
\newcommand{\bydef}{\stackrel{\triangle}{=}}
\newcommand{\ftomb}{}
\newcommand{\diamonds}{}

\newcommand{\buw}{\mathbf{W}^N}
\newcommand{\bua}{\mathbf{A}^N}
\newcommand{\bul}{\mathbf{L}^N}
\newcommand{\buc}{\mathbf{C}^N}
\newcommand{\bus}{\mathbf{S}^N}
\newcommand{\buv}{\mathbf{V}^N}
\newcommand{\bux}{\mathbf{X}^N}

\newcommand{\buf}{\mathbf{F}}
\newcommand{\buh}{\mathbf{H}}

\newcommand{\blw}{\mathbf{w}}
\newcommand{\bla}{\mathbf{a}}
\newcommand{\bll}{\mathbf{l}}
\newcommand{\blc}{\mathbf{c}}
\newcommand{\blm}{\mathbf{m}}
\newcommand{\bls}{\mathbf{s}}
\newcommand{\blv}{\mathbf{v}}
\newcommand{\blx}{\mathbf{x}}
\newcommand{\bly}{\mathbf{y}}

\newcommand{\pb}{\mathbb{P}}
\newcommand{\N}{\mathbb{N}}
\newcommand{\R}{\mathbb{R}}

\newcommand{\zp}{{\mathbb{Z}_+}}

\newcommand{\sps}{\mathcal{S}}
\newcommand{\spv}{\mathcal{V}}
\newcommand{\spq}{\mathcal{Q}}
\newcommand{\spc}{\mathcal{C}}
\newcommand{\vinf}{\overline{\mathcal{V}}^\infty}

\begin{document}

%
%
%
%
%
%
%
%

\title{On the Power of Centralization in Distributed Processing}

\author{Kuang Xu}
       \prevdegrees{B.S. Electrical Engineering (2009) \\\vspace{6pt} University of Illinois at Urbana-Champaign \vspace{10pt} }
\department{Department of Electrical Engineering and Computer Science}

\degree{Master of Science in Electrical Engineering and Computer Science}

\degreemonth{June}
\degreeyear{2011}
\thesisdate{May 2, 2011}


\supervisor{John N. Tsitsiklis}{Clarence J. Lebel Professor of Electrical Engineering}

\chairman{Leslie A. Kolodziejski}{Chairman, Department Committee on Graduate Theses}

\maketitle



\newpage
$\quad$
 \pagestyle{empty}
 \setcounter{savepage}{\thepage}

\newpage

\cleardoublepage
 \pagestyle{empty}
\setcounter{savepage}{\thepage}

\begin{abstractpage}
%
%
%
In this thesis, we propose and analyze a multi-server model that captures a performance trade-off between centralized and distributed processing. In our model, a fraction $p$ of an available resource is deployed in a centralized manner (e.g., to serve a most-loaded station) while the remaining fraction $1-p$ is allocated to local servers that can only serve requests addressed specifically to their respective stations.

Using a fluid model approach, we demonstrate a surprising \emph{phase transition} in the \emph{steady-state delay}, as $p$ changes: in the limit of a large number of stations, and when \emph{any amount} of centralization is available ($p>0$), the average queue length in steady state scales as $\log_{\frac{1}{1-p}}{\frac{1}{1-\lambda}}$ when the traffic intensity $\lambda$ goes to 1. This is \emph{exponentially smaller} than the usual $M/M/1$-queue delay scaling of $\frac{1}{1-\lambda}$, obtained when all resources are fully allocated to local stations ($p=0$). This indicates a strong qualitative impact of even a small degree of centralization.

We prove convergence to a fluid limit, and characterize both the transient and steady-state behavior of the finite system, in the limit as the number of stations $N$ goes to infinity. We show that the sequence of queue-length processes converges to a \emph{unique} fluid trajectory (over any finite time interval, as $N \rightarrow \infty$), and that this fluid trajectory converges to a unique invariant state $\mathbf{v}^I$, for which a simple closed-form expression is obtained. We also show that the steady-state distribution of the $N$-server system concentrates on $\mathbf{v}^I$ as $N$ goes to infinity. 
\end{abstractpage}

\cleardoublepage
%
%



\section*{Acknowledgments}

\begin{doublespace}
$\quad$ \\

I would like to express my deepest gratitude to my thesis supervisor, Professor John N. Tsitsiklis, for his invaluable guidance and support over the last two years. 

\vspace{10pt}

I would like to thank Yuan Zhong (MIT) for his careful reading of an early draft of this thesis.

\vspace{10pt}

I would like to thank Professor Devavrat Shah (MIT) and Professor Julien M. Hendrickx (Catholic University of Louvain) for helpful discussions on related subjects.

\vspace{10pt}

I would like to thank my family for their love and constant support over the years.

\vspace{10pt}

This research was supported in part by an MIT Jacobs Presidential Fellowship, a Xerox-MIT Fellowship, a Siebel Scholarship, and NSF grant CCF-0728554. 

\end{doublespace}

%


\pagestyle{plain}
\tableofcontents
\newpage
\listoffigures
\newpage
\listoftables

\chapter{Introduction}
\lbl{sec:intro}

\section{Distributed versus Centralized Processing}

The tension between \emph{distributed} and \emph{centralized} processing seems to have existed ever since the inception of computer networks. Distributed processing allows for simple implementation and robustness, while a centralized scheme guarantees optimal utilization of computing resources at the cost of implementation complexity and communication overhead. A natural question is how performance varies with the \emph{degree of centralization}. Such an understanding is of great interest in the context of, for example, infrastructure planning (static) or task scheduling (dynamic) in large server farms or cloud clusters, which involve a trade-off between performance (e.g., delay) and cost (e.g., communication infrastructure, energy consumption, etc.). In this thesis, we address this problem by formulating and analyzing a multi-server model with an \emph{adjustable} level of centralization. We begin by describing informally two motivating applications.

\subsection{Primary Motivation: Server Farm with Local and Central Servers}
\lbl{sec:mot1}

Consider a server farm consisting of $N$ stations, depicted in Figure \ref{fig:mot1}. Each station is fed with an independent stream of tasks, arriving at a rate of $\lambda$ tasks per second, with $0< \lambda < 1$.\footnote{Without loss of generality, we normalize so that the largest possible arrival rate is $1$.} Each station is equipped with a \emph{local server} with identical performance; the server is local in the sense that it only serves its own station. All stations are also connected to a single \emph{centralized server} which will serve a station with the longest queue whenever possible.

We consider an $N$-station system. The system designer is granted a total amount $N$ of divisible \emph{computing resources} (e.g., a collection of processors). In a loose sense (to be formally defined in Section \ref{sec:mod}), this means that the system is capable of processing $N$ tasks per second when fully loaded. The system designer is faced with the problem of allocating computing resources to local and central servers. Specifically, for some $p \in (0,1)$, each of the $N$ local servers is able to process tasks at a maximum rate of $1-p$ tasks per second, while the centralized server, equipped with the remaining computing power, is capable of processing tasks at a maximum rate of $pN$ tasks per second. The parameter $p$ captures the amount of centralization in the system. Note that since the total arrival rate is $\lambda N$, with $0<\lambda<1$, the system is underloaded for any value $p \in (0,1)$.

When the arrival processes and task processing times are random, there will be times when some stations are empty while others are loaded. Since a local server cannot help another station process tasks, the total computational resources will be better utilized if a larger fraction is allocated to the central server. However, a greater degree of centralization (corresponding to a larger value of $p$) entails more frequent communications and data transfers between the local stations and the central server, resulting in higher infrastructure and energy costs. 

How should the system designer choose the coefficient $p$? Alternatively, we can ask an even more fundamental question: is there any significant difference between having a small amount of centralization (a small but positive value of $p$), and complete decentralization (no central server and $p=0$)?

\subsection{Secondary Motivation: Partially Centralized Scheduling}
\lbl{sec:mot2} 
Consider the system depicted in Figure \ref{fig:mot2}. The arrival assumptions are the same as in Section~\ref{sec:mot1}. However, there is no local server associated with a station; all stations are served by a single central server. Whenever the central server becomes free, it chooses a task to serve as follows. With probability $p$, it processes a task from a most loaded station, with an arbitrary tie-breaking rule. Otherwise, it processes a task from a station selected uniformly at random; if the randomly chosen station has an empty queue, the current round is in some sense ``wasted'' (to be formalized in Section \ref{sec:mod}).

This second interpretation is intended to model a scenario where resource allocation decisions are made at a centralized location on a \emph{dynamic} basis, but \emph{communications} between the decision maker (central server) and local stations are costly or simply unavailable from time to time. While it is intuitively obvious that longest-queue-first (LQF) scheduling is more desirable, up-to-date system state information (i.e., queue lengths at all stations) may not always be available to the central server. Thus, the central server may be forced to allocate service blindly. In this setting, a system designer is interested in setting the optimal \emph{frequency} ($p$) at which global state information is collected, so as to balance performance and communication costs.

As we will see in the sequel, the system dynamics in the two applications are captured by the \emph{same} mathematical structure under appropriate stochastic assumptions on task arrivals and processing times, and hence will be addressed jointly in this thesis.

\begin{figure}
\centering
\includegraphics[scale=0.85]{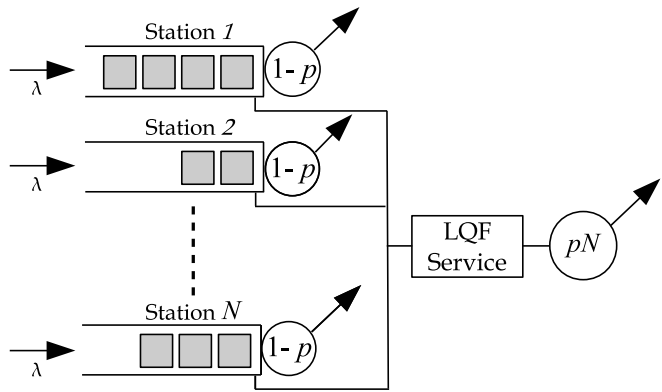}
\caption{Server farm with local and central servers.}
\label{fig:mot1}
\end{figure}

\begin{figure}
\centering
\includegraphics[scale=0.85]{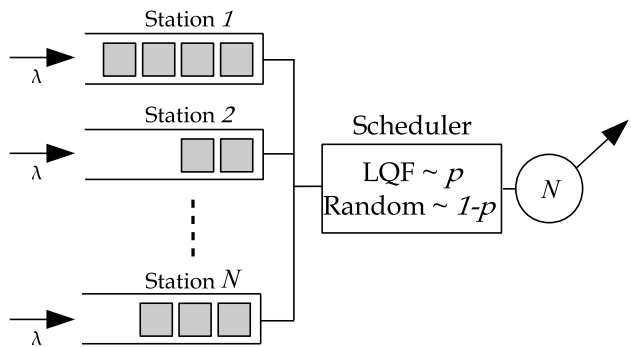}
\caption{Centralized scheduling with communication constraints.}
\label{fig:mot2}
\vspace{-0.3cm}
\end{figure}

\section{Overview of Main Contributions}
\lbl{sec:maincontr}
We provide here an overview of the main contributions. Exact statements of the results will be provided in Chapter \ref{sec:sumres} after the necessary terminology has been introduced.

Our goal is to study the performance implications of varying degrees of centralization, as expressed by the coefficient $p$. To accomplish this, we use a so-called \emph{fluid approximation}, whereby the queue length dynamics at the local stations are approximated, as $N \rar \ity$, by a deterministic \emph{fluid model}, governed by a system of ordinary differential equations (ODEs).

Fluid approximations typically involve results of two flavors: qualitative results derived from the fluid model that give insights into the performance of the original finite stochastic system, and technical convergence results (often mathematically involved) that justify the use of such approximations. We summarize our contributions along these two dimensions:

\benum
\item On the {\bf qualitative end}, we derive an exact expression for the invariant state of the fluid model, for any given traffic intensity $\lambda$ and centralization coefficient $p$, thus characterizing the steady-state distribution of the queue lengths in the system as $N \rar \ity$. This enables a system designer to use any performance metric and analyze its sensitivity with respect to $p$. In particular, we show a surprising \emph{exponential phase transition} in the scaling of average system delay as the load approaches capacity ($\lambda \rar 1$) (Corollary \ref{cor:logsc} in Section \ref{sec:qualresult}): when an \emph{arbitrarily small} amount of centralized computation is applied ($p>0$), the average queue length in the system scales as\footnote{The $\sim$ notation used in this thesis is to be understood as \emph{asymptotic closeness} in the following sense: $\lt[f\lt(x\rt) \sim g\lt(x\rt), \mbox{ as } x \rar 1 \rt] \Leftrightarrow \lim_{x \rar 1} \frac{f\lt(x\rt)}{g\lt(x\rt)} = 1$.}
\beq
\mb{E}(Q) \sim \log_{\frac{1}{1-p}}{\frac{1}{1-\lambda}},
\eeq
as the traffic intensity $\lambda$ approaches $1$. This is \emph{drastically smaller} than the $\frac{1}{1-\lambda}$ scaling obtained if there is no centralization ($p=0$).\footnote{When $p=0$, the system degenerates into $N$ independent queues. The $\frac{1}{1-\lambda}$ scaling comes from the mean queue length expression for $M/M/1$ queues.} This suggests that for large systems, even a small degree of centralization provides significant improvements in the system's delay performance, in the heavy traffic regime. 

\item On the {\bf technical end}, we show that:
\benum
\item Given any finite initial queue sizes, and with high probability, the evolution of the queue length process can be approximated by the unique solution to a fluid model, over any finite time interval, as $N\rar\ity$.
\item All solutions to the fluid model converge to a unique invariant state, as $t \rar \ity$, for any finite initial condition (global stability).
\item The steady-state distribution of the finite system converges to the invariant state of the fluid model as $N\rar\ity$.
\eenum
The most notable technical challenge comes from the fact that the longest-queue-first policy used by the centralized server causes discontinuities in the drift in the fluid model (see Section \ref{sec:fmodel} for details). In particular, the classical approximation results for Markov processes (see, e.g., \cite{KZ05}), which rely on a Lipschitz-continuous drift in the fluid model, are hard to apply. Thus, in order to establish the finite-horizon approximation result ($a$), we employ a sample-path based approach: we prove tightness of sample paths of the queue length process and characterize their limit points. Establishing the convergence of steady-state distributions in ($c$) also becomes non-trivial due to the presence of discontinuous drifts. To derive this result, we will first establish the uniqueness of solutions to the fluid model and a uniform speed of convergence of stochastic sample paths to the solution of the fluid model over a compact set of initial conditions.

\eenum

\section{Related Work} 

To the best of our knowledge, the proposed model for the splitting of processing resources between distributed and central servers has not been studied before. However, the fluid model approach used in this thesis is closely related to, and partially motivated by, the so-called supermarket model of randomized load-balancing. In that literature, it is shown that by routing tasks to the shorter queue among a small number ($d\geq2$) of randomly chosen queues, the probability that a typical queue has at least $i$ tasks (denoted by $\bls_i$) decays as $\lambda^{\frac{d^i-1}{d-1}}$ (super-geometrically) as $i \rar \ity$ (\cite{DOB96},\cite{MM96}); see also the survey paper \cite{MM01} and references therein. However, the approach used in load-balancing seems to offer little improvement when adapted to scheduling. In \cite{AD08}, a variant of the randomized load-balancing policy was applied to a scheduling setting with channel uncertainties, where the server always schedules a task from a longest queue among a finite number of randomly selected queues. It was observed that $\bls_i$ no longer exhibits super-geometric decay and only moderate performance gain can be harnessed from sampling more than one queue. 

In our setting, the system dynamics causing the exponential phase transition in the average queue length scaling are significantly different from those for the randomized load-balancing scenario. In particular, for any $p>0$, the tail probabilities $\bls_i$ become zero for sufficiently large finite $i$, which is significantly faster than the super-geometric decay in the supermarket model. 


On the technical side, arrivals and processing times used in supermarket models are often memoryless (Poisson or Bernoulli) and the drifts in the fluid model are typically continuous with respect to the underlying system state. Hence convergence results can be established by invoking classical approximation results, based on the convergence of the generators of the associated Markov processes. An exception is \cite{BLB10}, where the authors generalized the supermarket model to arrival and processing times with general distributions. Since the queue length process is no longer Markov, the authors rely on an asymptotic independence property of the limiting system and use tools from statistical physics to establish convergence. 

Our system remains Markov with respect to the queue lengths, but a significant technical difference from the supermarket model lies in the fact that the longest-queue-first service policy introduces \emph{discontinuities} in the drifts. For this reason, we need to use a more elaborate set of techniques to establish the connection between stochastic sample paths and the fluid model. Moreover, the presence of discontinuities in the drifts creates challenges even for proving the uniqueness of solutions for the deterministic fluid model. (Such uniqueness is needed to establish convergence of steady-state distributions.) Our approach is based on a state representation that is different from the one used in the popular supermarket models, which turns out to be surprisingly more convenient to work with for establishing the uniqueness of solutions to the fluid model.

Besides the queueing-theoretic literature, similar fluid model approaches have been used in many other contexts to study systems with large populations. Recent results in \cite{NG10} establish convergence for finite-dimensional symmetric dynamical systems with drift discontinuities, using a more probabilistic (as opposed to sample path) analysis, carried out in terms of certain conditional expectations. We believe that it is possible to prove our results using the methods in \cite{NG10}, with additional work. However, the coupling approach used in this thesis provides strong physical intuition on the system dynamics, and avoids the need for additional technicalities from the theory of multi-valued differential inclusions.

Finally, there has been some work on the impact of service flexibilities in routing problems, motivated by applications such as multilingual call centers. These date back to the seminal work of \cite{FS78}, with a more recent numerical study in \cite{HE09}. These results show that the ability to route a portion of customers to a least-loaded station can lead to a constant-factor improvement in average delay under diffusion scaling. This line of work is very different from ours, but in a broader sense, both are trying to capture the notion that system performance in a random environment can benefit significantly from even a small amount of centralized coordination. 

\section{Organization of the Thesis}
Chapter \ref{sec:setup} introduces the precise model to be studied, our assumptions, and the notation to be used throughout. The main results are summarized in Chapter \ref{sec:sumres}, where we also discuss their implications along with some numerical results. The remainder of the thesis is devoted to establishing the technical results, and the reader is referred to Section \ref{sec:proofstrat} for an overview of the proofs. The steps of two of the more technical proofs are outlined in the main text, while the complete proofs are relegated to Appendix \ref{app:techproof}. The procedure and parameters used for numerical simulations are described in Appendix \ref{app:sim}.

\chapter{Model and Notation}
\lbl{sec:setup}

This chapter covers the modeling assumptions, system state representations, and mathematical notation, which will be used throughout the thesis. We will try to provide the intuition behind our modeling choices and assumptions if possible. In some cases, we will point the reader to explanations that will appear later in the thesis, if the ideas involved are not immediately obvious at this stage.

\section{Model}
\lbl{sec:mod}
We present our model using terminology that corresponds to the server farm application in Section \ref{sec:mot1}. Time is assumed to be continuous.

\benum
\item {\bf System.} The system consists of $N$ parallel stations. Each station is associated with a queue which stores the tasks to be processed. The queue length (i.e., number of tasks) at station $n$ at time $t$ is denoted by $Q_n(t)$,  $n \in \{1,2,\ldots,N\}, t \geq 0$. For now, we do not make any assumptions on the queue lengths at time $t=0$, other than that they are finite.

\item {\bf Arrivals.} Stations receive streams of incoming tasks according to independent Poisson processes with a common rate $\lambda \in [0,1)$.

\item {\bf Task Processing.} We fix a centralization coefficient $p \in [0,1]$.
\benum
\item {\bf Local Servers}. The local server at station $n$ is modeled by an independent Poisson clock with rate $1-p$ (i.e., the times between two clock ticks are independent and exponentially distributed with mean $\frac{1}{1-p}$). If the clock at station $n$ ticks at time $t$, we say that a {\bf local service token} is generated at station $n$. If $Q_n(t) \neq 0$, exactly one task from station $n$ ``consumes'' the service token and leaves the system immediately. Otherwise, the local service token is ``wasted'' and has no impact on the future evolution of the system.


\item {\bf Central Server}. The central server is modeled by an independent Poisson clock with rate $Np$. If the clock ticks at time $t$ at the central server, we say that a {\bf central service token} is generated. If the system is non-empty at $t$ (i.e., $\sum_{n=1}^NQ_n(t)>0$), exactly one task from some station $n$, chosen uniformly at random out of the stations with a \emph{longest queue} at time $t$, consumes the service token and leaves the system immediately. If the whole system is empty, the central service token is wasted. 
\eenum
\eenum

\vspace{4pt}

{\bf Physical interpretation of service tokens}. We interpret $Q_n(t)$ as the number of tasks whose service has not yet started. For example, if there are four tasks at station $n$, one being served and three that are waiting, then $Q_n(t)=3$. The use of \emph{local service tokens} can be thought of as an approximation to a \emph{work-conserving}\footnote{A server is work-conserving if it is never idle when the queue is non-empty.} server with exponential service time distribution in the following sense. Let $t_k$ be the $k$th tick of the Poisson clock at the server associated with station $n$. If $Q_n(t_k-)>0$,\footnote{ Throughout the thesis, we use the short-hand notation $f(t-)$ to denote the left limit $\lim_{s \uparrow t} f(s)$.} the ticking of the clock can be thought of as the completion of a previous task, so that the server ``fetches'' a new task from the queue to process, hence decreasing the queue length by $1$. Therefore, as long as the queue remains non-empty, the time between two adjacent clock ticks can be interpreted as the service time for a task. However, if the local queue is currently empty, i.e., $Q_n(t_k-)=0$, the our modeling assumption implies that the local server does nothing until the next clock tick at $t_{k+1}$, even if some task arrives during the period $(t_{k},t_{k+1})$. Alternatively, this can be thought of as the server creating a ``virtual task'' whenever it sees an empty queue, and pretending to be serving the virtual task until the next clock tick. In contrast, a work-conserving server would start serving the next task immediately upon its arrival. We have chosen to use the service token setup, mainly because it simplifies analysis, and it can also be justified in the following ways.
\benum
\item  Because of the use of virtual tasks, one would expect the resulting queue length process under our setup to provide an \emph{upper bound} on queue length process in the case of a work-conserving server. We do not formally prove such a dominance relation in this thesis, but note that a similar dominance result in $GI/GI/n$ queues was proved recently (Proposition 1 of \cite{GG11}). 

\item Since the discrepancy between the two setups only occurs when the server sees an empty queue, one would also expect that the queue length processes under the two cases become close as traffic intensity $\lambda \rar 1$, in which case the queue will be non-empty for most of the time. 
\eenum
The same physical interpretation applies to the central service tokens.

\vspace{18pt}

{\bf Mathematical equivalence between the two motivating applications}. We note here that the scheduling application in Section \ref{sec:mot2} corresponds to the same mathematical model. The arrival statistics to the stations are obviously identical in both models. For task processing, note that we can equally imagine all service tokens as being generated from a single Poisson clock with rate $N$. Upon the generation of a service token, a coin is flipped to decide whether the token will be directed to process a task at a random station (corresponding to a \emph{local service token}), or a station with a longest queue (corresponding to a \emph{central service token}). Due to the Poisson splitting property, this produces identical statistics for the generation of local and central service tokens as described above. 

\section {System State}
\lbl{sec:sysstate}
Let us fix $N$. Since all events (arrivals of tasks and service tokens) are generated according to independent Poisson processes, the queue length vector at time $t$, $\lt(Q_1(t),Q_2(t),\rt.$ $\lt.\ldots,Q_N(t)\rt)$, is Markov. Moreover, the system is fully symmetric, in the sense that all queues have identical and independent statistics for the arrivals and local service tokens, and the assignment of central service tokens does not depend on the specific identity of stations besides their queue lengths. Hence we can use a Markov process $\lt\{\bus_i(t)\rt\}_{i=0}^\ity$ to describe the evolution of a system with $N$ stations, where
\beq
\lbl{eq:normq}
\bus_i(t) \bydef \frac{1}{N}\sum_{n=1}^N \mb{I}_{\lt[i, \ity \rt)}\lt(Q_n(t)\rt), \quad i \geq 0.
\eeq
Each coordinate $\bus_i\lt(t\rt)$ represents the fraction of queues with at least $i$ tasks. Note that $\bus_0(t)=1$ for all $t$ and $N$ according to this definition. We call $\bus\lt(t\rt)$ the {\bf normalized queue length process}. We also define the {\bf aggregate queue length process} as
\beq
\lbl{eq:aggq}
\buv_i\lt(t\rt) \bydef \sum_{j=i}^\infty \bus_j\lt(t\rt), \quad i \geq 0.
\eeq
Note that 
\beq
\bus_i(t)=\buv_i(t)-\buv_{i+1}(t).
\eeq
In particular, this means that $\buv_0(t)-\buv_{1}(t)=\bus_0(t)=1$. Note also that 
\beq
\buv_1\lt(t\rt) = \sum_{j=1}^\infty \bus_j\lt(t\rt)
\eeq
is equal to the \emph{average queue length} in the system at time $t$. When the total number of tasks in the system is finite (hence all coordinates of $\buv$ are finite), there is a straightforward bijection between $\bus$ and $\buv$. Hence $\buv(t)$ is Markov and also serves as a valid representation of the system state. While the $\bus$ representation admits a more intuitive interpretation as the ``tail'' probability of a typical station having at least $i$ tasks, it turns out the $\buv$ representation is significantly more convenient to work with, especially in proving uniqueness of solutions to the associated fluid model, and the detailed reasons will become clear in the sequel (see Section \ref{subsec:vvss} for an extensive discussion on this topic). For this reason, we will be working mostly with the $\buv$ representation, but will in some places state results in terms of $\bus$, if doing so provides a better physical intuition.

\section{Notation}

Let $\zp$ be the set of non-negative integers. The following sets will be used throughout the thesis (where $M$ is a positive integer): 
\beq
\lbl{eq:sps}
\sps \bydef \lt\{\bls \in [0,1]^\zp : 1 = \bls_0 \geq \bls_1 \geq \cdots \geq 0 \rt\}, 
\eeq
\beq
\lbl{eq:spsity}
\overline{\sps}^M \bydef \lt\{\bls \in \sps: \sum_{i=1}^\infty \bls_i \leq M \rt\}, \,\,\,\, \overline{\sps}^\infty \bydef \lt\{\bls \in \sps: \sum_{i=1}^\infty \bls_i < \infty \rt\}, 
\eeq
\beq
\overline{\spv}^M \bydef \lt\{\blv: \blv_i = \sum_{j=i}^\infty \bls_j, \mbox{ for some } \bls \in \overline{\sps}^M\rt\},
\eeq
\beq
\overline{\spv}^\infty \bydef \lt\{\blv:\blv_i = \sum_{j=i}^\infty \bls_j, \mbox{ for some } \bls \in \overline{\sps}^\infty\rt\},
\eeq
\beq
\spq^N \bydef \lt\{\mathbf{x}\in \R^\zp: \mathbf{x}_i = \frac{K}{N},\mbox{ for some } K \in \zp, \forall i \rt\}.
\eeq \normalsize
We define the weighted $L_2$ norm $\lt\| \, \cdot \, \rt\|_w$ on $\R^\zp$ as
\beq
\lbl{eq:infnorm}
\lt\| \mathbf{x}-\mathbf{y} \rt\|_w^2 = \sum_{i=0}^\ity \frac{\lt|\mathbf{x}_i-\mathbf{y}_i\rt|^2}{2^i}, \quad \mathbf{x}, \mathbf{y} \in \R^\zp.
\eeq
%

In general, we will be using bold letters to denote vectors and ordinary letters for scalars, with the exception that a bold letter with a subscript (e.g., $\blv_i$) is understood as a (scalar-valued) component of a vector. Upper-case letters are generally reserved for random variables (e.g., $\mbf{V}^{(0,N)}$) or scholastic processes (e.g., $\buv(t)$), and lower-case letters are used for constants (e.g., $\blv^0$) and deterministic functions (e.g., $\blv(t)$). Finally, a function is in general denoted by $x(\cdot)$, but is sometimes written as $x(t)$ to emphasize the type of its argument. 
\chapter{Summary of Main Results}
\lbl{sec:sumres}

In this chapter, we provide the exact statements of our main results. The main approach of our work is to first derive key performance guarantees using a simpler fluid model, and then apply probabilistic arguments (e.g., Functional Laws of Large Numbers) to formally justify that such guarantees also carry over to sufficiently large finite stochastic systems. Section \ref{sec:fmodel} gives a formal definition of the core fluid model used in this thesis, along with its physical interpretations. Section \ref{sec:qualresult} contains results that are derived by analyzing the dynamics of the fluid model, and Section \ref{sec:summaryconvg} contains the more technical convergence theorems which justify the accuracy of approximating a finite system using the fluid model approach. The proofs for the theorems stated here will be developed in later chapters.

\section{Definition of Fluid Model}
\lbl{sec:fmodel}

\begin{defn}
\label{def:fl} {\bf (Fluid Model)}
Given an initial condition $\blv^0 \in \overline{\spv}^\infty$, a function $\blv(t): [0,\infty) \rar \overline{\spv}^\infty$ is said to be a {\bf solution to the fluid model} (or {\bf fluid solution} for short) if:
\benum [(1)]
\item $\blv(0) = \blv^0$;
\item for all $t\geq 0$, 
\beqn
& \blv_0(t)-\blv_1(t)= 1, \\
& \mbox{and } 1\geq\blv_i(t)-\blv_{i+1}(t) \geq \blv_{i+1}(t)-\blv_{i+2}(t) \geq 0,  \quad \forall i \geq 0;
\eeqn
\item for almost all $t \in [0,\ity)$, and for every $i\geq1$, $\blv_i(t)$ is differentiable and satisfies
\beq
\lbl{eq:dft}
\dot{\blv}_i\lt(t\rt)= \lambda \lt(\blv_{i-1}-\blv_{i}\rt) - \lt(1-p\rt)\lt(\blv_{i}-\blv_{i+1}\rt) - g_i\lt(\blv\rt),
\eeq
where 
\begin{equation}
\lbl{eq:gdef0}
g_i\lt(\blv\rt) = \lt\{ \begin{array}{ll}
p, & \blv_i>0, \\
\min\lt\{\lambda\blv_{i-1}, p\rt\}, & \blv_i=0, \blv_{i-1}>0, \\
0, & \blv_i=0, \blv_{i-1}=0.
\end{array} \right.
\end{equation}
\eenum
\end{defn}

We can write Eq.\ \eqref{eq:dft} more compactly as
\beq
\dot{\blv}\lt(t\rt) = \buf\lt(\blv\rt),
\eeq
where 
\beq
\lbl{eq:driftF}
\buf_i\lt(\blv\rt)\bydef\lambda \lt(\blv_{i-1}-\blv_{i}\rt) - \lt(1-p\rt)\lt(\blv_{i}-\blv_{i+1}\rt) - g_i\lt(\blv\rt).
\eeq
We call $\buf\lt(\blv\rt)$ the {\bf drift at point $\blv$}. 

\vspace{15pt}

{\bf Interpretation of the fluid model}. The solution to the fluid model, $\blv(t)$, can be thought of as a deterministic approximation to the sample paths of $\buv(t)$ for large values of $N$. Conditions (1) and (2) correspond to initial and boundary conditions, respectively. The boundary conditions reflect the \emph{physical constraints} of the finite system. In particular, the condition that $\blv_0(t)-\blv_1(t)= 1$ corresponds to the fact that 
\beq
\buv_0(t)-\buv_1(t)\bydef \bus_0(t)=1,
\eeq
where $\bus_0(t)$ is the fraction of queues with a non-negative queue length, which is by definition $1$. Similarly, the condition that 
\beq
1\geq\blv_i(t)-\blv_{i+1}(t) \geq \blv_{i+1}(t)-\blv_{i+2}(t) \geq 0,  \quad \forall i \geq 0,
\eeq
is a consequence of 
\beq
\lt(\buv_i(t)-\buv_{i+1}(t)\rt)-\lt(\buv_{i+1}(t)-\buv_{i+2}(t)\rt)\bydef \bus_i(t) - \bus_{i+1}(t) \in [0,1],
\eeq
where $\bus_i(t) - \bus_{i+1}(t)$ is the faction of queues at time $t$ with exactly $i$ tasks, which is by definition between $0$ and $1$. 

\vspace{5pt}
We now provide some intuition for each of the drift terms in Eq.\ \eqref{eq:dft}:

{\bf I.} $\lambda \lt(\blv_{i-1}-\blv_{i}\rt)$: This term corresponds to arrivals. When a task arrives at a station with $i-1$ tasks, the system has one more queue with $i$ tasks, and $\bus_i$ increases by $\frac{1}{N}$. However, the number of queues with at least $j$ tasks, for $j\neq i$, does not change. Thus, $\bus_{i}$ is the only one that is incremented. Since $\buv_i\bydef \sum_{k=i}^\ity \bus_k$, this implies that $\buv_i$ is increased by $\frac{1}{N}$ \emph{if and only if} a task arrives at a queue with at least $i-1$ tasks. Since all stations have an identical arrival rate $\lambda$, the probability of $\buv_i$ being incremented upon an arrival to the system is equal to the fraction of queues with at least $i-1$ tasks, which is $\buv_{i-1}(t)-\buv_i(t)$. We take the limit as $N\rar\ity$, multiply by the total arrival rate, $N\lambda$, and then multiply by the increment due to each arrival, $\frac{1}{N}$, to obtain the term $\lambda \lt(\blv_{i-1}-\blv_{i}\rt)$.

{\bf II.} $\lt(1-p\rt)\lt(\blv_{i}-\blv_{i+1}\rt)$: This term corresponds to the completion of tasks due to \emph{local} service tokens. The argument is similar to that for the first term.

{\bf III.} $g_i\lt(\blv\rt)$: This term corresponds to the completion of tasks due to \emph{central} service tokens.
	\benum
	\item $g_i\lt(\blv\rt)=p, \mbox{ if } \blv_i>0$. If $i>0$ and $\blv_i>0$, then there is a positive fraction of queues with at least $i$ tasks. Hence the central server is working at full capacity, and the rate of decrease in $\blv_i$ due to central service tokens is equal to the (normalized) maximum rate of the central server, namely $p$.
	\item $g_i\lt(\blv\rt)=\min\lt\{\lambda\blv_{i-1}, p\rt\}, \mbox{ if } \blv_i=0, \blv_{i-1}>0$. This case is more subtle. Note that since $\blv_i=0$, the term $\lambda\blv_{i-1}$ is equal to $\lambda(\blv_{i-1}-\blv_i)$, which is the rate at which $\blv_i$ increases due to arrivals. Here the central server serves queues with at least $i$ tasks whenever such queues arise, to keep $\blv_i$ at zero. Thus, the total rate of central service tokens dedicated to $\blv_i$ matches exactly the rate of increase of $\blv_i$ due to arrivals.\footnote{\small Technically, the minimization involving $p$ is not necessary: if $\lambda\blv_{i-1}(t)>p$, then $\blv_i(t)$ cannot stay at zero and will immediately increase after $t$. We keep the minimization just to emphasize that the maximum rate of increase in $\blv_i$ due to central service tokens cannot exceed the central service capacity $p$.}
\item $g_i\lt(\blv\rt)=0, \mbox{ if } \blv_i=\blv_{i-1}=0$. Here, both $\blv_i$ and $\blv_{i-1}$ are zero and there are no queues with $i-1$ or more tasks. Hence there is no positive rate of increase in $\blv_i$ due to arrivals. Accordingly, the rate at which central service tokens are used to serve stations with at least $i$ tasks is zero.
\eenum
Note that, as mentioned in the introduction, the discontinuities in the fluid model come from the term $g(\blv)$, which reflects the presence of a central server. 
%
%
%

\section {Analysis of the Fluid Model and Exponential Phase Transition}
\lbl{sec:qualresult}
The following theorem characterizes the invariant state for the fluid model. It will be used to demonstrate an \emph{exponential improvement} in the rate of growth of the average queue length as $\lambda \rar 1$ (Corollary \ref{cor:logsc}).

\begin{thm}
\lbl{thm:ssprop}
The drift $\buf(\cdot)$ in the fluid model admits a unique invariant state $\blv^I$ (i.e., $\buf(\blv^I)=0$). Letting $\bls^I_i \bydef \blv^I_i-\blv^I_{i+1}$ for all $i \geq 0$, the exact expression for the invariant state as follows:
\benum [(1)]
\item If $p=0$, then $\bls_i^I=\lambda^i, \, \forall i \geq 1$.
\item If $p \geq \lambda$, then $\bls_i^I=0, \, \forall i \geq 1$.
\item If $0<p < \lambda$, and $\lambda = 1-p$, then\footnote{\small Here $\lfloor x \rfloor$ is defined as the largest integer that is less than or equal to $x$.}
\beq
\bls^I_i = \lt\{ \begin{array}{ll}
1-\lt(\frac{p}{1-p}\rt)i, & 1 \leq i \leq \tilde{i}^*\lt(p,\lambda\rt), \\
0, & i > \tilde{i}^*\lt(p,\lambda\rt),
\end{array} \right. \nnb
\eeq
where $\tilde{i}^*\lt(p,\lambda\rt) \bydef \lt\lfloor \frac{1-p}{p} \rt\rfloor$.
\item If $0<p < \lambda$, and $\lambda \neq 1-p$, then
\beq
\bls^I_i = \lt\{ \begin{array}{ll}
\frac{1-\lambda}{1-\lt(p+\lambda\rt)}\lt(\frac{\lambda}{1-p}\rt)^i - \frac{p}{1-\lt(p+\lambda\rt)}, & 1 \leq i \leq i^*\lt(p,\lambda\rt), \\
0, & i > i^*\lt(p,\lambda\rt),
\end{array} \right. \nnb
\eeq
where
\beq
\lbl{eq:sinv23}
i^*\lt(p,\lambda\rt) \bydef \lt\lfloor \log_{\frac{\lambda}{1-p}}{\frac{p}{1-\lambda}}\rt\rfloor,
\eeq
\eenum
\end{thm}
\begin{proof} The proof consists of simple algebra to compute the solution to $\mbf{F}(\blv^I)=0$. The proof is given in Section \ref{sec:pfinvstate}.\end{proof}

Case ($4$) in the above theorem is particularly interesting, as it reflects the system's performance under heavy load ($\lambda$ close to $1$). Note that since $\bls^I_1$ represents the probability of a typical queue having at least $i$ tasks, the quantity 
\beq
\blv^I_1\bydef\sum_{i=1}^\infty \bls^I_i
\eeq
represents the \emph{average queue length}. The following corollary, which characterizes the average queue length in the invariant state for the fluid model, follows from Case ($4$) in Theorem \ref{thm:ssprop} by some straightforward algebra. 

\begin{cor} {\bf (Phase Transition in Average Queue Length Scaling)}
\lbl{cor:logsc}
If $0 < p < \lambda$ and $\lambda\neq 1-p$, then
\beqn
\blv_1^I \bydef \sum_{i=1}^\infty \bls^I_i &=& \frac{\lt(1-p\rt)\lt(1-\lambda\rt)}{\lt(1-p-\lambda\rt)^2}\lt[1-\lt(\frac{\lambda}{1-p}\rt)^{i^*\lt(p,\lambda\rt)}\rt] \nnb \\
&&-\frac{p}{1-p-\lambda}i^*\lt(p,\lambda\rt),
\eeqn
with $i^*\lt(p,\lambda\rt)=\lt\lfloor \log_{\frac{\lambda}{1-p}}{\frac{p}{1-\lambda}}\rt\rfloor$. In particular, this implies that for any fixed $p>0$, $\blv_1^I$ scales as
\beq
\lbl{eq:expoim}
\blv_1^I \sim i^*\lt(p,\lambda\rt) \sim \log_{\frac{1}{1-p}}\frac{1}{1-\lambda},  \quad \mbox{as } \lambda \rar 1.
\eeq
\end{cor}

The scaling of the average queue length in Eq.\ \eqref{eq:expoim} with respect to arrival rate $\lambda$ is contrasted with (and is \emph{exponentially better} than) the familiar $\frac{1}{1-\lambda}$ scaling when no centralized resource is available ($p=0$).

\vspace{20pt}

{\bf Intuition for Exponential Phase Transition}. The exponential improvement in the scaling of $\blv^I_1$ is surprising, because the expressions for $\bls^I_i$ look ordinary and do not contain any super-geometric terms in $i$. However, a closer look reveals that for any $p>0$, the tail probabilities $\bls^I$ have {\bf finite support}: $\bls^I_i$ ``dips'' down to $0$ as $i$ increases to $i^*(p,\lambda)$, which is even faster than a super-geometric decay. Since $0\leq \bls^I_i \leq 1$ for all $i$, it is then intuitive that $\blv^I_1=\sum_{i=1}^{i^*(p,\lambda)} \bls^I_i$ is upper-bounded by $i^*(p,\lambda)$, which scales as $\log_{\frac{1}{1-p}}\frac{1}{1-\lambda}$ as $\lambda \rar 1$. Note that a tail probability with ``finite-support'' implies that the fraction of stations with more than $i^*(p,\lambda)$ tasks \emph{decreases to zero} as $N \rar \ity$. For example, we may have a strictly positive fraction of stations with, say, $10$ tasks, but stations with more than $10$ tasks hardly exist. While this may appear counterintuitive, it is a direct consequence of centralization in the resource allocation schemes. Since a fraction $p$ of the total resource is constantly going after the longest queues, it is able to prevent long queues (i.e., queues with more than $i^*(p,\lambda)$ tasks) from even appearing. The thresholds $i^*(p,\lambda)$ increasing to infinity as $\lambda \rar 1$ reflects the fact that the central server's ability to annihilate long queues is compromised by the heavier traffic loads; our result essentially shows that the increase in $i^*(\lambda,p)$ is surprisingly slow. $\diamonds$

\newpage

{\bf Numerical Results}: Figure \ref{fig:pandnop} compares the invariant state vectors for the case $p=0$ (stars) and $p=0.05$ (diamonds). When $p=0$, $\bls^I_i$ decays exponentially as $\lambda^i$, while when $p=0.05$, $\bls^I_i$ decays much faster, and reaches zero at around $i=40$. Figure \ref{fig:heavytraffic} demonstrates the exponential phase transition in the average queue length as the traffic intensity reaches $1$, where the solid curve, corresponding to a positive $p$, increases significantly slower than the usual $\frac{1}{1-\lambda}$ delay scaling (dotted curve). Simulations show that the theoretical model offers good predictions for even a moderate number of servers ($N=100$). The detailed simulation setup can be found in Appendix B. Table \ref{tb:spsi} gives examples of the values for $i^*(p,\lambda)$; note that these values in some sense correspond to the \emph{maximum delay} an average customer could experience in the system. $\diamonds$

\begin{figure}
\begin{center}
\includegraphics[scale=.36]{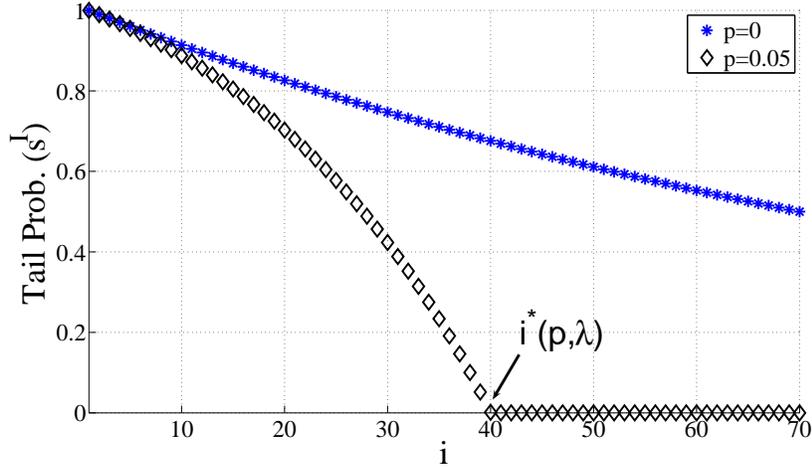}
\caption{Values of $\bls^I_i$, as a function of $i$, for $p=0$ and $p=0.05$, with traffic intensity $\lambda=0.99$.}
\label{fig:pandnop}
\end{center}
\vspace{-0.6cm}
\end{figure}

\begin{figure}[h]
\begin{center}
\includegraphics[scale=0.36]{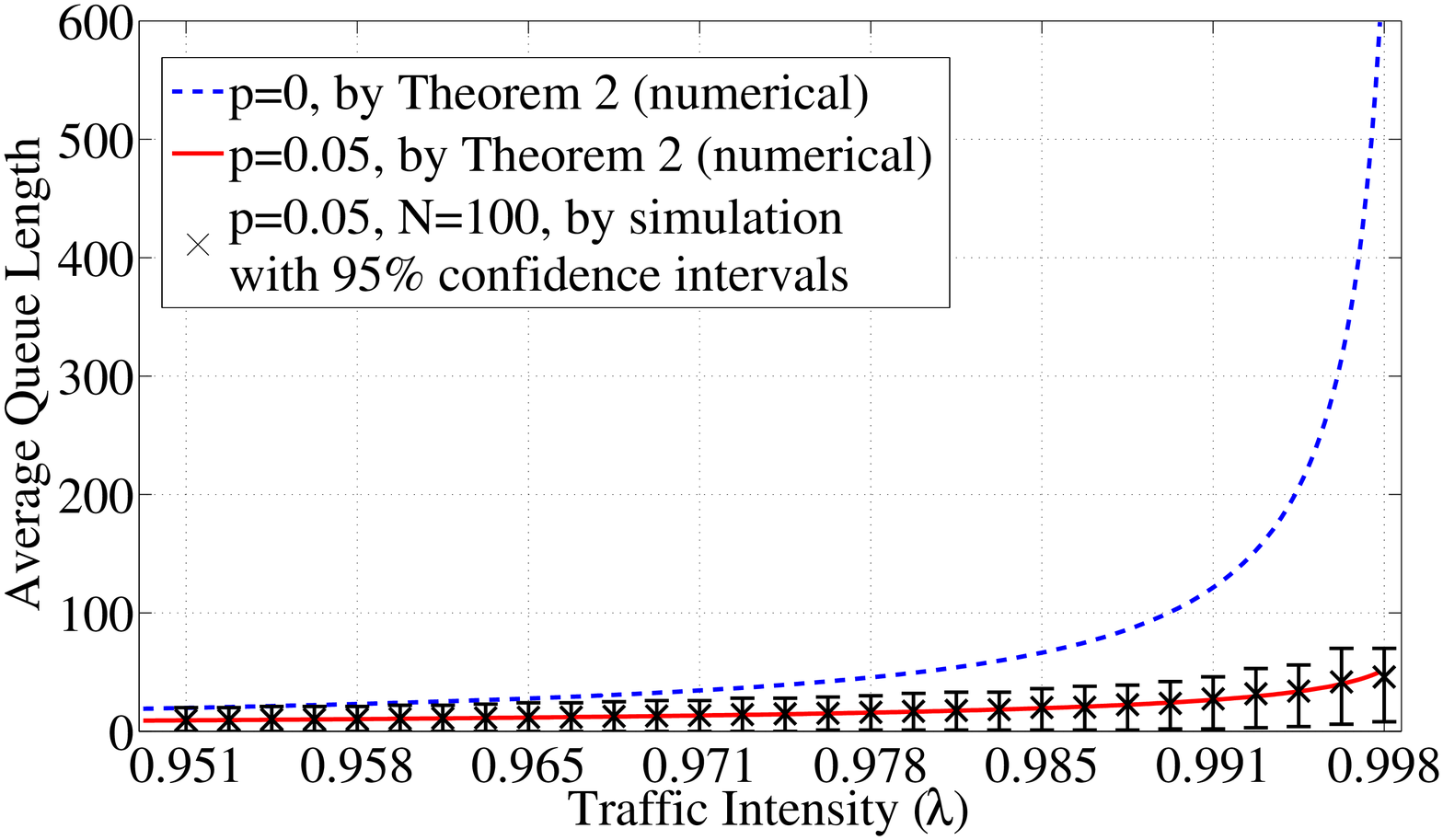}
\caption{Illustration of the exponential improvement in average queue length from $O(\frac{1}{1-\lambda})$ to $O(\log\frac{1}{1-\lambda})$ as $\lambda \rar 1$, when we compare $p=0$ to $p=0.05$.}
\label{fig:heavytraffic}
\end{center}
\end{figure}

\vspace{20pt}

Theorem \ref{thm:ssprop} characterizes the invariant state of the fluid model, without saying if and how a solution of the fluid model reaches it. The next two results state that given any finite initial condition, the solution to the fluid model is unique and converges to the unique invariant state as time goes to infinity.

\begin{thm}{\bf (Uniqueness of Solutions to Fluid Model)}
\lbl{thm:fluidunique}
Given any initial condition $\blv^0 \in \overline{\spv}^\infty$, the fluid model has a unique solution $\blv(\blv^0,t)$, $t \in [0,\ity)$.
\end{thm}
\begin{proof} See Section \ref{sec:unifl}. \end{proof}

\begin{thm}{\bf (Global Stability of Fluid Solutions)}
\lbl{thm:fluidconv}
Given any initial condition $\blv^0 \in \overline{\spv}^\infty$, and with $\blv(\blv^0,t)$ the unique solution to the fluid model, we have
\beq
\lim_{t \rar \infty} \lt\|\blv\lt(\blv^0,t\rt) - \blv^I \rt\|_w=0,
\eeq
where $\blv^I$ is the unique invariant state of the fluid model given in Theorem \ref{thm:ssprop}.
\end{thm}
\begin{proof} See Section \ref{sec:fldconv}. \end{proof}
\begin{table}
\begin{center}\small
\begin{tabular}{| c | c c c c c |}
\hline
$p= \backslash \; \lambda=$ & 0.1 & 0.6& 0.9 & 0.99 & 0.999 \\ \hline
0.002 &2	&10&	37&	199	&692\\
0.02 &1	&6&	18&	68&	156\\
0.2 &0	&2&	5	&14	&23\\
0.5 &0	&1&	2	&5&	8\\
0.8 &0	&0	&1&	2	&4\\ 
\hline
\end{tabular}\normalsize
\caption{Values of $i^*(p,\lambda)$ for various combinations of $(p,\lambda)$.}
\lbl{tb:spsi}
\end{center}
\vspace{-0.7cm}
\end{table}

\section {Convergence to a Fluid Solution - Finite Horizon and Steady State}
\lbl{sec:summaryconvg}
The two theorems in this section justify the use of the fluid model as an approximation for the finite stochastic system. The first theorem states that as $N \rar \ity$ and with high probability, the evolution of the aggregated queue length process $\buv(t)$ is uniformly close, over any finite time horizon $[0,T]$, to the unique solution of the fluid model.

\begin{thm}{\bf (Convergence to Fluid Solutions over a Finite Horizon)}
\lbl{thm:transconv} 
Consider a sequence of systems as the number of servers $N$ increases to infinity. Fix any $T >0$. If for some $\blv^0 \in \overline{\spv}^\infty$,
\beq
\lim_{N \rar \infty} \pb\lt(\|\buv \lt(0\rt) - \blv^0 \|_w > \gm \rt) = 0, \quad \forall \gm >0,
\eeq
then 
\beq
\lbl{eq:probconv}
\lim_{N \rar \infty} \pb\lt(\sup_{t \in [0,T]} \|\buv \lt(t\rt) - \blv\lt(\blv^0,t\rt)\|_w > \gm \rt) = 0, \quad \forall \gm >0.
\eeq
where $\blv\lt(\blv^0,t\rt)$ is the unique solution to the fluid model given initial condition $\blv^0$.
\end{thm}
\begin{proof} See Section \ref{sec:thmtransconvpf}. \end{proof}

Note that if we combine Theorem \ref{thm:transconv} with the convergence of $\blv(t)$ to $\blv^I$ in Theorem \ref{thm:fluidconv}, we see that the finite system ($\buv$) is approximated by the invariant state of the fluid model $\blv^I$ after a fixed time period. In other words, we now have
\beq
\lbl{eq:limorder1}
\lim_{t \rar \ity} \lim_{N \rar \ity }\buv(t) = \blv^I, \mbox{ in distribution.}
\eeq
If we switch the order in which the limits over $t$ and $N$ are taken in Eq.\ \eqref{eq:limorder1}, we are then dealing with the limiting behavior of the \emph{sequence of steady-state distributions} (if they exist) as the system size grows large. Indeed, in practice it is often of great interest to obtain a performance guarantee for the steady state of the system, if it were to run for a long period of time. In light of Eq.\ \eqref{eq:limorder1}, we may expect that
\beq
\lim_{N \rar \ity} \lim_{t \rar \ity }\buv(t) {=} \blv^I, \mbox{ in distribution.}
\eeq
The following theorem shows that this is indeed the case, i.e., that a unique steady-state distribution of $\blv^N(t)$ (denoted by $\pi^N$) exists for all $N$, and that the sequence $\pi^N$ concentrates on the invariant state of the fluid model ($\blv^I$) as $N$ grows large. 
\begin{thm}{\bf(Convergence of Steady-state Distributions to $\blv^I$)}
\lbl{thm:convss}
Denote by $\mcal{F}_{\overline{\spv}^\ity}$ the $\sigma$-algebra generated by $\overline{\spv}^\ity$. For any $N$, the process $\buv(t)$ is positive recurrent, and it admits a unique steady-state distribution $\pi^N$. Moreover,
\beq
\lim_{N\rar\ity} \pi^N  = \delta_{\blv^I}, \, \mbox{ in distribution,}
\eeq
where $\delta_{\blv^I}$ is a probability measure on $\mcal{F}_{\overline{\spv}^\ity}$ that is concentrated on $\blv^I$, i.e., for all $X \in \mcal{F}_{\overline{\spv}^\ity}$,
\beq
\delta_{\blv^I}(X)  = \lt\{ \begin{array}{ll}
1, & \blv^I \in X , \\
0, & \mbox{otherwise}. 
\end{array} \right. \nnb
\eeq
\end{thm}
\begin{proof} The proof is based on the tightness of the sequence of steady-state distributions $\pi^N$, and a uniform rate of convergence of $\buv(t)$ to $\blv(t)$ over any compact set of initial conditions. The proof is given in Chapter \ref{sec:conss}. \end{proof}

\begin{figure}[h]
\centering
\includegraphics[scale=1.1]{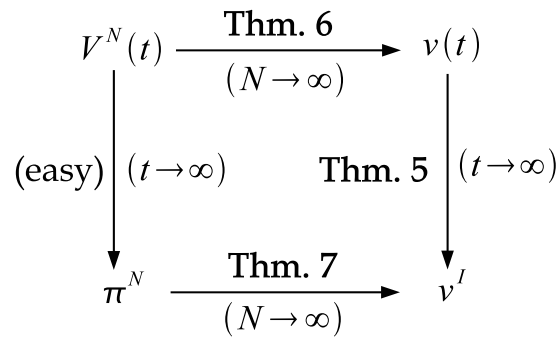}
\caption{Relationships between convergence results.}
\label{fig:techres}
\vspace{-0.1cm}
\end{figure}

Figure \ref{fig:techres} summarizes the relationships between the convergence to the solution of the fluid model over a finite time horizon (Theorem \ref{thm:fluidconv} and Theorem \ref{thm:transconv}) and the convergence of the sequence of steady-state distributions (Theorem \ref{thm:convss}).

\chapter{Probability Space and Coupling}

Starting from this chapter, the remainder of the thesis will be devoted to proving the results summarized in Chapter \ref{sec:sumres}. We begin by giving an outline of the main proof techniques, as well as the relationships among them, in Section \ref{sec:proofstrat}. The remainder of the current chapter focuses on constructing probability spaces and appropriate couplings of stochastic sample paths, which will serve as the foundation for later analysis.

\section{Overview of Technical Approach}
\lbl{sec:proofstrat}
We begin by coupling the sample paths of processes of interest (e.g., $\buv(\cdot)$) with those of two fundamental processes that drive the system dynamics (Section \ref{sec:probsp}). This approach allows us to link deterministically the convergence properties of the sample paths of interest to those of the fundamental processes, on which probabilistic arguments are easier to apply (such as the Functional Law of Large Numbers). Using this coupling framework, we show in Chapter \ref{sec:flim} that almost all sample paths of $\buv(\cdot)$ are ``tight'' in the sense that, as $N\rar\ity$, they are uniformly approximated by a set of Lipschitz-continuous trajectories, which we refer to as the fluid limits, and that all such fluid limits are valid solutions to the fluid model. This result connects the finite stochastic system with the deterministic fluid solutions. Chapter \ref{sec:fluidmodel} studies the properties of the fluid model, and provides proofs for Theorem \ref{thm:fluidunique} and \ref{thm:fluidconv}. Note that Theorem \ref{thm:transconv} (convergence of $\buv(\cdot)$ to the unique fluid solution, over a finite time horizon) now follows from the tightness results in Chapter \ref{sec:flim} and the uniqueness of fluid solutions (Theorem \ref{thm:fluidunique}). The proof of Theorem \ref{thm:ssprop} stands alone, and will be given in Section \ref{sec:pfinvstate}. Finally, the proof of Theorem \ref{thm:convss} (convergence of steady state distributions to $\blv^I$) is given in Chapter \ref{sec:conss}.

The goal of the current chapter is to formally define the probability spaces and stochastic processes with which we will be working in the rest of the thesis. Specifically, we begin by introducing two \emph{fundamental processes}, from which all other processes of interest (e.g., $\buv(\cdot)$) can be derived on a per sample path basis. 

\section{Definition of Probability Space}
\lbl{sec:probsp}

\begin{defn}{\bf (Fundamental Processes and Initial Conditions)}
\benum [(1)]
\item {\bf The Total Event Process}, $\lt\{W(t)\rt\}_{t \geq 0}$, defined on a probability space $(\Omega_W,\mcal{F}_W,\pb_W)$, is a Poisson process with rate $\lambda + 1$, where each jump marks the \emph{time} when an ``event'' takes place in the system.

\item {\bf The Selection Process}, $\lt\{U(n)\rt\}_{n \in \zp}$, defined on a probability space $(\Omega_U,\mcal{F}_U,\pb_U)$, is a discrete-time process, where each $U(n)$ is independent and uniformly distributed in $[0,1]$. This process, along with the current system state, determines the \emph{type} of each event (i.e., whether it is an arrival, a local token generation, or a central token generation). 

\item {\bf The (Finite) Initial Conditions}, $\{\mbf{V}^{(0,N)}\}_{N \in \N}$, is a sequence of random variables defined on a common probability space  $(\Omega_0,\mcal{F}_0,\pb_0)$, with $\mbf{V}^{(0,N)}$ taking values\footnote{For a finite system of $N$ stations, the measure induced by $\buv_i(t)$ is discrete and takes positive values only in the set of rational numbers with denominator $N$.} in $\overline{\spv}^\ity\cap\spq^N$. Here, $\mbf{V}^{(0,N)}$ represents the initial queue length distribution. 
\eenum
\end{defn}

For the rest of the thesis, we will be working with the product space
\beq
(\Omega,\mcal{F},\pb) \bydef (\Omega_W \times \Omega_U \times \Omega_0,\mcal{F}_W \times \mcal{F}_U \times \mcal{F}_0,\pb_W \times \pb_U \times \pb_W).
\eeq
With a slight abuse of notation, we use the same symbols $W(t)$, $U(n)$ and $\mbf{V}^{(0,N)}$ for their corresponding \emph{extensions} on $\Omega$, i.e. $W(\omega,t)\bydef W(\omega_W,t)$,
where $\omega\in\Omega$ and $\omega = (\omega_W,\omega_U,\omega_0)$. The same holds for $U$ and $\mbf{V}^{(0,N)}$.

\section{A Coupled Construction of Sample Paths}
\lbl{sec:spathcons}
Recall the interpretation of the fluid model drift terms in Section \ref{sec:fmodel}. Mimicking the expression of $\dot{\blv}_i(t)$ in Eq.\ \eqref{eq:dft}, we would like to decompose $\buv_i(t)$ into three non-decreasing right-continuous processes,
\beq
\lbl{eq:decomp}
\buv_i(t)=\buv_i(0)+\bua_i(t)-\bul_i(t)-\buc_i(t), \quad i \geq 1,
\eeq
so that $\bua_i(t)$, $\bul_i(t)$, and $\buc_i(t)$ correspond to the \emph{cumulative changes} in $\buv_i$ due to arrivals, local service tokens, and central service tokens, respectively. We will define processes $\bua(t), \bul(t)$, $\buc(t)$, and $\buv(t)$ on the common probability space $(\Omega,\mcal{F},\pb)$, and \emph{couple} them with the sample paths of the fundamental processes $W(t)$ and $U(n)$, and the value of $\mbf{V}^{(0,N)}$, for each sample $\omega \in \Omega$. First, note that since the $N$-station system has $N$ independent Poisson arrival streams, each with rate $\lambda$, and an exponential server with rate $N$, the total event process for this system is a Poisson process with rate $N(1+\lambda)$. Hence, we define $W^N(\omega, t)$, the $N$th \emph{normalized event process}, as 
\beq
W^N(\omega, t) \bydef \frac{1}{N}W(\omega, Nt), \quad \forall t \geq 0, \omega \in \Omega.
\eeq
Note that $W^N(\omega,t)$ is normalized so that all of its jumps have a magnitude of $\frac{1}{N}$.

\vspace{6pt}

\begin{figure}[h]
\lbl{fig:unitpart}
\centering
\includegraphics[scale=1.32]{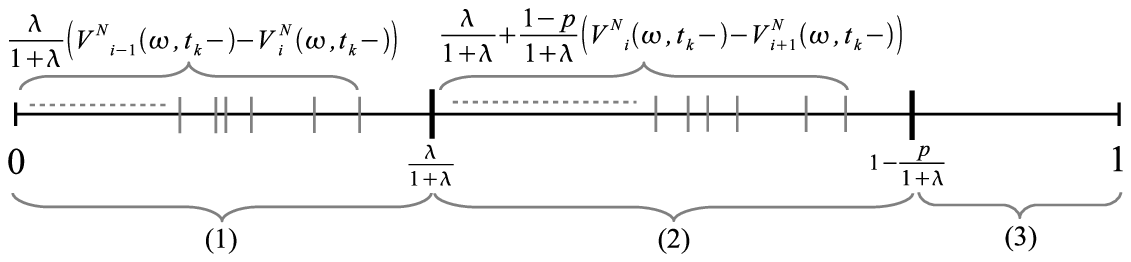}

\vspace{-20pt}

\caption{Illustration of the partition of $\lt[0,1\rt]$ for constructing $\buv(\omega,\cdot)$.}
\end{figure}

\vspace{4pt}

The coupled construction is intuitive: whenever there is a jump in $W^N(\omega,\cdot)$, we decide the type of event by looking at the value of the corresponding selection variable $U(\omega,n)$ and the current state of the system $\buv(\omega,t)$. Fix $\omega$ in $\Omega$, and let $t_k, k\geq1$, denote the time of the $k$th jump in $W^N(\omega,\cdot)$. 

We first set all of $\bua$, $\bul$, and $\buc$ to zero for $t \in [0,t_1)$. Starting from $k=1$, repeat the following steps to for increasing values of $k$. 
The partition of the interval $[0,1]$ used in the procedure is illustrated in Figure \ref{fig:unitpart}.

\benum [(1)]
\item If $U(\omega,k) \in \frac{\lambda}{1+\lambda} \lt[0, \buv_{i-1}(\omega,t_k-)-\buv_i(\omega,t_k-)\rt)$ for some $i \geq 1$, the event corresponds to an {\bf arrival} to a station with at least $i-1$ tasks. Hence we increase $\bua_i(\omega,t)$ by $\frac{1}{N}$ at all such $i$. 
\item If $U(\omega,k) \in \frac{\lambda}{1+\lambda}+\frac{1-p}{1+\lambda} \lt[0, \buv_{i}(\omega,t_k-)-\buv_{i+1}(\omega,t_k-)\rt)$ for some $i \geq 1$, the event corresponds to the {\bf completion} of a task at a station with at least $i$ tasks due to a {\bf local service token}. We increase $\bul_i(\omega,t)$ by $\frac{1}{N}$ at all such $i$. Note that $i=0$ is \emph{not} included here, reflecting the fact that if a local service token is generated at an empty station, it is immediately wasted and has no impact on the system.

\item Finally, if $U(\omega,k) \in \frac{\lambda}{1+\lambda}+\frac{1-p}{1+\lambda}+\lt[0 , \frac{p}{1+\lambda} \rt)= \lt[1-\frac{p}{1+\lambda}, 1\rt)$, the event corresponds to the generation of a {\bf central service token}. Since the central service token is alway sent to a station with the longest queue length, we will have a task completion in a most-loaded station, unless the system is empty. Let $i^*(t)$  be the last positive coordinate of $\buv(\omega,t-)$, i.e., $i^*(t) =\sup\{i:\buv_i(\omega,t-)>0\}$. We increase $\buc_j(\omega,t)$ by $\frac{1}{N}$ for all $j$ such that $1 \leq j \leq i^*(t_k)$.
\eenum

To finish, we set $\buv(\omega,t)$ according to Eq.~\eqref{eq:decomp}, and keep the values of all processes unchanged between $t_{k}$ and $t_{k+1}$. We set $\buv_0 \bydef \buv_1 + 1$, so as to stay consistent with the definition of $\buv_0$.

\chapter{Fluid Limits of Stochastic Sample Paths}
\lbl{sec:flim}

In this chapter, we formally establish the connections between the stochastic sample paths ($\buv(\cdot)$) and the solutions to the fluid model ($\blv(\cdot)$). Through two important technical results (Propositions \ref{prop:smooth} and \ref{prop:drifts}), it is shown that, as $N \rar \ity$ and almost surely, any subsequence of $\lt\{\buv(\cdot)\rt\}_{N \geq 1}$ contains a subsequence that convergences uniformly to \emph{a solution} of the fluid model, over any finite horizon $[0,T]$. This provides strong justification for using the fluid model as an approximation for the stochastic system over a finite time period. However, we note that results presented in this chapter do not imply the converse, that \emph{any} solution to the fluid model corresponds to a limit point of some sequence of stochastic sample paths. This issue will be resolved in the next chapter where we show the important property of the uniqueness of fluid solutions. The results presented in chapter, together with the uniqueness of fluid solutions, will then have established that the fluid model \emph{fully characterizes} the \emph{transient behavior} of sufficiently large finite stochastic systems over a finite time horizon $[0,T]$.

\vspace{15pt}
In the sample-path wise construction in Section \ref{sec:spathcons}, all randomness is attributed to the initial condition $\mbf{V}^{(0,N)}$ and the two fundamental processes $W\lt(\cdot\rt)$ and $U\lt(\cdot\rt)$. Everything else, including the system state $\buv(\cdot)$ that we are interested in, can be derived from a deterministic mapping, given a particular realization of $\mbf{V}^{(0,N)}$, $W(\cdot)$, and $U(\cdot)$. With this in mind, the approach we will take to prove convergence to a fluid limit (i.e., a limit point of $\lt\{\buv(\cdot)\rt\}_{N \geq 1}$), over a finite time interval $[0,T]$, can be summarized as follows.

\benum [(1)]
\item (Lemma \ref{lm:nice1}) Define a subset $\spc$ of the sample space $\Omega$, such that $\pb\lt(\spc\rt)=1$ and the sample paths of $W$ and $U$ are sufficiently ``nice'' for every $\omega \in \spc$. 
\item (Proposition \ref{prop:smooth}) Show that for all $\omega$ in this nice set, the derived sample paths $\buv(\cdot)$ are also ``nice'', and contain a subsequence converging to a Lipschitz-continuous trajectory $\blv(\cdot)$, as $N \rar \infty$. 
\item (Proposition \ref{prop:drifts}) Characterize the derivative at any regular point\footnote{Regular points are points where derivative exists along \emph{all coordinates} of the trajectory. Since the trajectory is Lipschitz-continuous along every coordinate, almost all points are regular.} of $\blv(\cdot)$ and show that it is identical to the drift in the fluid model. Hence $\blv(\cdot)$ is a solution to the fluid model.
\eenum

The proofs will be presented according to the above order.

\section{Tightness of Sample Paths over a Nice Set}

We begin by proving the following lemma which characterizes a ``nice'' set $\spc \subset \Omega$ whose elements have desirable convergence properties.
\begin{lemma}
\label{lm:nice1}
Fix $T>0$. There exists a measurable set $\spc \subset \Omega$ such that $\pb\lt(\spc\rt)=1$ and for all $\omega \in \spc$, 
\beqn
\lbl{eq:nice1}
\hspace{-26pt} &&\lim_{N\rar\infty}\sup_{t \in [0,T]}\lt|W^N\lt(\omega,t\rt) - \lt(1+\lambda\rt)t\rt|=0, \\
\lbl{eq:nice2}
\hspace{-26pt} &&\lim_{N\rar\infty}\frac{1}{N} \sum_{i=1}^N \mb{I}_{[a,b)}\lt(U\lt(\omega,i\rt)\rt)= b-a, \, \, \mbox{ if } a<b \mbox{ and } [a,b) \subset [0,1].
\eeqn
\end{lemma}

\begin{proof} 
Based on the Functional Law of Large Numbers for Poisson processes, we can find $\mcal{C}_W \subset \Omega_W$, with $\pb_W\lt(\mathcal{C}_W\rt)=1$, over which Eq.~\eqref{eq:nice1} holds. For Eq.~\eqref{eq:nice2}, we invoke the Glivenko-Cantelli theorem, which states that the empirical measures of a sequence of i.i.d. random variables converge \emph{uniformly} almost surely, i.e.,
\beq
\lbl{eq:gcthm}
\lim_{N\rar\infty} \sup_{x \in [0,1]}\lt|\frac{1}{N}\sum_{i=1}^N \mb{I}_{[0,x)}\lt(U\lt(i\rt)\rt) - x\rt| = 0, \quad  \mbox{almost surely}.
\eeq
This implies the existence of some $\mathcal{C}_U\subset\Omega_U$, with $\pb_U\lt(\mathcal{C}_U\rt)=1$, over which Eq.~\eqref{eq:nice2} holds. (This is stronger than the ordinary Strong Law of Large Numbers for i.i.d. uniform random variables on $[0,1]$, which states convergence for a \emph{fixed} set $[0,x)$.) We finish the proof by taking $\spc = \mcal{C}_W\times\mcal{C}_U\times\Omega_0$.  $\ftomb$
\end{proof}

\begin{defn}
We call the 4-tuple, $\bux \bydef \lt( \buv, \bua, \bul, \buc \rt)$, the {\bf $N$th system}. Note that all four components are infinite-dimensional processes.\footnote{If necessary, $\bux$ can be enumerated by writing it explicitly as $\bux = \lt(\buv_0, \bua_0, \bul_0, \buc_0, \buv_1, \bua_1, \ldots \rt).$}
\end{defn}

Consider the space of functions from $[0,T]$ to $\R$ that are right-continuous-with-left-limits (RCLL), denoted by $D[0,T]$, and let it be equipped with the uniform metric, $d\lt(\cdot,\cdot\rt)$:
\beq
d\lt(x,y\rt) \bydef \sup_{t \in [0,T]}\lt|x\lt(t\rt)-y\lt(t\rt)\rt|,  \quad x,y \in D[0,T].
\eeq
Denote by $D^\ity[0,T]$ the set of functions from $[0,T]$ to $\R^\zp$ that are RCLL on every coordinate. Let $d^\zp(\cdot,\cdot)$ denote the uniform metric on $D^\ity[0,T]$:
\beq
\lbl{eq:dmetric}
d^\zp\lt(\blx,\bly\rt) \bydef \sup_{t \in [0,T]}\lt\|\blx\lt(t\rt)-\bly\lt(t\rt)\rt\|_w,  \quad \blx,\bly \in D^\zp[0,T],
\eeq
with $\|\cdot\|_w$ defined in Eq.~\eqref{eq:infnorm}.

The following proposition is the main result of this section. It shows that for sufficiently large $N$, the sample paths are sufficiently close to some absolutely continuous trajectory.

\begin{prop}
\lbl{prop:smooth}
Fix $T >0$. Assume that there exists some $\blv^0 \in \overline{\spv}^\infty$ such that  
\beq
\lim_{N\rar\infty} \|\buv\lt(\omega,0\rt) - \blv^0\|_w=0,
\eeq
for all $\omega \in \spc$. Then for all $\omega \in \spc$, any subsequence of $\lt\{\bux\lt(\omega,\cdot\rt)\rt\}$ contains a further subsequence, $\lt\{\mbf{X}^{N_i}\lt(\omega,\cdot\rt)\rt\}$, that converges to some coordinate-wise Lipschitz-continuous function $\blx\lt(t\rt)=\lt(\blv\lt(t\rt),\bla\lt(t\rt),\bll\lt(t\rt),\blc\lt(t\rt)\rt)$, with $\blv\lt(0\rt)=\blv^0$, $\bla(0)=\bll(0)=\blc(0)=0$ and 
\beq
\lt|\blx_i\lt(a\rt)-\blx_i\lt(b\rt)\rt| \leq L|a-b|, \quad \forall a,b \in [0,T], \quad i \in \zp,
\eeq
where $L>0$ is a universal constant, independent of the choice of $\omega$, $\blx$ and $T$. Here the convergence refers to $d^\zp(\mbf{V}^{N_i},\blv)$, $d^\zp(\mbf{A}^{N_i},\bla)$, $d^\zp(\mbf{L}^{N_i},\bll)$, and $d^\zp(\mbf{C}^{N_i},\blc)$ all converging to $0$, as $i \rar \ity$. 
\end{prop}

For the rest of the thesis, we will refer to such a limit point $\blx$, or any subset of its coordinates, as a {\bf fluid limit}.

{\bf Proof outline:} Here we lay out the main steps of the proof; interested readers are referred to Appendix \ref{app:smoothproof} for a complete proof. 

We first show that for all $\omega \in \spc$, and for every coordinate $i$, any subsequence of $\lt\{X^{N}_i\lt(\omega,\cdot\rt)\rt\}$ has a convergent subsequence with a Lipschitz-continuous limit. We then use the coordinate-wise limit to construct a limit point in the space $D^\zp$. To establish  coordinate-wise convergence, we use a tightness technique previously used in the literature of multiclass queueing networks (see, e.g., \cite{BRS98}). A key realization in this case, is that the total number of jumps in any derived process $\bua$, $\bul$, and $\buc$ \emph{cannot} exceed that of the event process $W^N(t)$ for any particular sample path. Since $\bua$, $\bul$, and $\buc$ are non-decreasing, we expect their sample paths to be ``smooth'' for large $N$, due to the fact that the sample path of $W^N(t)$ does become ``smooth'' for large $N$, for all $\omega\in\spc$ (Lemma \ref{lm:nice1}). More formally, it can be shown that for all $\omega \in \spc$ and $T>0$, there exist diminishing positive sequences $M_N\downarrow0$ and $\gamma_N\downarrow0$, such that the sample path along any coordinate of $\bux$ is $\gamma_N$-approximately-Lipschitz continuous with a uniformly bounded initial condition, i.e., for all $i$,
\beqn
&\lt|\bux_i(\omega,0)-x_i^0\rt|\leq M_N,  \,\nnb \\
& \mbox{and }  \lt|\bux_i(\omega,a)-\bux_i(\omega,b)\rt| \leq L|a-b| + \gamma_N, \quad \forall a,b \in [0,T], \nnb
\eeqn
where  $L$ is the Lipschitz constant, and $T < \ity$ is a fixed time horizon. Using a linear interpolation argument, we then show that sample paths of the above form can be uniformly approximated by a set of $L$-Lipschitz-continuous function on $[0,T]$. We finish by using the Arzela-Ascoli theorem (sequential compactness) along with closedness of this set, to establish the existence of a convergent further subsequence along any subsequence (compactness) and that any limit point must also $L$-Lipschitz-continuous (closedness). This completes the proof for coordinate-wise convergence. 

With the coordinate-wise limit points, we then use a diagonal argument to construct the limit points of $\bux$ in the space $D^\zp[0,T]$. Let $\blv_1$ be any $L$-Lipschitz-continuous limit point of $\buv_1$, so that a subsequence $\mbf{V}^{N^1_j}_1(\omega,\cdot) \rar \blv_1$, as $j \rar \ity$, with respect to $d(\cdot,\cdot)$. Then, we proceed recursively by letting $\blv_{i+1}$ be a limit point of a subsequence of $\lt\{\mbf{V}^{N^{i}_j}_{i+1}(\omega,\cdot)\rt\}_{j=1}^\infty$, where $\{N^i_j\}_{j=1}^\infty$ are the indices for the $i$th subsequence. We claim that $\blv$ is indeed a limit point of $\buv$ under the norm $d^\zp(\cdot,\cdot)$. Note that since $\blv_1(0) = \blv_1^0$, $0 \leq \buv_i(t) \leq \buv_1(t)$, and $\blv_1(\cdot)$ is $L$-Lipschitz-continuous, we have that
\beq
\lbl{eq:bdv}
\sup_{t \in [0,T]} \lt|\blv_i(t)\rt| \leq \sup_{t \in [0,T]} \lt|\blv_1(t)\rt| \leq \lt|\blv^0_1\rt| + LT, \quad  \forall i\in \zp.
\eeq
Set $N_1=1$, and let, for $k\geq2$,
\beq
\lbl{eq:defNk}
N_k = \min\lt\{N \geq N_{k-1}: \sup_{1\leq i \leq k} d(\mbf{V}^{N}_i(\omega,\cdot),\blv_i) \leq \frac{1}{k}\rt\}.
\eeq
Note that the construction of $\blv$ implies that $N_k$ is well defined and finite for all $k$. From Eqs.~\eqref{eq:bdv} and \eqref{eq:defNk}, we have, for all $k \geq 2$, 
\beqn
d^\zp\lt(\mbf{V}^{N_k}(\omega,\cdot), \blv\rt) &=& \sup_{t\in[0,T]}\sqrt{\sum_{i=0}^\infty\frac{\lt|\mbf{V}^{N_k}_i(\omega,t)-\blv_i(t)\rt|^2}{2^i}} \nnb \\
&\leq& \frac{1}{k}+\sqrt{\lt(\lt|\blv^0_1\rt| + LT\rt)^2\sum_{i=k+1}^\infty\frac{1}{2^i}} \nnb\\
&=& \frac{1}{k}+\frac{1}{2^{k/2}} \lt(\lt|\blv^0_1\rt| + LT\rt).
\eeqn
Hence $d^\zp\lt(\mbf{V}^{N_k}(\omega,\cdot), \blv\rt)\rar 0$, as $k\rar \infty$. The existence of the limit points $\bla(t)$, $\bll(t)$ and $\blc(t)$ can be established by an identical argument. This completes the proof. $\ftomb$

\section{Derivatives of the Fluid Limits}

The previous section established that any sequence of ``good'' sample paths ($\lt\{\bux(\omega,\cdot)\rt\}$ with $\omega \in \mcal{C}$) eventually stays close to some Lipschitz-continuous, and therefore absolutely continuous, trajectory. In this section, we will characterize the derivatives of $\blv(\cdot)$ at all regular (differentiable) points of such limiting trajectories. We will show, as we expect, that they are the same as the drift terms in the fluid model (Definition \ref{def:fl}). This means that all fluid limits of $\buv(\cdot)$ are in fact solutions to the fluid model.

\begin{prop}{\bf (Fluid Limits and Fluid Model)}
\lbl{prop:drifts}
Fix $\omega \in \mcal{C}$ and $T>0$. Let $\blx$ be a limit point of some subsequence of $\bux(\omega,\cdot)$, as in Proposition \ref{prop:smooth}. Let $t$ be a point of differentiability of all coordinates of $\blx$. Then, for all $i \in \N$,
\beqn
\lbl{eq:adr}
\dot{\bla}_i(t) &=& \lambda(\blv_{i-1}-\blv_{i}), \\
\lbl{eq:ldr}
\dot{\bll}_i(t) &=& (1-p)(\blv_{i}-\blv_{i+1}), \\
\lbl{eq:cdr}
\dot{\blc}_i(t) &=& g_i(\blv),
\eeqn
where $g$ was defined in Eq.\ \eqref{eq:gdef0}, with the initial condition $\blv(0)=\blv^0$ and boundary condition $
\blv_0(t)-\blv_1(t) = 1,  \forall t \in [0,T].$ In other words, all fluid limits of $\buv(\cdot)$ are solutions to the fluid model.
\end{prop}

\begin{proof} We fix some $\omega \in \spc$ and for the rest of this proof we will suppress the dependence on $\omega$ in our notation. The existence of Lipschitz-continuous limit points for the given $\omega \in \mcal{C}$ is guaranteed by Proposition \ref{prop:smooth}. 
Let $\lt\{\mbf{X}^{N_k}(\cdot)\rt\}_{k=1}^\infty$ be a convergent subsequence such that $\lim_{k \rar \infty} d^\zp(\mbf{X}^{N_k}(\cdot),\blx) = 0$. 
We now prove each of the three claims (Eqs. \eqref{eq:adr}-\eqref{eq:cdr}) separately, and index $i$ is always fixed unless otherwise stated.

{\bf Claim 1}: $\dot{\bla}_i(t) = \lambda (\blv_{i-1}(t)-\blv_{i}(t))$. Consider the sequence of trajectories $\lt\{\mbf{A}^{N_k}(\cdot)\rt\}_{k=1}^\infty$. By construction, $\mbf{A}_i^{N}(t)$ receives a jump of magnitude $\frac{1}{N}$ at time $t$ if and only if an event happens at time $t$ and the corresponding selection random variable, $U(\cdot)$, falls in the interval $\frac{\lambda}{1+\lambda}\lt[0,\buv_{i-1}(t-)-\buv_{i}(t-)\rt)$. Therefore, we can write:
\beq
\lbl{eq:ANKdiff0}
\mbf{A}^{N_k}_{i}(t+\epsilon)- \mbf{A}^{N_k}_{i}(t) = \frac{1}{N_k}\sum_{j=N_kW^{N_k}(t)}^{N_kW^{N_k}(t+\epsilon)} \mb{I}_{I_j}\big(U(j)\big), 
\eeq
where $I_j \bydef \frac{\lambda}{1+\lambda}\lt[0, \mbf{V}^{N_k}_{i-1}(t^{N_k}_j-)-\mbf{V}^{N_k}_{i}(t^{N_k}_j-)\rt)$ and  $t^N_j$ is defined to be the time of the $j$th jump in $W^N(\cdot)$, i.e.,
\beq
t^N_j\bydef \inf\lt\{s\geq0:W^N(s)\geq \frac{j}{N}\rt\}.
\eeq
Note that by the definition of a fluid limit, we have that 
\beq
\lim_{k\rar\infty} \lt(\mbf{A}^{N_k}_{i}(t+\epsilon)-\mbf{A}^{N_k}_{i}(t)\rt) = \bla_i(t+\epsilon)-\bla_i(t).
\eeq
The following lemma bounds the change in $\bla_i(t)$ on a small time interval.
\begin{lemma} Fix $i$ and $t$. For all sufficiently small $\epsilon>0$
\lbl{lm:adiff}
\beq
\lbl{eq:adiff}
\lt| \bla_i(t+\epsilon)-\bla_i(t) - \epsilon\lambda(\blv_{i-1}(t)-\blv_{i}(t)) \rt| \leq 2\epsilon^2 L
\eeq
\end{lemma}

\begin{proof} 
While the proof involves heavy notation, it is based on the fact that $\omega \in \spc$: using Lemma \ref{lm:nice1}, Eq.\ \eqref{eq:adiff} follows from Eq.\ \eqref{eq:ANKdiff0} by applying the convergence properties of $W^N(t)$ (Eq.\ \eqref{eq:nice1}) and $U(n)$ (Eq.\ \eqref{eq:nice2}).

For the rest of the proof, fix some $\omega \in \mcal{C}$. Also, fix $i\geq 1$, $t>0$, and $\epsilon>0$. Since the limiting function $\blx$ is $L$-Lipschitz-continuous on all coordinates by Proposition \ref{prop:smooth}, there exists a non-increasing sequence $\gm_n\downarrow0$ such that for all $s \in [t,t+\epsilon]$ and all sufficiently large $k$,
\beq
\lbl{eq:inclu}
\mbf{V}^{N_k}_j( s) \in \big[\blv_j(t)-(\epsilon L + \gamma_{N_k}),\blv_j(t)+(\epsilon L + \gamma_{N_k})\big), \quad j \in \{i-1,i,i+1\},
\eeq
The above leads to:\footnote{Here $[x]^+\bydef\max\{0,x\}.$}
\beqn
&\lt[0,\mbf{V}^{N_k}_{i-1}(s)-\mbf{V}^{N_k}_{i}(s)\rt) \supset \big[0, \lt[\blv_{i-1}(t)-\blv_{i}(t)-2(\epsilon L + \gamma_{N_k})\rt]^+\big), \nnb \\
\lbl{eq:intvinclu}
&\mbox{ and } \lt[0,\mbf{V}^{N_k}_{i-1}(s)-\mbf{V}^{N_k}_{i}(s)\rt) \subset \big[0,\blv_{i-1}(t)-\blv_{i}(t)+2(\epsilon L + \gamma_{N_k})\big),
\eeqn
for all sufficiently large $k$.

Define the sequence of set-valued functions $\{\eta^n(t)\}$ as
\beq
\eta^n(t) \bydef \frac{\lambda}{1+\lambda} \lt[0, \blv_{i-1}(t)-\blv_i(t)+2(\epsilon L+\gamma_{n}) \rt).
\eeq
Note that since $\gamma_n \downarrow 0$,
\beq
\lbl{eq:eta1}
\eta^{n}(t) \supset \eta^{n+1}(t) \mbox{ and } \bigcap_{n=1}^{\infty}\eta^n(t)= \frac{\lambda}{1+\lambda}\lt[0, \blv_{i-1}(t)-\blv_{i}(t)+2\epsilon L \rt].
\eeq
We have for all sufficiently large $k$, and any $l$ such that $1 \leq l \leq N_k$,
\beqn
\lbl{eq:ANKdiff}
\mbf{A}^{N_k}_{i}(t+\epsilon)-\mbf{A}^{N_k}_{i}(t) &\leq& \frac{1}{N_k}\sum_{j=N_kW^{N_k}(t)+1}^{N_kW^{N_k}(t+\epsilon)} \mb{I}_{\eta^{N_k}(t)}\lt(U(j)  \rt)  \nnb \\
&\leq& \frac{1}{N_k}\sum_{j=N_kW^{N_k}(t)+1}^{N_kW^{N_k}(t+\epsilon)} \mb{I}_{\eta^{l}(t)}\lt(U(j)  \rt)  \nnb \\
&=& \frac{1}{N_k} \lt( \sum_{j=1}^{N_kW^{N_k}(t+\epsilon)} \mb{I}_{\eta^{l}(t)}\lt(U(j)\rt) - \sum_{j=1}^{N_kW^{N_k}(t)} \mb{I}_{\eta^{l}(t)}\lt(U(j)\rt) \rt) \nnb \\
\eeqn
where the first inequality follows from the second containment in Eq. \eqref{eq:intvinclu}, and the second inequality follows from the monotonicity of $\{\eta^n(t)\}$ in Eq.~\eqref{eq:eta1}.

We would like to show that for all sufficiently small $\epsilon>0$,
\beq
\lbl{eq:adiff01}
\bla_i(t+\epsilon)-\bla_i(t) - \epsilon\lambda(\blv_{i-1}(t)-\blv_{i}(t)) \leq 2\epsilon^2 L
\eeq
To prove the above inequality, we first claim that for any interval $[a,b) \subset [0,1]$, 
\beq
\lbl{eq:core1}
\lim_{N \rar \infty} \frac{1}{N}\sum_{i=1}^{NW^N(t)}\mb{I}_{[a,b)}\lt(U(i)\rt) = (\lambda+1)t(b-a),
\eeq
To see this, rewrite the left-hand side of the equation above as
\beqn
\lbl{eq:ranlln1}
& & \lim_{N \rar \infty} \frac{1}{N}\sum_{i=1}^{NW^N(t)}\mb{I}_{[a,b)}\lt(U(i)\rt) \nnb \\
&=& \lim_{N \rar \infty} (\lambda+1) t \frac{1}{(\lambda+1) Nt}\sum_{i=1}^{(\lambda+1) Nt}\mb{I}_{[a,b)}\lt(U(i)\rt) \nnb \\
& & + \lim_{N \rar \infty} (\lambda+1) t \frac{1}{(\lambda+1) Nt} \lt( \sum_{i=1}^{NW^N(t)}\mb{I}_{[a,b)}\lt(U(i)\rt) - \sum_{i=1}^{(\lambda+1) Nt}\mb{I}_{[a,b)}\lt(U(i)\rt)\rt).
\eeqn
Because the magnitude of the indicator function $\mb{I}\{\cdot\}$ is bounded by $1$, we have
\beqn
&& \lt|\sum_{i=1}^{NW^N(t)}\mb{I}_{[a,b)}\lt(U(i)\rt) - \sum_{i=1}^{(\lambda+1) Nt}\mb{I}_{[a,b)}\lt(U(i)\rt)\rt| \leq N\lt|(\lambda+1)t - W^N(t)\rt|.
\lbl{eq:ranlln2}
\eeqn
Since $\omega \in \mcal{C}$, by Lemma \ref{lm:nice1} we have that
\beqn
\lbl{eq:ranlln3}
&&\lim_{N \rar\infty }\lt|(\lambda+1)t - W^N( t)\rt| = 0, \\
\lbl{eq:ranlln4}
&&\lim_{N \rar \infty} \frac{1}{(\lambda+1) Nt}\sum_{i=1}^{(\lambda+1) Nt}\mb{I}_{[a,b)}\lt(U(i)\rt) = b-a,
\eeqn
for any $t <\ity$. Combining Eqs.~\eqref{eq:ranlln1}$-$\eqref{eq:ranlln4}, we have 
\beqn
\lbl{eq:ranlln5}
& & \lim_{N \rar \infty} \frac{1}{N}\sum_{i=1}^{W^N(t)}\mb{I}_{[a,b)}\lt(U(i)\rt) \nnb \\
&=& (\lambda+1) t \lim_{N \rar \infty} \frac{1}{(\lambda+1) Nt}\sum_{i=1}^{(\lambda+1) Nt}\mb{I}_{[a,b)}\lt(U(i)\rt) + \lim_{N \rar\infty } \frac{1}{(\lambda+1)t} \lt|(\lambda+1)t - W^N(t)\rt|  \nnb \\
&=& (\lambda+1) t (b-a),
\eeqn
which establishes Eq.~\eqref{eq:core1}. By the same argument, Eq.~\eqref{eq:ranlln5} also holds when $t$ is replaced by $t+\epsilon$. Applying this result to Eq.~\eqref{eq:ANKdiff}, we have
\beqn
& &\bla_i(t+\epsilon)-\bla_i(t) \nnb \\
&=& \lim_{k\rar\ity} \lt ( \mbf{A}^{N_k}_{i}(t+\epsilon)-\mbf{A}^{N_k}_{i}(t) \rt) \nnb \\
&\leq& (t+\epsilon-t)(\lambda+1)\frac{\lambda}{\lambda+1}\lt[\blv_i(t)-\blv_{i-1}(t) + 2(\epsilon L + \gamma_l)\rt] \nnb \\
&=& \epsilon\lambda(\blv_{i-1}(t)-\blv_{i}(t)) + \lambda(2\epsilon^2 L + 2\epsilon\gamma_l) \nnb \\
&<& \epsilon\lambda(\blv_{i-1}(t)-\blv_{i}(t)) + 2\epsilon^2 L + 2\epsilon\gamma_l,
\eeqn
for all $l\geq 1$, where the last inequality is due to the fact that $\lambda < 1$. Taking $l \rar \ity$ and using the fact that $\gamma_l \downarrow 0$, we have established Eq.~\eqref{eq:adiff01}.

Similarly, changing the definition of $\eta^n(t)$ to 
\beq
\eta^n(t) \bydef \frac{\lambda}{1+\lambda} \big[ 0, \lt[ \blv_{i-1}(t)-\blv_{i}(t)-2(\epsilon L +\gamma_{n})\rt]^+ \big),
\eeq
we can obtain a similar lower bound 
\beq
\lbl{eq:adiff02}
\bla_i(t+\epsilon)-\bla_i(t) - \epsilon\lambda(\blv_{i-1}(t)-\blv_{i}(t)) \geq -2\epsilon^2 L,
\eeq
which together with Eq.~\eqref{eq:adiff01} proves the claim. Note that if $\blv_i(t)=\blv_{i-1}(t)$, the lower bound trivially holds because $\mbf{A}^{N_k}_{i}(t)$ is a cumulative arrival process and is hence non-decreasing in $t$ by definition. $\ftomb$
\end{proof}

Since by assumption $\bla(\cdot)$ is differentiable at $t$, Claim 1 follows from Lemma \ref{lm:adiff} by noting $\dot{\bla}_i(t) \bydef \lim_{\epsilon \downarrow 0} \frac{\bla_i(t+\epsilon)-\bla_i(t)}{\epsilon}$.

{\bf Claim 2}: $\dot{\bll}_i(t) =  (1-p)(\blv_{i}(t)-\blv_{i+1}(t))$.
Claim 2 can be proved using an identical approach to the one used to prove Claim 1. The proof is hence omitted.

{\bf Claim 3}: $\dot{\blc}_i(t)=g_i\lt(\blv\rt)$. We prove Claim 3 by considering separately the three cases in the definition of $\blv$.

\benum [(1)]
\item {\bf Case 1}:  $\dot{\blc}_i(t) = 0$, if  $\blv_{i-1}=0, \blv_i = 0$. Write 
\beq
\dot{\blc}_i(t)= \dot{\bla}_i(t)-\dot{\bll}_i(t)-\dot{\blv}_i(t).
\eeq
We calculate each of the three terms on the right-hand side of the above equation. By Claim 1, $\dot{\bla}_i(t)=\lambda(\blv_{i-1}-\blv_i)=0$, and by Claim 2, $\dot{\bll}_i(t)=\lambda(\blv_i - \blv_{i+1})=0$. To obtain the value for $\dot{\blv}_i(t)$, we use the following trick: since $\blv_i(t)=0$ \emph{and} $\blv_i$ is \emph{non-negative}, the \emph{only} possibility for $\blv_i(t)$ to be differentiable at $t$ is that  $\dot{\blv}_i(t)=0$. Since $\dot{\bla}_i(t)$, $\dot{\bll}_i(t)$, and $\dot{\blv}_i(t)$ are all zero, we have that $\dot{\blc}_i(t)=0$.

\item  {\bf Case 2}: $\dot{\blc}_i(t) = \min\{\lambda \blv_{i-1},p\}$, if  $\blv_i=0, \blv_{i-1} > 0$.

In this case, the fraction of queues with at least $i$ tasks is zero, hence $\blv_i$ receives no drift from the local portion of the service capacity by Claim 2.  First consider the case $\blv_{i-1}(t)\leq\frac{p}{\lambda}$. Here the line of arguments is similar to the one in Case 1. By Claim 1, $\dot{\bla}_i(t)=\lambda (\blv_{i-1}-\blv_i)=\lambda \blv_{i-1}$, and by Claim 2, $\dot{\bll}_i(t)=\lambda(\blv_i - \blv_{i+1})=0$. Using again the same trick as in Case 1, the non-negativity of $\blv_i$ and the fact that $\blv_i(t)=0$ together imply that we must have $\dot{\blv}_i(t)=0$. Combining the expressions for $\dot{\bla}_i(t)$, $\dot{\bll}_i(t)$, and $\dot{\blv}_i(t)$, we have
\beq
\dot{\blc}_i(t)= -\dot{\blv}_i(t)+\dot{\bla}_i(t)-\dot{\bll}_i(t) = \lambda \blv_{i-1}.
\eeq
Intuitively, here the drift due to random arrivals to queues with $i-1$ tasks, $\lambda \blv_{i-1}$, is ``absorbed'' by the central portion of the service capacity.

If $\blv_{i-1}(t)>\frac{p}{\lambda}$, then the above equation would imply that $\dot{\blc}_i(t)=\lambda\blv_{i-1}(t)>p$, if $\dot{\blc}_i(t)$ exists. But clearly $\dot{\blc}_i(t) \leq p$. This simply means $\blv_i(t)$ cannot be differentiable at time $t$, if $\blv_i(t)=0, \blv_{i-1}(t)>\frac{p}{\lambda}$.  Hence we have the claimed expression. 

\item  {\bf Case 3}: $\dot{\blc}_i(t) = p$, if  $\blv_i>0, \blv_{i+1} > 0$.

Since there is a positive fraction of queues with more than $i$ tasks, it follows that $\buv_i$ is decreased by $\frac{1}{N}$ \emph{whenever} a central token becomes available. Formally, for some small enough $\epsilon$, there exists $K$ such that $
\mbf{V}^{N_k}_i(s)>0$ for all $k \geq K, \, s \in [t,t+\epsilon]$. Given the coupling construction, this implies for all $ k \geq K$, $s \in [t,t+\epsilon]$
\beq
\mbf{V}^{N_k}_i(s)-\mbf{V}^{N_k}_i(t) = \frac{1}{N_k}\sum_{j=N_kW^{N_k}(t)}^{N_kW^{N_k}(s)}\mb{I}_{[1-\frac{p}{1+\lambda},1)}\lt(U(j)\rt). \nnb
\eeq
Using the same arguments as in the proof of Lemma \ref{lm:adiff}, we see that the right-hand side of the above equation converges to $(s-t)p+o(\ep)$ as $k \rar \infty$. Hence,$\dot{\blv}_i(t)=$ $\lim_{\epsilon \downarrow 0}\lim_{k\rar\infty} $ $\frac{\mbf{V}^{N_k}_i(t+\epsilon)-\mbf{V}^{N_k}_i(t)}{\epsilon}$ $=p$.
\eenum
Finally, note that the boundary condition $\blv_0(t)-\blv_1(t)=1$ is a consequence of the fact that $\buv_0(t)-\buv_1(t)\bydef \bus_1(t) = 1$ for all $t$. This concludes the proof of Proposition \ref{prop:drifts}. $\ftomb$
\end{proof}

\chapter{Properties of the Fluid Model}
\lbl{sec:fluidmodel}

In this chapter, we establish several important properties of the fluid model. We begin by proving Theorem \ref{thm:ssprop} in Section \ref{sec:pfinvstate}, which states that the fluid model admits a unique invariant state for each pair of $p$ and $\lambda$. Section \ref{sec:unifl} is devoted to proving the important property that for any initial condition $\blv^0 \in \overline{\spv}^\ity$, the fluid model admits a \emph{unique} solution $\blv(\cdot)$. As a corollary, it is shown that the fluid solution $\blv(\cdot)$ depend continuously on the initial condition $\blv^0$, and this technical result will be important for proving the steady-state convergence theorem in the next chapter. Using the uniqueness result and the results from the last chapter, one of our main approximation theorems, Theorem \ref{thm:transconv}, is proved in Section \ref{sec:thmtransconvpf}, which establishes the convergence of stochastic sample paths to the unique solution of the fluid model over any finite time horizon, with high probability. Finally, in Section \ref{sec:fldconv} we prove that the solutions to the fluid model are \emph{globally stable} (Theorem \ref{thm:fluidconv}), so any solution $\blv(t)$ converges to the unique invariant state $\blv^I$ as $t \rar \ity$. This suggest that the qualitative properties derived from $\blv^I$ serves as a good approximation for the transient behavior of the system. We note that by the end of this chapter, we will have establish all \emph{transient approximation results}, which correspond to the path
\beq
\buv(t) \stackrel{N \rar \ity}{\longrightarrow} \blv(t) \stackrel{t \rar \ity}{\longrightarrow} \blv^I,
\eeq
as was illustrated in Figure \ref{fig:techres} of Chapter \ref{sec:sumres}. The other path in Figure \ref{fig:techres}, namely, the approximation of the steady-state distributions of $\buv(\cdot)$ by $\blv^I$, will be dealt with in the next chapter. 

\section{Invariant State of the Fluid Model}
\lbl{sec:pfinvstate}
In this section we prove Theorem \ref{thm:ssprop}, which gives explicit expressions for the (unique) invariant state of the fluid model.

\begin{proof} {\bf (Theorem \ref{thm:ssprop})} In this proof we will be working with both $\blv^I$ and $\bls^I$, with the understanding that $\bls^I_i\bydef \blv^I_i-\blv^I_{i+1}, \forall i \geq 0$. It can be verified that the expressions given in all three cases are valid invariant states, by checking that $\mbf{F}(\blv^I)=0$. We show they are indeed unique.

First, note that if $p \geq \lambda$, then $\mbf{F}_1(\blv) < 0$ whenever $\blv_1 >0$. Since $\blv_1^I\geq 0$, we must have $\blv_1^I=0$, which by the boundary conditions implies that all other $\blv_i^I$ must also be zero. This proves case ($2$) of the theorem. 

Now suppose that $0< p < \lambda$. We will prove case ($4$). We observe that $\mbf{F}_1(\blv)>0$ whenever $\blv_1=0$. Hence $\blv^I_1$ must be positive. By Eq.\ \eqref{eq:gdef0} this implies that $g_1(\blv^I)=p$. Substituting $g_1(\blv^I)$ in Eq.\ \eqref{eq:dft}, along with the boundary condition $\blv^I_0-\blv^I_1=\bls^I_0=1$, we have
\beq
0 = \lambda\cdot1 - (1-p)\bls^I_1-p,
\eeq
which yields
\beq
\bls^I_1 = \frac{\lambda-p}{1-p}.
\eeq
Repeating the same argument, we obtain the recursion
\beq
\bls^I_i = \frac{\lambda \bls^I_{i-1}-p}{1-p},
\eeq
for as long as $\bls^I_i$ (and therefore, $\blv_i^I$) remains positive. Combining this with the expression for $\bls^I_1$, we have
\beq
\bls^I_i=\frac{1-\lambda}{1-\lt(p+\lambda\rt)}\lt(\frac{\lambda}{1-p}\rt)^i - \frac{p}{1-\lt(p+\lambda\rt)},  \quad 1 \leq i \leq i^*\lt(p,\lambda\rt),
\eeq
where $i^*\lt(p,\lambda\rt) \bydef \lt\lfloor \log_{\frac{\lambda}{1-p}}{\frac{p}{1-\lambda}}\rt\rfloor$ marks the last coordinate where $\bls^I_i$ remains non-negative. This proves uniqueness of $\bls^I_i$ up to $i\leq i^*\lt(p,\lambda\rt)$. We can then use the same argument as in case ($2$), to show that $\bls_i^I$ must be equal to zero for all $i > i^*\lt(p,\lambda\rt)$. Cases $(1)$ and $(3)$ can be established using similar arguments as those used in proving case $(4)$. This completes the proof. $\ftomb$
\end{proof}

\section {Uniqueness of Fluid Limits \& Continuous Dependence on Initial Conditions}
\lbl{sec:unifl}

We now prove Theorem \ref{thm:fluidunique}, which states that given an initial condition $\blv^0 \in \overline{\spv}^\infty$, a solution to the fluid model exists and is unique. As a direct consequence of the proof, we obtain an important corollary, that the unique solution $\blv(\cdot)$ depends \emph{continuously} on the initial condition $\blv^0$. 

The uniqueness result justifies the use of the fluid approximation, in the sense that the evolution of the stochastic system is close to a \emph{single} trajectory. The uniqueness along with the continuous dependence on the initial condition will be used to prove convergence of steady-state distributions to $\blv^I$ (Theorem~\ref{thm:convss}). 

We note that, in general, the uniqueness of solutions is not guaranteed for a differential equation with a discontinuous drift (see, e.g., \cite{AB08}). In our case, $\buf(\cdot)$ is discontinuous on the domain $\overline{\spv}^\ity$ due to the drift associated with central service tokens (Eq.~\eqref{eq:driftF}).

\begin{proof} {\bf (Theorem \ref{thm:fluidunique})} The existence of a solution to the fluid model follows from the fact that $\buv$ has a limit point (Proposition~\ref{prop:smooth}) and that all limit points of $\buv$ are solutions to the fluid model (Proposition \ref{prop:drifts}). We now show uniqueness. Define $i^p(\blv)\bydef \sup\{i: \blv_i > 0\}$.\footnote{\small $i^p(\blv)$ can be infinite; this happens if all coordinates of $\blv$ are positive.} Let $\blv(t),\blw(t)$ be two solutions to the fluid model such that $\blv(0)=\blv^0$ and $\blw(0)=\blw^0$, with $\blv^0,\blw^0\in \overline{\spv}^\infty$. At any regular point $t \geq 0$, where all coordinates of $\blv(t),\blw(t)$ are differentiable, without loss of generality, assume $i^p(\blv(t))\leq i^p(\blw(t))$, with equality if both are infinite. Let $\bla^\blv(\cdot)$ and $\bla^\blw(\cdot)$ be the arrival trajectories corresponding to $\blv(\cdot)$ and $\blw(\cdot)$, respectively, and similarly for $\bll$ and $\blc$. Since $\blv_0(t) = \blv_1(t) + 1$ for all $t\geq0$ by the boundary condition, and $\dot{\blv}_1 = \dot{\bla}^{\blv}_1-\dot{\bll}^{\blv}_1-\dot{\blc}^{\blv}_1$, for notational convenience we will write
\beq
\dot{\blv}_0 = \dot{\bla}^{\blv}_0-\dot{\bll}^{\blv}_0-\dot{\blc}^{\blv}_0,
\eeq
where
\beq
\lbl{eq:v0drift}
\dot{\bla}^{\blv}_0\bydef\dot{\bla}^{\blv}_1, \,\,\, \dot{\bll}^{\blv}_0\bydef\dot{\bll}^{\blv}_1, \,\, \mbox{ and } \dot{\blc}^{\blv}_0\bydef\dot{\blc}^{\blv}_1.
\eeq 
The same notation will be used for $\dot{\blw}(t)$. 

We have, 
\beqn
&& \frac{d}{dt}\lt\|\blv-\blw\rt\|^2_w  \bydef \frac{d}{dt}\sum_{i=0}^{\ity}\frac{\lt|\blv_i-\blw_i\rt|^2}{2^i} \stackrel{(a)}{=}\sum_{i=0}^{\ity} \frac{\lt(\blv_i-\blw_i\rt)\lt(\dot{\blv}_i-\dot{\blw}_i\rt)}{2^{i-1}} \nnb \\
&=& \sum_{i=0}^{\ity} \frac{\lt(\blv_i-\blw_i\rt)\lt[\lt(\dot{\bla}^\blv_i-\dot{\bll}_i^\blv\rt)-\lt(\dot{\bla}_i^\blw-\dot{\bll}_i^\blw\rt)\rt]}{2^{i-1}} - \sum_{i=0}^{\ity} \frac{\lt(\blv_i-\blw_i\rt)\lt(\dot{\blc}_i^\blv-\dot{\blc}^\blw_i\rt)}{2^{i-1}} \nnb \\
&\stackrel{(b)}{\leq}& C\lt\|\blv-\blw\rt\|_w^2- \sum_{i=0}^{\ity} \frac{\lt(\blv_i-\blw_i\rt)\lt(\dot{\blc}_i^\blv-\dot{\blc}^\blw_i\rt)}{2^{i-1}} \nnb \\
&=& C\lt\|\blv-\blw\rt\|_w^2- \sum_{i=0}^{i^p(\blv)}\frac{1}{2^{i-1}}\lt(\blv_i-\blw_i\rt)(p-p) \nnb \\
& & -  \frac{1}{2^{i^p(\blv)}}(0-\blw_{i^p(\blv)+1})(\min\{\lambda \blv_{i^p(\blv)},p\}-p) \nnb \\
& & - \sum_{i=i^p(\blv)+2}^{i^p(\blw)}\frac{1}{2^{i-1}}(0-\blw_i)(0-p) \nnb \\
& & - \sum_{j=i^p(\blw)+1}^{\ity}\frac{1}{2^{i-1}}(0-0)\lt(\dot{\blc}_i^\blv-\dot{\blc}_i^\blw\rt) \nnb \\
&\leq& C\lt\|\blv-\blw\rt\|_w^2,
\lbl{eq:gron1}
\eeqn
where $C=6(\lambda+1-p)$. {\color{black} We first justify the existence of the derivative $\frac{d}{dt}\lt\|\blv-\blw\rt\|^2_w$ and the exchange of limits in ($a$). Because $\blv_i(t)$ and $\blw_i(t)$ are $L$-Lipschitz-continuous for all $i$, it follows that there exists $L'>0$ such that for all $i$, $h(i,s)\bydef \lt|\blv_i(s)-\blw_i(s)\rt|^2$ is $L'$-Lipschitz-continuous in the second argument, within a small neighborhood around $s=t$. In other words,
\beq
\lbl{eq:exdiff}
\lt|\frac{h(i,t+\epsilon)-h(i,t)}{\epsilon}\rt| \leq L'
\eeq
for all $i$ and all sufficiently small $\epsilon$. Then,
\beqn
\lbl{eq:exdiff2}
& &\frac{d}{dt}\lt\|\blv-\blw\rt\|^2_w = \lim_{\epsilon\downarrow0}\sum_{i=0}^\ity 2^{-i} \frac{h(i,t+\epsilon)-h(i,t)}{\epsilon} \nnb\\
&=& \lim_{\epsilon\downarrow0} \int_{i \in \zp} \frac{h(i,t+\epsilon)-h(i,t)}{\epsilon} d\mu_{\N},
\eeqn
where $\mu_{\N}$ is a measure on $\zp$ defined by $\mu_{\N}(i)=2^{-i}, i \in \zp$. By Eq.~\eqref{eq:exdiff} and the dominated convergence theorem, we can exchange the limit and integration in Eq.~\eqref{eq:exdiff2} and obtain
\beqn
& &\frac{d}{dt}\lt\|\blv-\blw\rt\|^2_w  = \lim_{\epsilon\downarrow0} \int_{i \in \zp} \frac{h(i,t+\epsilon)-h(i,t)}{\epsilon} d\mu_{\N} \nnb\\
&=& \int_{i \in \zp} \lim_{\epsilon\downarrow0}  \frac{h(i,t+\epsilon)-h(i,t)}{\epsilon} d\mu_{\N}  \nnb \\
&=& \sum_{i=0}^{\ity} \frac{\lt(\blv_i-\blw_i\rt)\lt(\dot{\blv}_i-\dot{\blw}_i\rt)}{2^{i-1}},
\eeqn
which justifies step ($a$) in  Eq.~\eqref{eq:gron1}.} Step ($b$) follows from the fact that $\dot{\bla}$ and $\dot{\bll}$ are both continuous and linear in $\blv$ (see Eqs.~\eqref{eq:adr}--\eqref{eq:cdr}). The specific value of $C$ can be derived after some straightforward algebra, which we isolate in the following claim.
\begin{clm}
\beq
\sum_{i=0}^{\ity}\frac{\lt(\blv_i-\blw_i\rt)\lt[\lt(\dot{\bla}_i^\blv-\dot{\bll}_i^\blv\rt)-\lt(\dot{\bla}_i^\blw-\dot{\bll}_i^\blw\rt)\rt]}{2^{i-1}} \leq 6(\lambda+1-p)\lt\|\blv-\blw\rt\|_w^2,
\eeq
\end{clm}

\begin{proof}
Let $\blm_i\bydef\blv_i-\blw_i$. Note that for all $i\geq1$
\beqn
& & \lt(\blv_i-\blw_i\rt)\lt[\lt(\dot{\bla}_i^\blv-\dot{\bll}_i^\blv\rt)-\lt(\dot{\bla}_i^\blw-\dot{\bll}_i^\blw\rt)\rt] \nnb \\
&=& \lt(\blv_i-\blw_i\rt)\lt[\lambda(\blv_{i-1}-\blw_{i-1})-\lambda\lt(\blv_{i}-\blw_{i}\rt)  -(1-p)\lt(\blv_{i}-\blw_{i}\rt)+(1-p)\lt(\blv_{i+1}-\blw_{i+1}\rt)\rt] \nnb \\
&=& \blm_i\lt(\lambda \blm_{i-1} - \lambda \blm_i - (1-p)\blm_i + (1-p)\blm_{i+1}\rt) \nnb \\
&\leq& \frac{\lambda}{2}\lt(\blm_{i-1}^2+\blm_i^2\rt)-\lt(\lambda+1-p\rt)\blm_i^2+\frac{1-p}{2}\lt(\blm_{i}^2+\blm_{i+1}^2\rt)\nnb \\
&=& \lambda\blm_{i-1}^2+\lt(1-p\rt)\blm_{i+1}^2 -\frac{\lambda+1-p}{2}\blm_i^2 \nnb \\
&\leq& \lambda\blm_{i-1}^2+\lt(1-p\rt)\blm_{i+1}^2
\lbl{eq:m1}
\eeqn
For $i=0$, by Eq.~\eqref{eq:v0drift}, we have
\beqn \lt(\blv_0-\blw_0\rt)\lt[\lt(\dot{\bla}^\blv_0-\dot{\bll}^\blv_0\rt)-\lt(\dot{\bla}^\blw_0-\dot{\bll}^\blw_0\rt)\rt] &=& \lt(\blv_1-\blw_1\rt)\lt[\lt(\dot{\bla}^\blv_1-\dot{\bll}^\blv_1\rt)-\lt(\dot{\bla}^\blw_1-\dot{\bll}^\blw_1\rt)\rt] \nnb\\
&\leq& \lambda\blm_{0}^2+(1-p)\blm_{2}^2
\lbl{eq:m2}
\eeqn
Combining Eqs.~\eqref{eq:m1} and \eqref{eq:m1}, we have
\beqn
& & \sum_{i=0}^{\ity}\frac{\lt(\blv_i-\blw_i\rt)\lt[\lt(\dot{\bla}_i^\blv-\dot{\bll}_i^\blv\rt)-\lt(\dot{\bla}_i^\blw-\dot{\bll}_i^\blw\rt)\rt]}{2^{i-1}} \nnb \\
&\leq& 2(\lambda\blm_{0}^2+(1-p)\blm_{2}^2) + \sum_{i=1}^\ity \frac{1}{2^{i-1}}\lt(\lambda\blm_{i-1}^2+(1-p)\blm_{i+1}^2 \rt) \nnb \\
&\leq& 6(\lambda+1-p)\|\blv-\blw\|_w^2.
\eeqn
This proves the claim.
\end{proof}

Now suppose that $\blv^0=\blw^0$. By Gronwall's inequality
and Eq.\ \eqref{eq:gron1}, we have
\beq
\lbl{eq:gronw}
\lt\|\blv(t)-\blw(t)\rt\|^2_w \leq \lt\|\blv(0)-\blw(0)\rt\|^2_w e^{Ct} = 0, \quad \forall t \in [0,\infty),
\eeq
which establishes the uniqueness of fluid limits on $[0,\ity)$. $\ftomb$
\end{proof}

The following corollary is an easy, but important, consequence of the uniqueness proof.

\begin{cor}{\bf (Continuous Dependence on Initial Conditions)}
\lbl{cor:contdep}
Denote by $\blv(\blv^0,\cdot)$ the unique solution to the fluid model given initial condition $\blv^0\in \overline{\spv}^\infty$. If $\blw^n\ \in \overline{\spv}^\infty$ for all $n$, and $\|\blw^n-\blv^0\|_w\rar0$ as $n\rar\ity$, then for all $t \geq 0$,
\beq
\lim_{n\rar\infty}\|\blv(\blw^n,t)-\blv(\blv^0,t)\|_w=0.
\eeq
\end{cor}
\begin{proof}
The continuity with respect to the initial condition is a direct consequence of the inequality in Eq.~\eqref{eq:gronw}: if $\blv(\blw^n,\cdot)$ is a sequence of fluid limits with initial conditions $\blw^n \in \overline{\spv}^\ity$ and if $\|\blw^n-\blv^0\|^2_w\rar0$ as $N \rar \ity$, then for all $t \in [0,\ity)$,
\beq
\lt\|\blv(\blv^0,t)-\blv(\blw^n,t)\rt\|^2_w \leq \lt\|\blv^0-\blw^n\rt\|^2_w e^{Ct} \rar 0, \quad \mbox{as } n \rar \ity. \nnb
\eeq
This completes the proof. $\ftomb$
\end{proof}

\subsection{$\blv(\cdot)$ versus $\bls(\cdot)$, and the Uniqueness of Fluid Limits} 
\lbl{subsec:vvss}
As was mentioned in Section \ref{sec:sysstate}, we have chosen to work primarily with the aggregate queue length process, $\buv(\cdot)$ (Eq.~\eqref{eq:normq}), instead of the \emph{normalized queue length process}, $\bus(\cdot)$ (Eq.~\eqref{eq:aggq}). We offer some justification for such a choice in this section.

Recall that for any finite $N$, the two processes are related by simple transformations, namely,
\beqn
& \buv_i\lt(t\rt) \bydef \sum_{j=i}^\infty \bus_j\lt(t\rt), \quad i \geq 0. \nnb \\
& \mbox{ and } \bus_i(t) \bydef \buv_i(t)-\buv_{i+1}(t), \quad i \geq 0. \nnb
\eeqn
Therefore, there seems to be no obvious reason to favor one representation over the other when $N$ is finite. However, in the limit of $N \rar \ity$, it turns out that the \emph{fluid model} associated with $\buv(\cdot)$ is much easier to work with in establishing the important property of the uniqueness of fluid solutions (Theorem \ref{thm:fluidunique}).

A key to the proof of Theorem \ref{thm:fluidunique} is a contraction of the drift associated with $\blv(\cdot)$ (Eq.~\eqref{eq:gron1}), also known as the one-sided Lipschitz continuity (OSL) condition in the dynamical systems literature (see, e.g., \cite{AB08}). We first give a definition of OSL that applies to our setting.

\begin{defn}
\lbl{def:onel}
Let the coordinates of $\R^\ity$ be indexed by $\zp$ so that $\blx = (x_0, x_1, x_2, \ldots)$ for all $\blx \in \R^\ity$. A function  $\buh: \R^\ity \rar \R^\ity$ is said to be one-sided Lipschitz-continuous (OSL) over a subset $D \subset \R^\ity$, if there exists a constant $C$, such that for every $\blx, \bly \in D$, 
\beq
\lbl{eq:onesideL}
\lt\langle \blx-\bly, \buh(\blx)-\buh(\bly) \rt\rangle_w \leq C \lt\|\blx-\bly\rt\|_w^2,
\eeq
where the inner product $\lt\langle \cdot, \cdot \rt\rangle_w$ on $\R^\ity$ is defined by
\beq
\lt\langle \blx, \bly \rt\rangle_w \bydef \sum_{i=0}^\ity \frac{x_i y_i}{2^i}.
\eeq
\end{defn}

What is the usefulness of the above definition in the context of proving uniqueness of solutions to a fluid model? Recall that $\buf(\cdot)$ is the drift of the fluid model, as in Eq.~\eqref{eq:driftF}, i.e.,
\beq
\dot{\blv}\lt(t\rt) = \buf\lt(\blv(t)\rt),
\eeq
whenever $\blv(\cdot)$ is differentiable at $t$. Let $\blv(\cdot)$ and $\blw(\cdot)$ be two solutions to the fluid model such that both are differentiable at $t$, as in the proof of Theorem \ref{thm:fluidunique}. We have
\beq
\frac{d}{dt}\lt\|\blv(t)-\blw(t)\rt\|^2_w = 2\lt\langle \blv(t)-\blw(t), \buf(\blv(t))-\buf(\blw(t)) \rt\rangle_w.
\eeq
Therefore, if $\buf(\cdot)$ is one-sided Lipschitz-continuous, as defined by Eq.~\eqref{eq:onesideL}, we immediately obtain the key inequality in Eq.~\eqref{eq:gron1}, from which the uniqueness of fluid solutions follows. The computation carried out in Eq.~\eqref{eq:gron1} was essentially verifying the OSL condition of $\buf(\cdot)$ on the domain $\overline{\spv}^\ity$.

For the state representation based on $\bls(\cdot)$, can one use the same proof technique to show the uniqueness of $\bls(\cdot)$ by working directly with the drift associated with $\bls(\cdot)$? Recall that
\beq
\bls_i(t) \bydef \blv_i(t) - \blv_{i+1}(t), \quad \forall i \geq 0,
\eeq
so that at all $t$ where $\blv(t)$ is differentiable, the drift $\dot{\bls}(t)$ is given by
\beq
\lbl{eq:drifth}
\buh_i(\bls(t)) = \dot{\bls}_i (t) = \dot{\blv}_i (t)-\dot{\blv}_{i+1}(t) = \lambda(\bls_{i-1}-\bls_i)-(1-p)(\bls_{i}-\bls_{i+1}) - g^s_i(\bls),
\eeq
for all $i \geq 1$, where $g^s_i(\bls) \bydef g_i(\blv)-g_{i+1}(\blv)$, i.e.,

\begin{equation}
\lbl{eq:driftgs}
g^s_i\lt(\bls\rt) = \lt\{ \begin{array}{ll}
0, & \bls_i>0, \bls_{i+1}>0, \\
p-\min\lt\{\lambda\bls_{i},p\rt\}, & \bls_i>0, \bls_{i+1}=0, \\
\min\lt\{\lambda\bls_{i-1}, p\rt\}, & \bls_i=0, \bls_{i-1}>0, \\
0, & \bls_i=0, \bls_{i-1}=0.
\end{array} \right.
\end{equation}

Interestingly, it turns out that the drift $\buh(\cdot)$, defined in Eq.~\eqref{eq:drifth}, does not satisfy the one-sided Lipschitz continuity condition in general. We show this fact by inspecting a specific example. To keep the example as simple as possible, we consider a degenerate case.

\begin{clm} 
\lbl{clm:discH}
If $\lambda=0$ and $p=1$, then $\buh(\cdot)$ is not one-sided Lipschitz-continuous on its domain $\overline{\sps}^\ity$, where $\overline{\sps}^\ity$ was defined in Eq.~\eqref{eq:spsity} as
\beq
\overline{\sps}^\ity \bydef \lt\{\bls \in \sps: \sum_{i=1}^\infty \bls_i < \infty \rt\}.  \nnb
\eeq
\end{clm}
\begin{proof}
We will look at a specific instance where the condition \eqref{eq:onesideL} cannot be satisfied for any $C$. For the rest of the proof, a vector $\bls \in \overline{\sps}^\ity$ will be written explicitly as $\bls=(\bls_0,\bls_1,\bls_2,\ldots)$.  Consider two vectors
\beq
\lbl{eq:sab}
\bls^a = (1, \alpha, 0, 0, \ldots) \mbox{ and } \bls^b = (1, \alpha+\epsilon, \beta, 0, 0, \ldots), 
\eeq
where $1 \geq \alpha + \epsilon \geq \beta > 0$, $1 \geq \alpha > 0$, and $\bls^a_i=\bls^b_i=0$ for all $i \geq 3$. Note that $\bls^a_i,\bls^b_i \in \overline{\sps}^\ity$.  

To prove the claim, it suffices to show that for any value of $C$, there exist some values of $\alpha$, $\beta$, and $\epsilon$ such that
\beq
\lbl{eq:invalidosl}
\lt\langle \bls^b-\bls^a, \buh(\bls^b)-\buh(\bls^a) \rt\rangle_w  >  C\lt\|\bls^b-\bls^a\rt\|^2_w.
\eeq
Since $\lambda=0$ and $p=1$, by the definition of $\buh(\cdot)$ (Eqs.~\eqref{eq:drifth} and \eqref{eq:driftgs}), we have
\beq
\lbl{eq:hsab}
\buh(\bls^a)=(0, -1, 0, 0, \ldots) \mbox{ and } \buh(\bls^b) = (0, 0, -1, 0, 0, \ldots). 
\eeq
Combining Eqs.~\eqref{eq:sab} and \eqref{eq:hsab}, we have
\beqn
\bls^b-\bls^a &=& (0, \epsilon, \beta, 0, 0, \ldots),  \nnb \\
\mbox{and } \buh(\bls^b)-\buh(\bls^a) &=& (0, 1, -1, 0, 0, \ldots), \nnb
\eeqn
which yields
\beq
\lbl{eq:onellhs}
 \lt\langle \bls^b-\bls^a, \buh(\bls^b)-\buh(\bls^a) \rt\rangle_w = \frac{1}{2}\epsilon-\frac{1}{4}\beta.
\eeq
Since 
\beq
\lbl{eq:onellhr}
C\lt\|\bls^b-\bls^s\rt\|^2_w \bydef C\sum_{i=0}^{\ity} \frac{1}{2^i}\lt(\bls^b_i-\bls^a_i\rt)^2 = C\lt(\frac{1}{2}\epsilon^2+ \frac{1}{4}\beta^2\rt),
\eeq
we have that for all $C$ and all $\epsilon<\frac{1}{C}$, 
\beq
\lbl{eq:hineq}
\lt\langle \bls^a-\bls^b, \buh(\bls^b)-\buh(\bls^a) \rt\rangle_w = \frac{1}{2}\epsilon-\frac{1}{4}\beta > C\lt(\frac{1}{2}\epsilon^2+ \frac{1}{4}\beta^2\rt) = C\lt\|\bls^a-\bls^b\rt\|^2_w,
\eeq
for all sufficiently small $\beta$, which proves Eq. \eqref{eq:invalidosl}. This completes the proof of the claim.
\end{proof}

Claim \ref{clm:discH} indicates that a direct proof of uniqueness of fluid solutions using the OSL property of the drift will not work for $\bls(\cdot)$. The uniqueness of $\bls(\cdot)$ should still hold, but the proof can potentially be much more difficult, requiring an examination of all points of discontinuity of $\buh(\cdot)$ on the domain $\overline{\sps}^\ity$.

We now give some intuition as for why the discontinuity in Claim \ref{clm:discH} occurs for $\buh(\cdot)$, but not for $\buf(\cdot)$. The key difference lies in the expressions of the \emph{drifts due to central service tokens} in two fluid models, namely, $g(\cdot)$ (Eq.~\eqref{eq:gdef0}) for $\blv(\cdot)$ and $g^s(\cdot)$ (Eq.~\eqref{eq:driftgs}) for $\bls(\cdot)$. For $g^s(\cdot)$, note that
\beqn
\lbl{eq:gs1}
&g_i^s(\bls) = 0, \mbox{ if } \bls_i>0 \mbox{ and } \bls_{i+1}>0, \\
\lbl{eq:gs2}
& \mbox{and } g_i^s(\bls) = p-\min\lt\{\lambda\bls_{i},p\rt\}, \mbox{ if } \bls_i>0 \mbox{ and } \bls_{i+1}=0.
\eeqn
In other words, the $i$th coordinate of $\bls(t)$, $\bls_i(t)$ receives \emph{no drift} due to the central service tokens if there is a strictly positive fraction of queues in the system with at least $i+1$ tasks, that is, if $\bls_{i+1}(t)>0$ (Eq.~\eqref{eq:gs1}). However, as soon as $\bls_{i+1}(t)$ becomes zero, $\bls_i(t)$ immediately receives a strictly positive amount of drift due to the central service tokens (Eq.~\eqref{eq:gs2}), as long as $\lambda\bls_{i}(t)<p$. Physically, since the central server always targets the longest queues, this means that when $\bls_{i+1}(t)$ becomes zero, the set of queues with exactly $i$ tasks becomes the longest in the system, and begins to receive a positive amount of attention from the central server. Such a sudden change in the drift of $\bls_i(t)$ as a result of $\bls_{i+1}(t)$ hitting zero is a main cause of the failure of the OSL condition, and this can be observed in Eq.~\eqref{eq:hineq} as $\beta \rar 0$. In general, the type of discontinuities that was exploited in the proof of Claim \ref{clm:discH} can happen at infinitely many points in $\overline{\sps}^\ity$. The particular choices of $\lambda=0$ and $p=1$ were non-essential, and were only chosen to simplify the calculations. 

We now turn to the expression for $g(\cdot)$, the drift of $\blv(\cdot)$ due to the central service tokens. We have that
\beq
g_i(\blv) = p, \mbox{ whenever } \blv_i>0.
\eeq
Note that the above-mentioned discontinuity in $g^s(\cdot)$ is not present in $g(\cdot)$. This is not surprising: since
$\blv_{i}(t)\bydef\sum_{j=i}^\ity s_j(t)$, $\blv_i(t)$ receives a \emph{constant amount} of drift from the central service token as long as $\blv_i(t)>0$, \emph{regardless} of the values of $\blv_{j}(t), j \geq i+1$. By adding up the coordinates $\bls_j(\cdot), j\geq i$, to obtain $\blv_i(\cdot)$, we have effectively eliminated many of the drift discontinuities in $\bls(\cdot)$. This is a key reason for the one-sided Lipschitz continuity condition to hold for $\buf(\cdot)$. 

To illustrate this ``smoothing'' effect, consider again the examples of $\bls^a$ and $\bls^b$ in Eq.~\eqref{eq:sab}. In terms of $\blv$, we have
\beq
\lbl{eq:vab}
\blv^a = (1+\alpha, \alpha, 0, 0, \ldots) \mbox{ and } \blv^b = (1+\alpha+\epsilon+\beta, \alpha+\epsilon+\beta, \beta, 0, 0, \ldots).
\eeq
We then have
\beq
\lbl{eq:fvab}
\buf(\blv^a)=(-1, -1, 0, 0, \ldots) \mbox{ and } \buf(\blv^b) = (-1, -1, -1, 0, 0, \ldots). 
\eeq
Combining Eqs.~\eqref{eq:vab} and \eqref{eq:fvab}, we have
\beqn
\blv^b-\blv^a &=& (\epsilon+\beta, \epsilon+\beta, \beta, 0, 0, \ldots),  \nnb \\
\mbox{and } \buf(\blv^b)-\buf(\blv^a) &=& (0, 0, -1, 0, 0, \ldots). \nnb
\eeqn
This implies that for all $C \geq 0$,
\beq
\lbl{eq:fineq}
\lt\langle \blv^a-\blv^b, \buf(\blv^b)-\buf(\blv^a) \rt\rangle_w = -\frac{1}{4}\beta \leq C\lt\|\blv^a-\blv^b\rt\|^2_w,
\eeq
for all $\beta \geq 0$. Contrasting Eq.~\eqref{eq:fineq} with Eq.~\eqref{eq:hineq}, notice that the $\frac{1}{2}\epsilon$ term is no longer present in the expression for the inner product, as a result of the additional ``smoothness'' of $\buf(\cdot)$. Therefore, unlike in the case of $\buh(\cdot)$, the OSL condition for $\buf$ does not break down at $\blv^a$ and $\blv^b$. 

\begin{figure*}[h]
\lbl{fig:trajcompar}
\centering
\subfigure
{
\label{fig:discompares}
\centering
\includegraphics[scale=0.33]{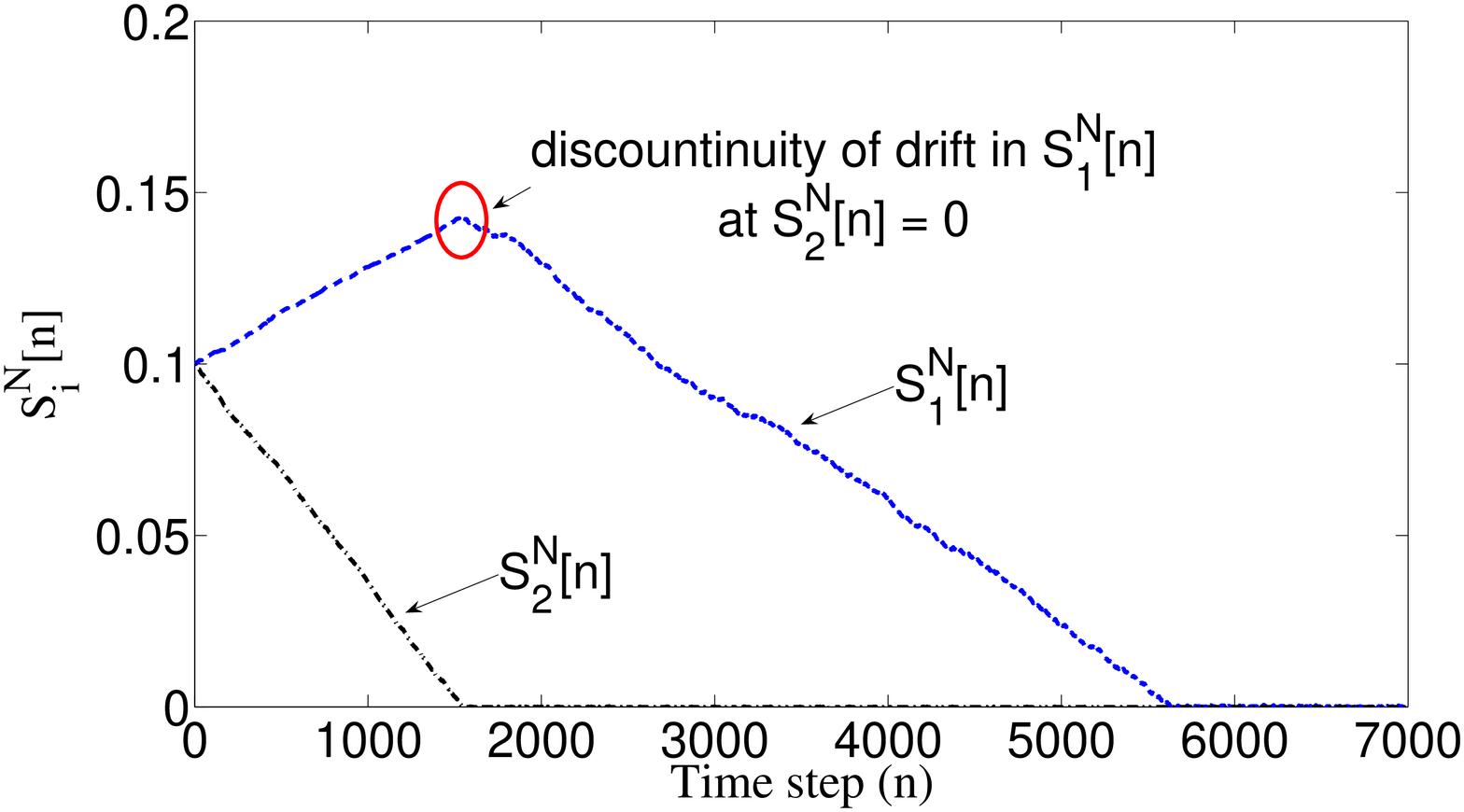}
}
\subfigure
{
\label{fig:discomparev}
\centering
\includegraphics[scale=0.33]{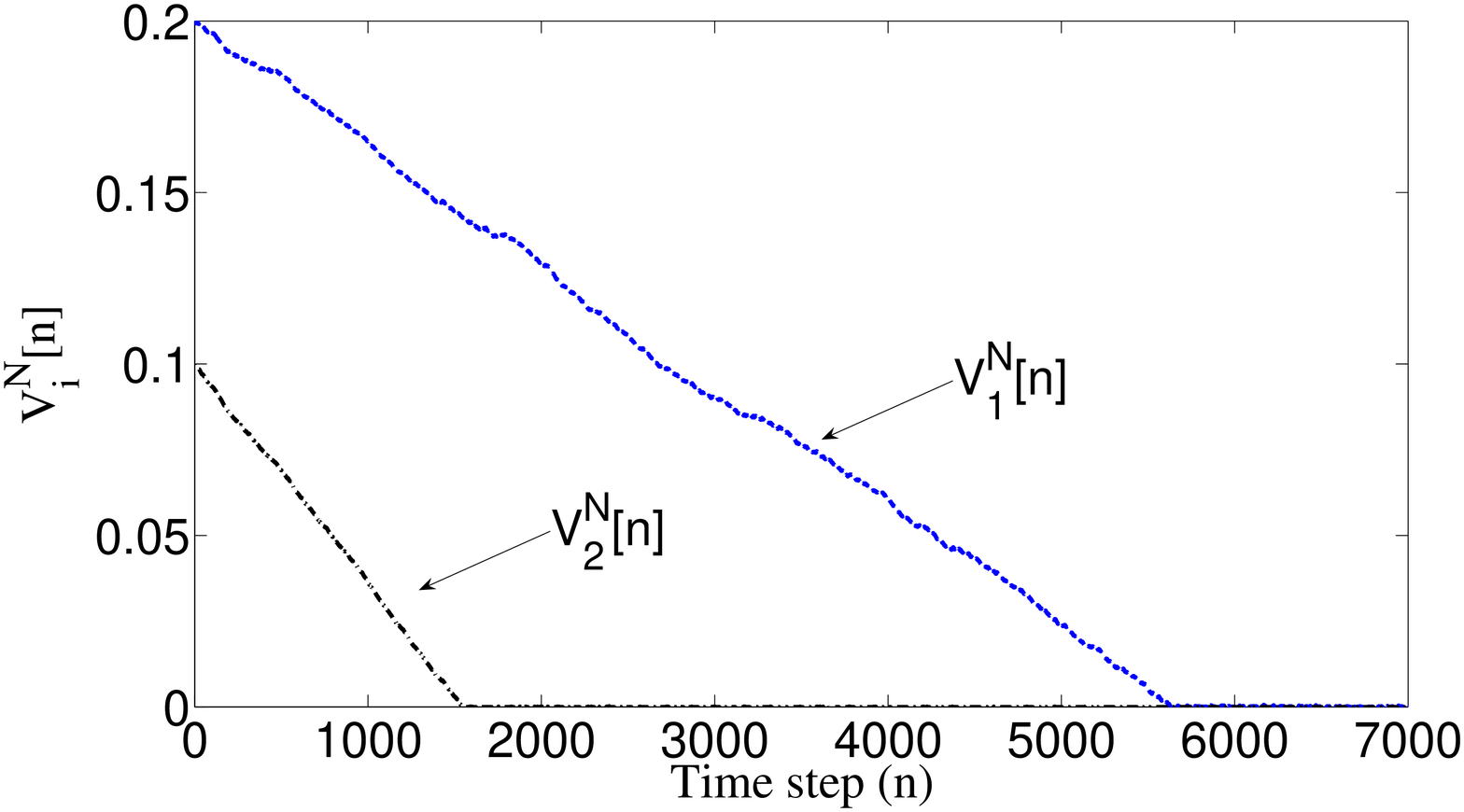}
}
\caption{Comparison between the $\buv[\cdot]$ and $\bus[\cdot]$ representations.}
\end{figure*}

The difference in drift patterns described above can also be observed in finite systems. The two graphs in Figure \ref{fig:trajcompar} display the \emph{same} sample path of the embedded discrete-time Markov chain, in the representations of $\bus$ and $\buv$, respectively. Here $N=10000$, $p=1$, and $\lambda=0.5$, with an initial condition $\bus[0] = (1, 0.1, 0.1, 0, 0, \ldots)$ (i.e., $100$ queues contain $2$ tasks and the rest of queues are empty). Notice that when $\bus_2[n]$ hits zero, $\bus_1[n]$ immediately receives an extra amount of downward drift. On the other hand, there is no change in drift for $\buv_1[n]$ when $\buv_2[n]$ hits zero. This is consistent with the previous analysis on the fluid models.

In summary, the difficulty of proving the uniqueness of fluid solutions is greatly reduced by choosing an appropriate state representation, $\blv(\cdot)$. The fact that such a simple (linear) transformation from $\bls(\cdot)$ to $\blv(\cdot)$ can create one-sided Lipschitz continuity and greatly simplify the analysis may be of independent interest.

\section{Convergence to Fluid Solution over a Finite Horizon}
\lbl{sec:thmtransconvpf}

We now prove Theorem \ref{thm:transconv}. 

\begin{proof} {\bf (Theorem \ref{thm:transconv})}
The proof follows from the sample-path tightness in Proposition \ref{prop:smooth} and the uniqueness of fluid limits from Theorem \ref{thm:fluidunique}. By assumption, the sequence of initial conditions $\mbf{V}^{(0,N)}$ converges to some $\blv^0\in\overline{\spv}^\ity$, in probability. Since the space $\overline{\spv}^\ity$ is separable and complete under the $\|\cdot\|_w$ metric, by Skorohod's representation theorem, we can find a probability space $(\Omega_0,\mcal{F}_0,\pb_0)$ on which $\mbf{V}^{(0,N)} \rar \blv^0$ almost surely. By Proposition \ref{prop:smooth} and Theorem \ref{thm:fluidunique}, for almost every $\omega \in \Omega$, any subsequence of $\buv(\omega,t)$ contains a further subsequence that converges to the unique fluid limit $\blv(\blv^0,t)$ uniformly on any compact interval $[0,T]$. Therefore for all $T<\ity$,
\beq
\hspace{-4.5pt}\lim_{N\rar\ity}\sup_{t \in [0,T]}\lt\|\buv(\omega,t)-\blv(\blv^0,t)\rt\|_w=0, \quad \pb\hspace{-3pt}-\hspace{-3pt}\mbox{almost surely,}
\eeq
which implies convergence in probability, and Eq.~\eqref{eq:probconv} holds.$\ftomb$
\end{proof}

\section {Convergence to the Invariant State $\blv^I$}
\lbl{sec:fldconv}
We will prove Theorem \ref{thm:fluidconv} in this section. We switch to the alternative state representation, $\bls(t)$, where 
\beq
\lbl{eq:sav}
\bls_i(t) \bydef \blv_{i+1}(t) -\blv_i(t), \quad \forall i \geq 0, 
\eeq
to study the evolution of a fluid solution as $t \rar \ity$. It turns out that a nice monotonicity property of the evolution of $\bls(t)$ induced by the drift structure will help establish the convergence to the invariant state. We recall that $\bls_0(t)=1$ for all $t$, and that for all points where $\blv$ is differentiable,
\beq
\dot{\bls}_i (t) = \dot{\blv}_i (t)-\dot{\blv}_{i+1}(t) = \lambda(\bls_{i-1}-\bls_i)-(1-p)(\bls_{i}-\bls_{i+1}) - g^s_i(\bls), \nnb
\eeq
for all $i \geq 1$, where $g^s_i(\bls) \bydef g_i(\blv)-g_{i+1}(\blv)$.
Throughout this section, we will use both representations $\blv(t)$ and $\bls(t)$ to refer to the \emph{same} fluid solution, with their relationship specified in Eq.\ \eqref{eq:sav}. 

The approach we will be using is essentially a variant of the convergence proof given in \cite{DOB96}. The idea is to partition the space $\overline{\sps}^\ity$ into dominating classes, and show that ($i$) dominance in initial conditions is preserved by the fluid model, and ($ii$) any solution $\bls(t)$ to the fluid model with an initial condition that dominates or is dominated by the invariant state $\bls^I$ converges to $\bls^I$ as $t\rar \ity$. Properties ($i$) and ($ii$) together imply the convergence of the fluid solution $\bls(t)$ to $\bls^I$, as $t \rar \ity$, for any finite initial condition. It turns out that such dominance in $\bls$ is much stronger than a similarly defined relation for $\blv$. For this reason we do not use $\blv$ but instead rely on $\bls$ to establish the result.

\begin{defn}{\bf (Coordinate-wise Dominance)} 
For any $\bls, \bls' \in \overline{\sps}^\ity$, we write 
\benum 
\item $\bls \succeq \bls'$ if $\bls_i \geq \bls'_i$, for all $i \geq 0$.
\item $\bls \succ \bls'$ if $\bls \neq \bls'$, $\bls \succeq \bls'$ and $\bls_i > \bls'_i$, for all $i \geq 1$ where $\bls'_i>0$. \footnote{We need the condition $\bls \neq \bls'$ in order to rule out the case where $\bls=\bls'=0$.}
\eenum
\end{defn}
The following lemma states that $\succeq$-dominance in initial conditions is preserved by the fluid model.
\begin{lemma}
\lbl{lm:dom}
Let $\bls^1(\cdot)$ and $\bls^2(\cdot)$ be two solutions to the fluid model such that $\bls^1(0) \succeq \bls^2(0)$. Then $\bls^1(t) \succeq \bls^2(t), \forall t \geq 0$.
\end{lemma}

\begin{proof}
By the continuous dependency of a fluid limit on its initial condition (Corollary \ref{cor:contdep}), it suffices to verify that $\bls^1(t) \succeq \bls^2(t), \forall t \geq 0$, whenever $\bls^1(0) \succ \bls^2(0)$ (strictly dominated initial conditions). 

Let $t_1$ be the first time when $\bls^1(t)$ and $\bls^2(t)$ become equal \emph{and} are both positive at least one coordinate:
\beq
t_1 \bydef \inf \lt\{t \geq 0: \bls^1(t_1) \neq \bls^2(t_1), \bls^1_i(t) = \bls_i^2(t) >0, \mbox{ for some } i \geq 1 \rt\},
\eeq
If $t_1 = \infty$, one of the following must be true:
\benum [(1)]
\item $\bls^1(t) \succ \bls^2(t)$, for all $t \geq 0$, in which case the claim holds. 
\item $\bls^1(t') = \bls^2(t')$ at some $t' <\ity$. By the uniqueness of solutions, $\bls^1(t) = \bls^2(t)$ for all $t \geq t'$, in which case the claim also holds.
\eenum

Hence, we assume $t_1<\infty$. Let $k$ be the smallest coordinate index such that $\bls^1(t_1)$ and $\bls^2(t_1)$ are equal at $k$, but differ on \emph{at least one} of the two adjacent coordinates, $k-1$ and $k+1$:
\beq
k \bydef \min \lt\{i\geq0: \bls^1_{i}(t_1)=\bls^2_{i}(t_1) >0, \max_{j\in \{1,-1\}}\{\bls^1_{i+j}(t_1) - \bls^2_{i+j}(t_1)\}>0 \rt\}
\eeq
Since $\bls^1(t_1)\succ\bls^2(t_1)$, at all regular points $t<t_1$ that are close enough to $t_1$,
\beq
\dot{\bls}_{k}^1(t)-\dot{\bls}_{k}^2(t) = \lambda(\bls_{k-1}^1-\bls_{k-1}^2)-(1-p)(\bls_{k+1}^2-\bls_{k+1}^1) - (g^s_k(\bls^1)-g^s_k(\bls^2)),
\eeq
where
\beqn
g^s_k(\bls^1)-g^s_k(\bls^2) &\leq& 0 \cdot \mb{I}\lt\{\bls^2_{k+1}>0\rt\} \nnb \\
& & +\lt[\lt(p-\min\{p,\lambda\bls^1_k\}\rt)-\lt(p-\min\{p,\lambda\bls^2_k\}\rt)\rt]\cdot\mb{I}\lt\{\bls^2_{k+1}=0\rt\} \nnb \\
&=& 0,
\eeqn
where the last equality comes from the fact that $\bls^1_k(t)=\bls^2_k(t)$ by the definition of $k$. Because $\bls^1(t)$ and $\bls^2(t)$ is a continuous function of $t$ in every coordinate, we can find a time $t_0 < t_1$ such that $\bls^1_{k}(t_0)>\bls^2_{k}(t_0)$ and
\beq
\dot{\bls}_{k}^1(t)-\dot{\bls}_{k}^2(t) > 0,
\eeq
for all regular $t \in (t_0,t_1)$. Since $\bls^1_{k}(t_1)-\bls^2_{k}(t_1)=\bls^1_{k}(t_0)-\bls^2_{k}(t_0)+\int_{s=t_0}^{t_1}(\dot{\bls}_{k}^1(t)-\dot{\bls}_{k}^2(t))ds$, this contradicts with the fact that $\bls^1_{k}(t_1)=\bls^2_{k}(t_1)$, and hence proves the claim.
\end{proof} 

We are now ready to prove Theorem \ref{thm:fluidconv}.
\begin{proof}{\bf (Theorem \ref{thm:fluidconv})}
Let $\bls(\cdot)$, $\bls^u(\cdot)$,  and $\bls^l(\cdot)$ be three fluid limits with initial conditions in $\overline{\sps}^\ity$ such that $\bls^u(0) \succeq \bls(0) \succeq \bls^l(0)$ and $\bls^u(0) \succeq \bls^I \succeq \bls^l(0)$. By Lemma \ref{lm:dom}, we must have $\bls^u(t) \succeq \bls^I \succeq \bls^l(t)$ for all $t \geq 0$. 
Hence it suffices to show that $\lim_{t \rar \infty} \lt\|\bls^u(t) - \bls^I \rt\|_w = \lim_{t \rar \infty} \lt\|\bls^l(t) - \bls^I \rt\|_w = 0.$ Recall, for any regular $t >0$,
\beqn
\dot{\blv}_i(t)
&=& \lambda (\blv_{i-1}(t)-\blv_{i}(t)) - (1-p)(\blv_i(t)-\blv_{i+1}(t)) - g_i(\blv(t)) \nnb \\ 
&=& \lambda \bls_{i-1}(t) - (1-p)\bls_{i}(t) - g_i(\blv(t))\nnb \\ 
&=& (1-p)\lt(\frac{\lambda \bls_{i-1}(t)-g_i(\blv(t))}{1-p} - \bls_i\rt).
\lbl{eq:sconvgen}
\eeqn
Recall, from the expressions for $\bls_i^I$ in Theorem \ref{thm:ssprop}, that $\bls^I_{i+1} \geq \frac{\lambda \bls^I_i-p}{1-p},  \, \forall i \geq 0$.
From Eq.\ \eqref{eq:sconvgen} and the fact that $\bls^u_0=\bls^I_0=1$, we have
\beq
\lbl{eq:s1der}
\dot{\blv}^u_1(t) = (1-p)\lt(\frac{\lambda-g_1(\blv^u (t))}{1-p} - \bls^u_1 (t)\rt)\leq(1-p)\lt( \bls^I_1- \bls^u_1(t)\rt),
\eeq
for all regular $t\geq0$. To see why the above inequality holds, note that 
\beq
\frac{\lambda-g_1(\blv^u (t))}{1-p}=\frac{\lambda-p}{1-p} \leq \bls^I_1,
\eeq 
whenever $\bls^u_1(t)>0$, and 
\beq
\frac{\lambda-g_1(\blv^u (t))}{1-p}=\bls^u_1(t)=0,
\eeq
whenever $\bls^u_1(t)=\bls^I_1=0$. We argue that Eq.~\eqref{eq:s1der} implies that 
\beq
\lim_{t \rar \infty} \lt| \bls^I_1 -\bls^u_1(t) \rt| = 0.
\eeq
To see why this is true, let $h_1(t)\bydef \bls^I_1-\bls^u_1(t)$, and suppose instead that 
\beq
\limsup_{t \rar \infty}\lt| \bls^I_1 -\bls^u_1(t) \rt| = \delta > 0.
\eeq
Because $\bls^u(t) \succeq \bls^I$ for all $t$, this is equivalent to having 
\beq
\liminf_{t \rar \infty} h_1(t) = -\delta.
\eeq
Since $\bls(t)$ is a fluid limit and is $L$-Lipschitz-continuous along all coordinates, $h_1(t)$ is also $L$-Lipschitz-continuous. Therefore, we can find an increasing sequence $\{t_k\}_{k\geq1} \subset \R^+$ with $\lim_{k \rar \ity} t_k =\ity$, such that for some $\gamma>0$ and all $k\geq 1$,
\beq
\lbl{eq:h1}
h_1(t) \leq -\frac{1}{2}\delta, \quad \forall t \in [t_k-\gamma,t_k+\gamma].
\eeq
Because $\blv_1(0)<\ity$ and $h_1(t)\leq 0$ for all $t$, it follows from Eqs.~\eqref{eq:s1der} and \eqref{eq:h1} that there exists some $T_0>0$ such that
\beq
\blv_1^u(t) =  \int_{s=0}^t \dot{\blv}_1^u(s) ds \leq \int_{s=0}^t (1-p)h_1(s) ds < 0,
\eeq
for all $t \geq T$, which clearly contradicts with the fact that $\blv_1(t)\geq 0$ for all $t$. This shows that we must have $\lim_{t \rar \infty} \lt|\bls^u_1(t) - \bls^I_1\rt| = 0$.

We then proceed by induction. Suppose $\lim_{t \rar \infty} \lt|\bls^u_i(t) - \bls^I_i\rt| = 0$ for some $i\geq1$. By Eq.\  \eqref{eq:sconvgen}, we have
\beqn
\dot{\blv}^u_{i+1}(t)&=&(1-p)\lt(\frac{\lambda \bls^u_{i}(t)-g_i(\blv^u (t))}{1-p} - \bls^u_{i+1}(t)\rt) \nnb\\
&=& (1-p)\lt(\frac{\lambda\bls_i^I-g_i(\blv^u (t))}{1-p} - \bls^u_{i+1}(t) + \ep_{i}^u\rt) \nnb\\
&\leq& (1-p)\lt( \bls_{i+1}^I- \bls^u_{i+1}(t)+\ep_{i}^u(t)\rt),
\eeqn \normalsize
where $\ep_{i}^u (t) \bydef \frac{\lambda}{1-p}\lt(\bls^u_i(t)-\bls^I_i\rt) \rar 0$ as $t\rar \ity$ by the induction hypothesis. With the same argument as the one for $\bls_1$, we obtain $\lim_{t \rar \infty} |\bls^u_{i+1}(t) - \bls^I_{i+1}| = 0$.
This establishes the convergence of $\bls^u(t)$ to $\bls^I$ along all coordinates, which implies 
\beq
\lim_{t \rar \infty} \lt\| \bls^u(t) - \bls^I \rt\|_w = 0.
\eeq
Using the same set of arguments we can show that $\lim_{t \rar \infty} \lt\|\bls^l(t) - \bls^I \rt\|_w = 0$. This completes the proof. $\ftomb$
\end{proof}

\subsection{A Finite-support Property of $\blv(\cdot)$ and Its Implications}

In this section, we discuss a \emph{finite-support} property of the fluid solution $\blv(\cdot)$. Although this property is not directly used in the proofs of other results in our work, we have decided to include it here because it provides important, and somewhat surprising, qualitative insights into the system dynamics.

\begin{prop}
\lbl{prop:fsptofv}
Let $\blv^0 \in \vinf$, and let $\blv(\blv^0,\cdot)$ be the unique solution to the fluid model with initial condition $\blv(\blv^0, 0)=\blv^0$. If $p>0$, then $\blv(\blv^0,t)$ has a finite support for all $t >0$, in the sense that 
\beq
\sup \lt\{i: \blv_i(\blv^0, t)>0\rt\} < \ity, \quad \forall t >0.
\eeq
\end{prop}

Before presenting the proof, we observe that the finite-support property stated in Proposition \ref{prop:fsptofv} is \emph{independent} of the size of the support of the initial condition $\blv^0$; even if all coordinates of $\blv(t)$ are strictly positive at $t=0$, the support of $\blv(t)$ immediately ``collapses'' to a finite number for any $t>0$. 

Note that a critical assumption in Proposition \ref{prop:fsptofv} is that $p>0$, i.e., the system has a non-trivial central server. In some sense, the ``collapse'' of $\blv(\cdot)$ into a finite support is essentially due to the fact that the central server always allocates its service capacity to the longest queues in the system. Proposition \ref{prop:fsptofv} illustrates that the worst-case queue-length in the system is well under control \emph{at all times}, thanks to the power of the central server. 

Proposition \ref{prop:fsptofv} also sheds light on the structure of the invariant state of the fluid model, $\blv^I$. Recall from Theorem \ref{thm:ssprop} that $\blv^I$ has a finite support \emph{whenever} $p>0$. Since by the global stability of fluid solutions (Theorem \ref{thm:fluidunique}), we have that
\beq
\lim_{t \rar \infty} \lt\|\blv\lt(t\rt) - \blv^I \rt\|_w=0,
\eeq
the fact that $\blv(t)$ admits a finite support for any $t>0$ whenever $p>0$ provides strong intuition for and partially explains the finite-support property of $\blv^I$. 

We now prove Proposition \ref{prop:fsptofv}.

\begin{proof} {\bf (Proposition \ref{prop:fsptofv})} 
We fix some $\blv^0 \in \overline{\spv}^\ity$, and for the rest of the proof we will write $\blv(\cdot)$ in place of $\blv(\blv^0,\cdot)$. It is not difficult to show, by directly inspecting the drift of the fluid model in Eq. (4), that if we start with an initial condition $\blv^0$ with a finite support, then the support remains finite at all times. Hence, we now assume $\blv^0_i>0$ for all $i$. First, the fact that $\blv^0 \in \vinf$ (i.e., $\blv^0_1 <\ity$) implies \beq
\lim_{i\rar \ity}\blv^0_i=0.
\eeq
This is because all coordinates of the corresponding vector $\bls^0$ are non-negative, and
\beq
\blv^0_i = \blv^0_1-\sum_{j=1}^{i-1}\bls^0_{j},
\eeq
where the second term on the right-hand side converges to $\blv^0_1$.

Assume that $\blv_i(t)>0$ for all $i$, over some small time interval $t\in[0,s]$. Since the magnitude of the drift on any coordinate $\blv_i$ is uniformly bounded from above by $\lambda+1$, and $\lim_{i\rar\ity}\blv^0_i =0$, for any $\epsilon>0$ we can find $s',N>0$ such that for all $i \geq N$ and $t \in [0,s']$,
\beq
\dot{\blv}_i(t)= \lambda(\blv_{i-1}-\blv_i)-(1-p)(\blv_{i}-\blv_{i+1})-g_i(\blv) \leq \epsilon -g_i(\blv) = -p+\epsilon.
\lbl{eq:constdrift}
\eeq

Since $\lim_{i\rar\ity}\blv_i^0=0$, Eq.~\eqref{eq:constdrift} shows that it is impossible to find any strictly positive time interval $[0,s]$ during which the fluid trajectory $\blv(t)$ maintains an infinite support. This proves the claim.
\end{proof}

\chapter{Convergence of Steady-State Distributions}

\lbl{sec:conss}
We will prove Theorem \ref{thm:convss} in this chapter, which states that, for all $N$, the Markov process $\buv(t)$ converges to a unique steady-state distribution, $\pi^N$, as $t \rar \ity $, and that the sequence $\{\pi^N\}_{N \geq 1}$ concentrates on the unique invariant state of the fluid model, $\blv^I$, as $N \rar \ity$. This result is of practical importance, as it guarantees that key quantities, such as the average queue length, derived from the expressions of $\blv^I$ also serve as accurate approximations for that of an actual finite stochastic system in the long run.

Note that by the end of this chapter, we will have established our \emph{steady-state approximation} results, i.e.,
\beq
\buv(t) \stackrel{t \rar \ity}{\longrightarrow} \pi^N \stackrel{N \rar \ity}{\longrightarrow} \blv^I,
\eeq
as was illustrated in Figure \ref{fig:techres} of Chapter \ref{sec:sumres}. Together with the transient approximation results established in the previous chapters, these conclude the proofs of all approximation theorems in this thesis.

\vspace{15pt}

Before proving Theorem \ref{thm:convss}, we first give an important proposition which strengthens the finite-horizon convergence result stated in Theorem \ref{thm:transconv}, by showing a uniform speed of convergence over any compact set of initial conditions. This proposition will be critical to the proof of Theorem \ref{thm:convss} which will appear later in the chapter.

\section{Uniform Rate of Convergence to the Fluid Limit}
{\color{black}
Let the probability space $(\Omega_1,\mcal{F}_1,\pb_1)$ be the product space of $(\Omega_W,\mcal{F}_W,\pb_W)$ and $(\Omega_U,\mcal{F}_U,\pb_U)$. Intuitively, $(\Omega_1,\mcal{F}_1,\pb_1)$ captures all exogenous arrival and service information. Fixing $\omega_1 \in \Omega_1$ and $\blv^0\in \overline{\spv}^M \cap \spq^N$, denote by $\buv(\blv^0, \omega_1, t)$ the resulting sample path of $\buv$ given the initial condition $\buv(0)=\blv^0$. Also, denote by $\blv\lt(\blv^0,t\rt)$ the solution to the fluid model for a given initial condition $\blv^0$. We have the following proposition.

\begin{prop}{\bf (Uniform Rate of Convergence to the Fluid Limit)}
\lbl{prop:unicon}
Fix $T>0$ and $M \in \N$. Let $K^N \bydef \overline{\spv}^M \cap \spq^N$. We have 
\beq
\lbl{eq:unirate}
\lim_{N \rar \infty} \sup_{\blv^{0} \in K^N} d^{\zp}\lt(\buv(\blv^0,\omega_1,\cdot),\blv(\blv^0,\cdot)\rt)=0, \quad \pb_1\mbox{-almost surely},
\eeq
where the metric $d^\zp(\cdot,\cdot)$ was defined in Eq.~\eqref{eq:dmetric}.
\end{prop}

\begin{proof}
The proof highlights the convenience of the sample-path based approach. By the same argument as in Lemma~\ref{lm:nice1}, we can find sets $\mcal{C}_W\subset \Omega_W$ and $\mcal{C}_U \subset \Omega_U$ such that the convergence in Eqs.~\eqref{eq:nice1} and \eqref{eq:nice2} holds over $\mcal{C}_W$ and $\mcal{C}_U$, respectively, and that $\pb_W(\mcal{C}_W)=\pb_U(\mcal{C}_U)=1$. Let $\mcal{C}_1\bydef \mcal{C}_W\times\mcal{C}_U$. Note that $\pb_1(\mcal{C}_1)=1$.

To prove the claim, it suffices to show that 
\beq
\lim_{N\rar\ity}\sup_{\blv^{0} \in K^N} d^{\zp}\lt(\buv(\blv^0,\omega_1,\cdot),\blv(\blv^0,\cdot)\rt)=0, \quad \forall\omega_1 \in \mcal{C}_1.\eeq
We start by assuming that the above convergence fails for some $\tilde{\omega}_1\in\mcal{C}_1$, which amounts to having a sequence of ``bad'' sample paths of $\buv$ that are always a positive distance away from the corresponding fluid solution with the same initial condition, as $N \rar\ity$. We then find nested subsequences within this sequence of bad sample paths, and construct two solutions to the fluid model with the \emph{same} initial condition, contradicting the uniqueness of fluid model solutions. 

Assume that there exists $\tilde{\omega}_1\in \mcal{C}_1$ such that
\beq
\lbl{eq:vdiverg}
\limsup_{N\rar\ity}\sup_{\blv^{0} \in K^N} d^{\zp}\lt(\buv(\blv^0,\tilde{\omega}_1,\cdot),\blv(\blv^0,\cdot)\rt)>0.
\eeq

This implies that there exists $\epsilon>0$, $\{N_i\}_{i=1}^\ity\subset \N$, and $\lt\{\blv^{(0,N_i)}\rt\}_{i=1}^\ity$ with $\blv^{(0,N_i)}\in K^{N_i}$, such that
\beq
d^{\zp}\lt(\buv(\blv^{(0,N_i)},\tilde{\omega}_1,\cdot),\blv(\blv^{(0,N_i)},\cdot)\rt)>\epsilon,
\eeq
for all $i \in \N$. We make the following two observations:
\benum
\item The set $\overline{\spv}^M$ is closed and bounded, and the fluid solution $\blv(\blv^{(0,N_i)},\cdot)$ is $L$-Lipschitz-continuous for all $i$. Hence the sequence of functions $\{\blv(\blv^{(0,N_i)},\cdot)\}_{i=1}^\ity$ are equicontinuous and uniformly bounded on $[0,T]$. We have by the Arzela-Ascoli theorem that there exists a subsequence $\lt\{N^2_i\rt\}_{i=1}^\ity$ of $\lt\{N^1_i\rt\}_{i=1}^\ity$ such that 
\beq
d^{\zp}\lt(\blv\lt(\blv^{(0,N^2_i)},\cdot\rt),\tilde{\blv}^a(\cdot)\rt) \rar 0,\eeq 
as $i \rar \ity$, for some Lipschitz-continuous function $\tilde{\blv}^a(\cdot)$ with $\tilde{\blv}^a(0)\in \overline{\spv}^M$. By the \emph{continuous dependence of fluid solutions on initial conditions} (Corollary~\ref{cor:contdep}), $\tilde{\blv}^a(\cdot)$ must be the unique solution to the fluid model with initial condition $\tilde{\blv}^a(0)$, i.e.,
\beq
\lbl{eq:va}
\tilde{\blv}^a(t)= \blv\lt(\tilde{\blv}^a(0),t\rt), \quad \forall t \in [0,T].
\eeq
\item Since $\omega_1 \in \mcal{C}_1$, by Propositions~\ref{prop:smooth} and \ref{prop:drifts}, there exists a further subsequence $\lt\{N^3_i\rt\}_{i=1}^\ity$ of $\lt\{N^2_i\rt\}_{i=1}^\ity$ such that $\mbf{V}^{N^3_i}\lt(\blv^{(0,N^3_i)},\cdot\rt) \rar \tilde{\blv}^b(\cdot)$ uniformly over $[0,T]$ as $i \rar \ity$, where  $\tilde{\blv}^b(\cdot)$ is a solution to the fluid model. Note that since $\lt\{N^3_i\rt\}_{i=1}^\ity \subset \lt\{N^2_i\rt\}_{i=1}^\ity$, we have $\tilde{\blv}^b(0)=\tilde{\blv}^a(0)$. Hence,
\beq
\lbl{eq:vb}
\tilde{\blv}^b(t)= \blv\lt(\tilde{\blv}^a(0),t\rt),  \quad \forall t \in [0,T].
\eeq
\eenum
By the definition of $\tilde{\omega}_1$ (Eq.\ \eqref{eq:vdiverg}) and the fact that $\tilde{\omega}_1 \in \mcal{C}_1$, we must have\\ $\sup_{t\in[0,T]}\lt\|\tilde{\blv}^a(t)-\tilde{\blv}^b(t)\rt\|_w > \epsilon$, which, in light of Eqs.~\eqref{eq:va} and \eqref{eq:vb}, contradicts the uniqueness of the fluid limit (Theorem \ref{thm:fluidunique}). This completes the proof. $\ftomb$
\end{proof}

The following corollary, stated in terms of convergence in probability, follows directly from Proposition \ref{prop:unicon}. The proof is straightforward and is omitted.
\begin{cor}
\lbl{cor:unicon}
Fix $T>0$ and $M \in \N$. Let $K^N \bydef \overline{\spv}^M \cap \spq^N$. Then, for all $\delta>0$,
\beq
\lbl{eq:unirate2}
\lim_{N \rar \infty} \pb_1 \lt(\omega_1\in \Omega_1: \sup_{\blv^{0} \in K^N} d^{\zp}\lt(\buv\lt(\blv^0,\omega_1,\cdot\rt),\blv(\blv^0,\cdot)\rt) >\delta\rt)=0.
\eeq
\end{cor}
}

\section{Proof of Theorem \ref{thm:convss}}
We first state a tightness result that will be needed in the proof of Theorem \ref{thm:convss}. 

\begin{prop}
\lbl{prop:vtinght}
For every $N<\ity$ and $p\in (0,1]$, $\buv(t)$ is positive-recurrent and $\buv(t)$ converges in distribution to a unique steady-state distribution $\pi^{N,p}$ as $t \rar \ity$. Furthermore, the sequence $\{\pi^{N,p}\}_{N=1}^\ity$ is tight, in the sense that for all $\epsilon>0$, there exists $M>0$ such that
\beq
\pi^{N,p} \lt(\overline{\spv}^M\rt) \bydef \pi^{N,p} \lt(\buv_1 \leq M\rt) \geq 1-\epsilon, \quad \forall N \geq 1.
\eeq
\end{prop}
\emph{Proof Sketch}. The proposition is proved using a stochastic dominance argument, by coupling with the case $p=0$. While the notation may seem heavy, the intuition is simple: when $p=0$, the system degenerates into a collection of $M/M/1$ queues with independent arrivals and departures (but possibly correlated initial queue lengths), and it is easy to show that the system is positive recurrent and the resulting sequence steady-state distributions is tight as $N\rar\ity$. The bulk of the proof is to formally argue that when $p>0$, the system behaves ``no worse'' than when $p=0$ in terms of positive recurrence and tightness of steady-state distributions. See Appendix \ref{app:domproof} for a complete proof using this stochastic dominance approach.$\qed$

{\bf Remark.} It is worth mentioning that the tightness of $\pi^{N,p}$ could alternatively be established by defining a Lyapunov function on $\overline{\mcal{V}}^N$ and checking its drift with respect to the underlying embedded-discrete-time Markov chain. By applying the Foster-Lyapunov stability criterion, one should be able to prove positive recurrence of $\buv$ and give an explicit upper-bound on the expected value of $\buv_1$ in steady state.\footnote{For an overview of the use of the Foster-Lyapunov criterion in proving stability in queueing networks, see, e.g., \cite{FK04}.} If this expectation is bounded as $N \rar \ity$, we will have obtained the desirable result by the Markov inequality. We do not pursue this direction in this thesis, because we believe that the stochastic dominance approach adopted here provides more insight by exploiting the monotonicity in $p$ in the steady-state queue length distribution. 

\begin{proof} {\bf(Theorem \ref{thm:convss})}
For the rest of the proof, since $p$ is fixed, we will drop $p$ in the super-script of $\pi^{N,p}$. By Proposition \ref{prop:vtinght}, the sequence of distributions $\pi^N$ is tight, in the sense that for any $\ep>0$, there exists $M(\ep) \in \N$ such that for all $M \geq M(\ep)$, $\pi^N\lt(\overline{\spv}^{M}\cap\mcal{Q}^N\rt) \geq 1-\ep, \mbox{ for all } N.$ 

The rest of the proof is based on a classical technique using continuous test functions (see Chapter 4 of \cite{KZ05}). The continuous dependence on initial conditions and the uniform rate of convergence established previously will be used here. Let $\overline{C}$ be the space of bounded continuous functions from $\overline{\spv}^\ity$ to $\mb{R}$. Define the mappings $T^N(t)$ and $T(t)$ on $\overline{C}$ by:
\beqn
\lt(T^N(t)f\rt)(\blv^0) &\bydef& \mb{E}\lt[f\lt(\buv(t)\rt)\mid\buv(0)=\blv^0\rt], \nnb \\
\mbox{and } \lt(T(t)f\rt)(\blv^0) &\bydef& \mb{E}\lt[f\lt(\blv(t)\rt)\mid\blv(0)=\blv^{0}\rt]=f(\blv(\blv^0,t)), \mbox{ for } f \in \overline{C}. \nnb
\eeqn 
With this notation, $\pi^N$ being a steady-state distribution for the Markov process $\buv(t)$ is equivalent to having for all $t \geq 0, \, f \in \overline{C}$,
\beq
\int_{\blv^0\in \overline{\spv}^\ity \cap \mcal{Q}^N} T^N(t)f(\blv^0) d \pi^N = \int_{\blv^0\in \overline{\spv}^\ity \cap \mcal{Q}^N} f(\blv^0) d \pi^N.
\eeq
Since $\{\pi^N\}$ is tight, it is sequentially compact under the topology of weak convergence, by Prokhorov's theorem. Let $\pi$ be the weak limit of some subsequence of $\lt\{\pi^N\rt\}$. We will show that for all $t \geq 0$, $f \in \overline{C}$,
\beq
\lbl{eq:pi1}
\lt|\int_{\blv^0\in \overline{\spv}^\ity} T(t)f(\blv^0) d \pi(\blv^0) - \int_{\blv^0\in \overline{\spv}^\ity} f(\blv^0) d \pi(\blv^0)\rt| = 0.
\eeq
In other words, $\pi$ is also a steady-state distribution for the deterministic fluid limit. Since by Theorem \ref{thm:ssprop}, the invariant state of the fluid limit is unique, Eq.\ \eqref{eq:pi1} will imply that $\pi\lt(\blv^I\rt)=1$, and this proves the theorem. 

To show Eq.\ \eqref{eq:pi1}, we write
\beqn
\lbl{eq:pi2}
& & \lt| \int{T(t)f d \pi} - \int{f d \pi}\rt| \leq \limsup_{N \rar \ity} \lt|\int{T(t)f d\pi} - \int{T(t)f d\pi^N}\rt|  \nnb \\
& & \quad\quad\quad\quad\quad\quad\quad  + \limsup_{N \rar \ity} \lt|\int{T(t)f d\pi^N} - \int{T^N(t)f d\pi^N}\rt|  \nonumber \\
& & \quad\quad\quad\quad\quad\quad\quad + \limsup_{N \rar \ity} \lt|\int{T^N(t) f d\pi^N} - \int{f d\pi}\rt|
\eeqn

We will show that all three terms on the right-hand side of Eq.\ \eqref{eq:pi2} are zero. Since $\blv(\blv^0,t)$ depends continuously on the initial condition $\blv^0$ (Corollary \ref{cor:contdep}), we have $T(t)f \in \overline{C}, \forall t \geq 0$, which along with $\pi^N \Rightarrow \pi$ implies that the first term is zero. For the third term, since $\pi^N$ is the steady-state distribution of $\buv$, we have that $\int{T^N(t) f d \pi^N} = \int{f d \pi^N}$, $\forall t \geq 0$, $f \in \overline{C}$. Since $\pi^N \Rightarrow \pi$, this implies that the last term is zero. 

To bound the second term, fix some $M\in\N$ and let $K = \overline{\spv}^M$. We have
\beqn
\lbl{eq:intm1}
& &\limsup_{N \rar \ity} \left|\int{T(t)f d\pi^N} - \int{T^N(t)f d\pi^N}\right|  \nnb \\
&\leq& \limsup_{N \rar \ity} \left|\int_{K}{T(t)f d\pi^N} - \int_{K}{T^N(t)f d\pi^N}\right| \nnb \\
& &+\limsup_{N \rar \ity} \left|\int_{K^c}{T(t)f d\pi^N} - \int_{K^c}{T^N(t)f d\pi^N}\right| \nnb \\
&\stackrel{(a)}{\leq}& \limsup_{N \rar \ity} \int_{K}{\lt|T^N(t)f - T(t)f\rt| d\pi^N} +\limsup_{N \rar \ity} 2\left\| f\right\| \pi^N\hspace{-3pt}\lt(K^c\rt)  \nnb \\ 
&\stackrel{(b)}{=}& \limsup_{N \rar \ity} 2\left\| f\right\| \pi^N\hspace{-3pt}\lt(K^c\rt), 
\eeqn
where $K^c \bydef \overline{\spv}^\ity - K$ and $\|f\|\bydef\sup_{\blv\in\overline{V}^\ity}|f(\blv)|$. The inequality  ($a$) holds because $T(t)$ and $T^N(t)$ are both conditional expectations and are hence contraction mappings with respect to the $\sup$-norm $\|f\|$. Equality $(b)$ ($\limsup_{N \rar \ity}$ $\int_{K}{\lt|T^N(t)f - T(t)f\rt| d\pi^N}=0$) can be shown using an argument involving interchanges of the order of integration, which essentially follows from the uniform rate of convergence to the fluid limit over the compact set $K$ of initial conditions (Corollary \ref{cor:unicon}). We isolate equality $(b)$ in the following claim:
\begin{clm} Let $K$ be a compact subset of $\overline{\mcal{V}}^\ity$, we have
\beq
\limsup_{N \rar \ity} \int_{K}{\lt|T^N(t)f - T(t)f\rt| d\pi^N}=0
\eeq
\end{clm}
\begin{proof}
Fix any $\delta > 0$, there exists $N(\delta)>0$ such that for all $N \geq N(\delta)$, we have
\beqn
& & \lt|\int_{K}{T(t)f d\pi^N} - \int_{K}{T^N(t)f d\pi^N}\rt| \leq
\int_{K} \lt|T(t)f - T^N(t)f\rt| d\pi^N \nnb \\
&=& \int_{\blv^0 \in K} \Big|f(\blv(\blv^0,t)) - \mb{E}\lt[ f\lt(\buv(t)\rt)\big| \buv(0)=\blv^0\rt]\Big| d\pi^N(\blv^0)\nnb \\
&\leq& \int_{\blv^0 \in K} \lt( \int_{\blv^t\in\overline{V}^\ity\cap\mcal{Q}^N} \lt|f\lt(\blv(\blv^0,t)\rt) - f\lt(\blv^t\rt) \rt| d \pb_{\lt.\buv(t)\rt|\buv(0)}\lt(\lt.\blv^t\rt|\blv^0\rt) \rt) d\pi^N(\blv^0) \nnb \\
&\stackrel{(a)}{\leq}& \int_{\blv^0 \in K} \sup_{\blv^t \in \overline{\mcal{V}}^\ity, \lt\|\blv^t-\blv(\blv^0,t)\rt\|_w\leq\delta}\lt|f(\blv(\blv^0,t))-f(\blv^t)\rt| d\pi^N(\blv^0) \nnb \\ 
&\leq& \omega_{f}(K^\delta,\delta), \nnb
\eeqn
where $K^\delta$ is the $\delta$-extension of $K$,
\beq
K^\delta \bydef \lt\{\blx \in \overline{\mcal{V}}^\ity: \|\blx-\bly\|_w \leq \delta \mbox{ for some } y\in K \rt\},
\eeq
and $\omega_{f}(X,\delta)$ is defined to be the modulus of continuity of $f$ restricted to set $X$:
\beq
\omega_{f}(K,\delta) \bydef \sup_{\blx, \bly \in X,\lt\|\blx-\bly\rt\|_w \leq \delta}\lt|f(\blx)-f(\bly)\rt|.
\eeq

\vspace{-10pt}

To see why inequality $(a)$ holds, recall that by Corollary \ref{cor:unicon}, starting from a compact set of initial conditions, the sample paths of a finite system stay uniformly close to that of the fluid limit on a compact time interval with high probability. Inequality $(a)$ then follows from Eq.~\eqref{eq:unirate2} and the fact that $f$ is bounded. Because $K$ is a compact set, it is not difficult show that $K^{\delta^0}$ is also compact for some fixed $\delta^0>0$. Hence $f$ is uniformly continuous on $K^{\delta^0}$, and we have
\beq
\limsup_{N \rar \ity} \left|\int_{K}{T(t)f d\pi^N} - \int_{K}{T^N(t)f d\pi^N}\right| \leq \limsup_{\delta \rar 0} \omega_{f}(K^{\delta^0},\delta) =0,
\eeq
which establishes the claim.
\end{proof}
\vspace{-5pt}
Going back, since Eq.~\eqref{eq:intm1} holds for \emph{any} $K = \overline{\spv}^M, M \in \N$, we have, by the tightness of $\pi^N$, that the middle term in Eq.~\eqref{eq:pi2} is also zero. This shows that any limit point $\pi$ of $\lt\{\pi^N\rt\}$ is indeed the unique invariant state of the fluid model ($\blv^I$). This completes the proof of Theorem \ref{thm:convss}.$\ftomb$
\end{proof}

\chapter{Conclusions and Future Work}

The overall objective of this thesis is to study how the degree of centralization in allocating computing or processing resources impacts performance. This investigation was motivated by applications in server farms, cloud centers, as well as more general scheduling problems with communication constraints. Using a fluid model and associated convergence theorems, we showed that any small degree of centralization induces an exponential performance improvement in the steady-state scaling of system delay, for sufficiently large systems. Simulations show good accuracy of the model even for moderately-sized finite systems ($N=100$).

There are several interesting and important questions which we did not address in this thesis. We have left out the question of what happens when the central server adopts a scheduling policy different from the Longest-Queue-First (LQF) policy considered in this thesis. Since scheduling a task from a longest queue may require significant global communication overhead, other scheduling policies that require less global information may be of great practical interest. Some alternatives include
\benum
\item (\emph{Random $k$-Longest-Queues}) The central server always serves a task from a queue chosen uniformly at random among the $k$ most loaded queues, where $k\geq 2$ is a fixed integer. Note that the LQF policy is a sub-case, corresponding to $k=1$.
\item (\emph{Random Work-Conserving}) The central server always serves a task from a queue chosen uniformly at random among all non-empty queues.
\eenum
It will be interesting to see whether a similar exponential improvement in the delay scaling is still present under these other policies. Based on the analysis done in this thesis, as well as some heuristic calculations using the fluid model, we conjecture that in order for the phase transition phenomenon to occur, a \emph{strictly positive} fraction of the central service tokens must be used to serve a longest queue. Hence, between the two policies listed above, the former is more likely to exhibit a similar delay scaling improvement than the latter. 

Assuming the LQF policy is used, another interesting question is whether a non-trivial delay scaling can be observed if $p$, instead of being fixed, is a function of $N$ and decreases to zero as $N \rar \ity$. This is again of practical relevance, because having a central server whose processing speed scales linearly with $N$ may be expensive or infeasible for certain applications. To this end, we conjecture that the answer is negative, in that as long as $\limsup_{N\rar\ity}p(N)=0$, the limiting delay scaling will be the same as if $p(N)$ is fixed at $p=0$, in which case $\blv_1 \sim \frac{1}{1-\lambda}$ as $\lambda \rar 1$.

On the modeling end, some of our current assumptions could be restrictive for practical applications. For example, the transmission delays between the local and central stations are assumed to be negligible compared to processing times; this may not be true for data centers that are separated by significant geographic distances. Also, the arrival and processing times are assumed to be Poisson, while in reality more general traffic distributions (e.g., heavy-tailed traffic) are observed. Finally, the speed of the central server may not be able to scale linearly in $N$ for large $N$. Further work to extend the current model by incorporating these realistic constraints could be of great interest, although obtaining theoretical characterizations seems quite challenging. 

Lastly, the surprisingly simple expressions in our results make it tempting to ask whether similar performance characterizations can be obtained for other stochastic systems with partially centralized control laws; insights obtained here may find applications beyond the realm of queueing theory.

\appendix
\chapter{Appendix: Additional Proofs}
\lbl{app:techproof}

\section{Complete Proof of Proposition~\ref{prop:smooth}}
\lbl{app:smoothproof}
Here we will follow a line of argument in \cite{BRS98} to establish the existence of a set of fluid limits. We begin with some definitions. Recall the uniform metric, $d(\cdot, \cdot)$, defined on $D[0,T]$:
\beq
d(x,y) \bydef \sup_{t \in [0,T]}\lt|x(t)-y(t)\rt|, \quad x, y \in D[0,T].
\eeq

\begin{defn}
Let $E_c$ be a non-empty compact subset of $D[0,T]$. A sequence of subsets of $D[0,T]$, $\mcal{E}=\lt\{E_N\rt\}_{N\geq1}$, is said to be \emph{asymptotically close} to the set $E_c$ if the distance to $E_c$ of any element in $E_N$ decreases to zero uniformly, i.e.:
\beq
\lim_{N \rar\infty}\sup_{x \in E_N} d\lt(x,E_c\rt) =0,
\eeq
where the distance from a point to a set is defined as
\beq
d\lt(x,E_c\rt) \bydef \inf_{y\in E_c}d\lt(x,y\rt).
\eeq
\end{defn}

\begin{defn}
A point $y \in D[0,T]$ is said to be a cluster point of a sequence $\lt\{x_N\rt\}_{N\geq1}$,  if its $\gm-$neighborhood is visited by $\lt\{x_N\rt\}_{N\geq1}$ infinitely often for all $\gm > 0$, i.e., 
\beq
\liminf_{n\rar\infty}d\lt(x_N,y\rt) = 0.
\eeq
A point $y \in D[0,T]$ is a cluster point of a sequence of subsets $\mathcal{E}=\lt\{E_N\rt\}_{N\geq 1}$, if it is a cluster point of some $\lt\{x_N\rt\}_{N\geq1}$ such that:
\beq
x_N \in E_N, \quad \forall N \geq 1.
\eeq
\end{defn}

\begin{lemma}
\lbl{lm:msp}
Let $C\lt(\mathcal{E}\rt)$ be the set of cluster points of $\mathcal{E}=\lt\{E_N\rt\}_{N\geq 1}$. If $\mathcal{E}$ is asymptotically close to a compact and closed set $E_c$ then,
\benum
\item $\mcal{E}$ is asymptotically close to $C\lt(\mathcal{E}\rt)$.
\item $C\lt(\mathcal{E}\rt) \subset E_c$.
\eenum
\end{lemma}

\begin{proof}
Suppose that the first claim is false. Then there exists a subsequence $\lt\{x_{N_i}\rt\}_{i\geq 1}$, where $x_{N_i} \in E_{N_i}, \forall i$, such that
\beq
\lbl{eq:posdis1}
d\lt(x_{N_i},C\lt(\mcal{E}\rt)\rt) = \gm > 0, \quad \forall i\geq1.
\eeq
However, since $\mcal{E}$ is asymptotically close to $E_c$ by assumption, there exists $\lt\{y_i\rt\} \subset E_c$ such that 
\beq
\lbl{eq:tight1}
d\lt(x_{N_i},y_i\rt) \rightarrow 0, \quad \mbox{as } i \rar \infty.
\eeq
Since $E_c$ is compact, $\lt\{y_i\rt\}$ has a convergent subsequence with limit $\tilde{y}$. By Eq.~\eqref{eq:tight1}, $\tilde{y}$ is a cluster point of $\lt\{x_{N_i}\rt\}$, and hence a cluster point of $\mcal{E}$, contradicting Eq.~\eqref{eq:posdis1}. This proves the first claim.

The second claim is an easy consequence of the closedness of $E_c$: Let $\tilde{x}$ be any point in $C\lt(\mcal{E}\rt)$. There exists a subsequence $\lt\{x_{N_i}\rt\}$, where $x_{N_i} \in E_{N_i}, \forall i$, such that \\$\lim_{i \rar \infty} d\lt(x_{N_i},\tilde{x}\rt)=0$, by the definition of a cluster point. By the same reasoning as the first part of the proof (Eq.~\eqref{eq:tight1}), there exists a sequence $\lt\{y_i\rt\} \subset E_c$ which also converges to $\tilde{x}$. Since $E_c$ is closed, $\tilde{x} \in E_c$.
\end{proof}

We now put the above definition into our context. Define $\mcal{E}=\lt\{E_N\rt\}_{N\geq1}$ to be a sequence of subsets of $D[0,T]$ such that
\beqn
E_N &=& \lt\{x \in D[0,T]: |x\lt(0\rt)-x^0|\leq M_N, \mbox{ and } \rt. \nonumber \\ 
 & & \lt. |x\lt(a\rt)-x\lt(b\rt)| \leq L|a-b| +\gm_N, \quad \forall a,b \in [0,T]\rt\},
\lbl{eq:EN}
\eeqn
where $x^0$ is a constant, $M_N \downarrow 0$ and $\gm_N \downarrow 0$ are two sequences of diminishing non-negative numbers. We first characterize the set of cluster points of the sequence $\mcal{E}$. Loosely speaking, $\mcal{E}$ represents a sequence of sample paths that tend increasingly ``close'' to the set of $L$-Lipschitz continuous functions, and that all elements of $\mcal{E}$ are ``$\gm_N$-approximate'' Lipschitz-continuous. The definition below and the lemma that follows will formalize this notion.

Define $E_c$ as the set of \emph{Lipschitz-continuous} functions on $[0,T]$ with Lipschitz constant $L$ and initial values bounded by a positive constant $M$, defined by:
\beq
E_c \bydef \lt\{x \in D[0,T]: |x\lt(0\rt)| \leq M, \mbox{ and } |x\lt(a\rt)-x\lt(b\rt)| \leq L|a-b|, \forall a,b \in [0,T] \rt\}.
\eeq

We have the following characterization of $E_c$.

\begin{lemma}
\lbl{lm:Eccom}
$E_c$ is compact.
\end{lemma}

\begin{proof}
$E_c$ is a set of $L$-Lipschitz continuous functions $x(\cdot)$ on $[0,T]$ with initial values contained in a closed and bounded interval. By the Arzela-Ascoli theorem, every sequence of points in $E_c$ contains a further subsequence which converges to some $x^*(\cdot)$ uniformly on $[0,T]$. Since all elements in $E_c$ are $L$-Lipschitz continuous, $x^*(\cdot)$ is also Lipschitz continuous on $[0,T]$. It is clear that $x^*(\cdot)$ also satisfies $x^*(0)\leq M$. Hence, $x^*(\cdot) \in E_c$.
\end{proof}

\begin{lemma}
\lbl{lm:EnEc}
$\mcal{E}$ is asymptotically close to $E_c$. 
\end{lemma}

\begin{proof}
It suffices to show for all $x \in E_N$, there exists some $L$-Lipschitz-continuous function $y$ such that
\beq
d\lt(x,y\rt) \leq C\gm_N.
\eeq
where $C$ is a fixed constant, independent of $N$. Fixing $x \in D[0,T]$, such that 
\beq
\lbl{eq:x1d}
|x\lt(a\rt)-x\lt(b\rt)| \leq L|a-b|+\gm, \forall a,b \in [0,T],
\eeq 
we will use a truncation argument to construct an $L$-Lipschitz-continuous function $y\lt(t\rt)$ that uniformly approximates $x\lt(t\rt)$. For the rest of the proof, we use the short-hand $[a\pm\gm]$ to denote the closed interval $[a-\gm,a+\gm]$. The following two claims are useful:

\begin{clm}
\lbl{clm:xy1}
There exist $y_0 \in [x\lt(0\rt)\pm\gm]$ and $y_T \in [x\lt(T\rt)\pm\gm]$ such that  
\beq
|y_T-y_0|\leq TL.
\eeq
In particular, this implies that the linear interpolation between $\lt(0,y_0\rt)$ and $\lt(T,y_T\rt)$
\beq
y\lt(t\rt) \bydef y_0 + \frac{y_T-y_0}{T}t
\eeq
is $L$-Lipschitz-continuous.
\end{clm}

\begin{proof}
Substituting $a=0, b=T$ in Eq.~\eqref{eq:x1d}, we get
\beq
\lbl{eq:smcl1}
-LT-\gm \leq x\lt(0\rt)-x\lt(T\rt) \leq LT + \gm.
\eeq
Write 
\beq
\lbl{eq:smcl2}
y_0-y_T=x\lt(0\rt)-x\lt(T\rt)+ \lt(y_0-x\lt(0\rt)\rt)-\lt(y_T-x\lt(T\rt)\rt).
\eeq
The claim then follows from the above Eqs.~\eqref{eq:smcl1} and \eqref{eq:smcl2}, by noting that $\lt(y_0-x\lt(0\rt)\rt)- \lt(y_T-x\lt(T\rt)\rt)$ can take \emph{any} value between $-2\gm$ and $2\gm$.
\end{proof}

\begin{clm}
\lbl{clm:xy2}
Given any two points $y_0 \in [x\lt(0\rt)\pm\gm]$ and $y_T \in [x\lt(T\rt)\pm\gm]$ such that $\lt|y_0-y_T\rt|\leq TL$, there exists $y_{\frac{T}{2}} \in \lt[x\lt(\frac{T}{2}\rt)\pm\gm\rt]$ such that
\beq
\lbl{eq:ym}
\lt|y_0-y_{\frac{T}{2}}\rt|\leq \frac{TL}{2}, \mbox{ and } \lt|y_{\frac{T}{2}}-y_T\rt|\leq \frac{TL}{2}.
\eeq
\end{clm}

\begin{proof}
Without loss of generality, assume that $y_0 \leq y_T$. We have,
\beq
\lt|y_0-z\rt|\geq \lt|y_T-z\rt|, \forall z \geq \frac{y_0+y_T}{2}, \mbox{ and } |y_0-z| \leq \lt|y_T-z\rt|, \forall z \leq \frac{y_0+y_T}{2}.
\lbl{eq:ymz}
\eeq
By Claim \ref{clm:xy1}, we can find $y_{\frac{T}{2}}^l, y_{\frac{T}{2}}^r \in \lt[x\lt(\frac{T}{2}\rt)\pm\gm\rt]$ such that
\beq
\lt|y_0-y_{\frac{T}{2}}^l\rt| \leq \frac{TL}{2}, \mbox{ and } \lt|y_{\frac{T}{2}}^r-y_T \rt|\leq \frac{TL}{2}.
\eeq
By Eq.~\eqref{eq:ymz}, at least one of $y_{\frac{T}{2}}^l$ and $y_{\frac{T}{2}}^r$ can be used as $y_{\frac{T}{2}}$ to satisfy Eq.~\eqref{eq:ym}. An identical argument applies if $y_0 \geq y_T$.
\end{proof}

Using Claim \ref{clm:xy2}, we can repeat the same process to find $y_{\frac{T}{4}}$ given $y_0$ and $y_{\frac{T}{2}}$, and $y_{\frac{3T}{4}}$ given $y_{\frac{T}{2}}$ and $y_{T}$. Proceeding recursively as such, at the $N$th iteration we will have found a sequence $\lt\{y_{\frac{iT}{2^N}}\rt\}_{i=0}^{2^N}$, such that 
\beq
\lbl{eq:ychop}
y_{\frac{iT}{2^N}} \in \lt[x\lt(\frac{iT}{2^N}\rt)\pm\gm\rt] \mbox{ and} \lt|y_{\frac{iT}{2^N}}-y_{\frac{\lt(i+1\rt)T}{2^N}} \rt|\leq \frac{LT}{2^N}.
\eeq
Denote by $y^N\lt(t\rt)$ the linear interpolation of $\lt\{y_{\frac{iT}{2^N}}\rt\}_{i=0}^{2^N}$, we then have that for all $0 \leq t \leq \frac{T}{2^N}$
\beqn
\lt|y^N\lt(t\rt)-x\lt(t\rt)\rt| &=& \lt|y^N\lt(0\rt)-x\lt(0\rt)+\lt(y^N\lt(t\rt)-y^N\lt(0\rt)\rt)-\lt(x\lt(t\rt)-x\lt(0\rt)\rt) \rt| \nnb \\
&\leq& \lt|y^N\lt(0\rt)-x\lt(0\rt)\rt|+\lt|y^N\lt(\frac{T}{2^N}\rt)-y^N\lt(0\rt)\rt| + \lt|x\lt(\frac{T}{2^N}\rt)-x\lt(0\rt)\rt| \nnb \\
&\leq& \gm + \frac{LT}{2^N} + \lt(\frac{LT}{2^N} + \gm\rt),
\eeqn
where the first two terms in the last inequality follow from Eq.~\eqref{eq:ychop}, and the last term follows from Eq.~\eqref{eq:x1d}. An identical bound on $\lt|y^N\lt(t\rt)-x\lt(t\rt)\rt|$ holds over all other intervals $\lt[\frac{iT}{2^N},\frac{\lt(i+1\rt)T}{2^N}\rt], 1 \leq i \leq N-1$. Since $y^N\lt(t\rt)$ is a piece-wise linear with magnitudes of the slopes no greater than $L$, we have constructed a sequence of $L$-Lipschitz-continuous functions such that
\beq
\sup_{0 \leq t\leq T} \lt|y^N\lt(t\rt)-x\lt(t\rt)\rt| \leq 2\gm+\frac{LT}{2^{N-1}}.
\eeq
The proof for the lemma is completed by letting $C$ be any constant greater than $2$, and pick the approximating function $y$ to be any $y^N$ for sufficiently large $N$. \end{proof}

Finally, the following lemma states that all sample paths $\bux\lt(\omega,\cdot\rt)$ with $\omega \in \spc$ belong to $E_N$, with appropriately chosen $\{M_N\}_{N\geq1}$ and $\{\gamma_N\}_{N\geq1}$.

\begin{lemma}
\lbl{lm:tight1}
Suppose that there exists $\blv^0 \in \overline{\spv}^\ity$ such that for all $\omega \in \spc$
\beq
\lt\|\buv\lt(\omega,0\rt) - \blv^0\rt\|_w \leq \tilde{M}_N,
\eeq
for some $\tilde{M}_N \downarrow 0$. Then for all $\omega \in \spc$ and $i \in \zp$, there exist $L>0$ and sequences $M_N \downarrow 0$ and $\gm_N \downarrow 0$ such that 
\beq
\bux_i \lt(\omega,\cdot\rt) \in E_N,
\eeq
where $E_N$ is defined as in Eq.~\eqref{eq:EN}.
\end{lemma}

\begin{proof}
Intuitively, the lemma follows from the uniform convergence of scaled sample paths of the event process $W^N\lt(\omega,t\rt)$ to $\lt(1+\lambda\rt)t$ (Lemma \ref{lm:nice1}), that jumps along any coordinate of the sample path as a magnitude of $\frac{1}{N}$, and that all coordinates of $\bux$ are dominated by $W$ in terms of the total number of jumps. 

Based on the previous coupling construction, each coordinate of $\bua, \bul$ and $\buc$ are monotonically non-decreasing, with a positive jump at time $t$ of magnitude $\frac{1}{N}$ only if there is a jump of same size at time $t$ in $W\lt(\omega,\cdot\rt)$. Hence for all $i \geq 1$,
\beq
\lt|\bua_i\lt(\omega, a\rt)-\bua_i\lt(\omega,b\rt)\rt| \leq \lt|W^N\lt(\omega,a\rt) - W^N\lt(\omega,b\rt)\rt|, \quad \forall a,b \in [0,T].
\eeq
The same inequalities hold for $\bul$ and $\buc$. Since by construction,
\beq
\lbl{eq:vreq}
\buv_i\lt(\omega,t\rt)=\buv_i\lt(\omega,0\rt)+\bua_i\lt(\omega,t\rt)-\bul_i\lt(\omega,t\rt)-\buc_i\lt(\omega,t\rt), \quad \forall i \geq 1,
\eeq 
we have that for all $i \geq 1$,
\beq
\lbl{lb:xdom}
\lt|\bux_i\lt(\omega, a\rt)-\bux_i\lt(\omega,b\rt)\rt| \leq 3\lt|W^N\lt(\omega,a\rt) - W^N\lt(\omega,b\rt)\rt|, \quad \forall a,b \in [0,T].
\eeq
Since $\omega \in \spc$, $\buw\lt(\omega,\cdot\rt)$ converges uniformly to $\lt(\lambda+1\rt)t$ on [0,T] by Lemma \ref{lm:nice1}. This implies that there exists a sequence $\tilde{\gm}_N \downarrow 0$ such that for all $N \geq 1$,
\beq
\lt|W^N\lt(\omega,a\rt)-W^N\lt(\omega,b\rt)\rt| \leq \lt(\lambda+1\rt)|a-b| +\tilde{\gm}_N, \quad \forall a,b \in [0,T],
\eeq
which, in light of Eq.~\eqref{lb:xdom}, implies
\beq
\lt|\bux_i\lt(\omega,a\rt)-\bux_i\lt(\omega,b\rt)\rt| \leq 3\lt(\lambda+1\rt)|a-b| +3\tilde{\gm}_N, \quad \forall a,b \in [0,T], i \geq 1.
\eeq
Finally, note that all coordinates of $\bux\lt(\omega,0\rt)$ except for $\buv\lt(\omega,0\rt)$ are equal to $0$ by definition. Proof is completed by setting $M_N=2^i\tilde{M}_N$, $\gamma_N=3\tilde{\gamma}_N$, and $L=3\lt(\lambda+1\rt)$.
\end{proof}

We are now ready to prove Proposition \ref{prop:smooth}. 

\begin{proof} {\bf (Proposition \ref{prop:smooth})}
Let us first summarize the key results we have so far:
\benum
\item (Lemma \ref{lm:Eccom}) $E_c$ is a set of $L$-Lipschitz continuous functions with bounded values at $0$, and it is compact and closed. 
\item (Lemma \ref{lm:EnEc}) $\mcal{E}=\lt\{E_N\rt\}_{N\geq1}$, a sequence of sets of $\gm_N$-approximate $L$-Lipschitz-continuous functions with convergent initial values, is asymptotically close $E_c$.
\item (Lemma \ref{lm:tight1}) For all $\omega \in \spc$, $\bux\lt(\omega,\cdot\rt)$ is in $\mcal{E}$.
\eenum
The rest is straightforward: Pick any $\omega \in \spc$. By the above statements, for any $i \in \zp$ one can find a subsequence $\lt\{\mbf{X}^{N_j}\lt(\omega,\cdot\rt)\rt\}_{j=1}^\infty$ and a sequence $\lt\{y_j\rt\}_{j=1}^\infty \subset E_c$ such that
\beq
d\lt(\mbf{X}_i^{N_j}\lt(\omega,\cdot\rt), y_j\rt) \longrightarrow 0, \mbox{ as } j \rar \infty.
\eeq
Since by Lemma \ref{lm:Eccom} (statement $1$ above), $E_c$ is compact and closed, $\lt\{y_j\rt\}_{j=1}^\infty$ has a limit point $y^*$ in $E_c$, which implies that a further subsequence of $\lt\{\mbf{X}_i^{N_j}\lt(\omega,\cdot\rt)\rt\}_{i=1}^\infty$ converges to $y^*$. Moreover, since $\buv(\omega,0) \rar \blv^0$ and $A^N(\omega,0)=L^N(\omega,0)=C^N(\omega,0)=0$, $y^*(0)$ is unique. This proves the existence of a $L$-Lipschitz-continuous limit point $y^*(\cdot)$ at any single coordinate $i$ of $\bux(\cdot)$. 

With the coordinate-wise limit points, we then use a diagonal argument to construct the limit points of $\bux$ in the $D^\zp[0,T]$ space. Now let $v_1(t)$ be any $L$-Lipschitz-continuous limit point of $\buv_1$, so that a subsequence $\mbf{V}^{N^1_j}_1(\omega,\cdot) \rar v_1$ as $j \rar \ity$ in $d(\cdot,\cdot)$. Then proceed recursively by letting $v_{i+1}(t)$ be a limit point of a subsequence of $\lt\{\mbf{V}^{N^{i}_j}_{i+1}(\omega,\cdot)\rt\}_{j=1}^\infty$, where $\{N^i_j\}_{j=1}^\infty$ are the indices for the $i$th subsequence. Finally, define
\beq
\blv_i = v_i, \quad \forall i \in \zp.
\eeq
we claim that $\blv$ is indeed a limit point of $\buv$ in the $d^\zp(\cdot,\cdot)$ norm. To see this, first note that for all $N$,
\beq
\buv_1(\omega,t)\geq\buv_i(\omega,t)\geq0, \quad \forall i\geq1, t\in [0,T].
\eeq
Since we constructed the limit point $\blv$ by repeatedly selecting nested subsequences, this property extends to $\blv$, i.e.,
\beq
\blv_1(t)\geq\blv_i(t)\geq0, \quad \forall i\geq1, t\in [0,T].
\eeq
Since $\blv_1(0) = \blv_1^0$ and $\blv_1(t)$ is $L$-Lipschitz-continuous, we have that
\beq
\lbl{eq:bdv2}
\sup_{t \in [0,T]} \lt|\blv_i(t)\rt| \leq \sup_{t \in [0,T]} \lt|\blv_1(t)\rt| \leq \lt|\blv^0_1\rt| + LT, \quad \forall i\in \zp.
\eeq
Set $N_1=1$, and let
\beq
\lbl{eq:defNk2}
N_k = \min\lt\{N \geq N_{k-1}: \sup_{1\leq i \leq k} d(\mbf{V}^{N}_i(\omega,\cdot),\blv_i) \leq \frac{1}{k}\rt\}, \quad \forall k \geq 2.
\eeq
Note that the construction of $\blv$ implies $N_k$ is well defined and finite for all $k$. From Eq.~\eqref{eq:bdv2} and Eq.~\eqref{eq:defNk2}, we have for all $k \geq 2$
\beqn
d^\zp\lt(\mbf{V}^{N_k}(\omega,\cdot), \blv\rt) &=& \sup_{t\in[0,T]}\sqrt{\sum_{i=0}^\infty\frac{\lt|\mbf{V}^{N_k}_i(\omega,t)-\blv_i(t)\rt|}{2^i}} \nnb \\
&\leq& \frac{1}{k}+\sqrt{\lt(\lt|\blv^0_1\rt| + LT\rt)\sum_{i=k+1}^\infty\frac{1}{2^i}} \nnb\\
&=& \frac{1}{k}+\frac{1}{2^{k/2}} \lt(\lt|\blv^0_1\rt| + LT\rt)
\eeqn
Hence $d^\zp\lt(\mbf{V}^{N_k}(\omega,\cdot), \blv\rt)\rar 0$ as $k\rar \infty$. The existence of the limit points $\bla(t),\bll(t)$ and $\blc(t)$ can be established by an identical argument. This completes the proof.
\end{proof}

\section{Proof of Proposition \ref{prop:vtinght}}
\lbl{app:domproof}


\begin{proof} {\bf (Proposition \ref{prop:vtinght})}
Fix $N>0$ and $0<p\leq1$. For the rest of the proof, denote by $\mbf{V}^{N,p_0}(t)$ the sample path of $\buv(t)$ when $p=p_0$. Let $\{\mbf{V}^{N,p}[n]\}_{n \geq 0}$ be the discrete-time embedded Markov chain for $\mbf{V}^{N,p}(t)$, defined as 
\beq
\mbf{V}^{N,p}[n]\bydef\mbf{V}^{N,p}(t_n), \quad n \geq 0
\eeq
where $t_n, n \geq 1$ is defined previously as the time for the $n$th event taking place in the system (i.e., the $n$th jump in $W^N(\cdot)$), with the convention that $t_0=0$. 

\begin{defn} {\bf (Stochastic Dominance)}
Let $\{X[n]\}_{n\geq0}$ and $\{Y[n]\}_{n\geq0}$ be two discrete-time stochastic processes taking values in $\R^{\zp}$. We say that $\{X[n]\}_{n\geq0}$ is stochastically dominated by $\{Y[n]\}_{n\geq0}$, denoted by $\{X[n]\}_{n\geq0} \preceq_{st} \{Y[n]\}_{n\geq0}$, if there exist random processes $\{X'[n]\}_{n\geq0}$ and $\{Y'[n]\}_{n\geq0}$ defined on a common probability space $(\Omega,\mcal{F},\pb)$, such that
\benum
\item $X'$ and $Y'$ have the same distributions as $X$ and $Y$, respectively.
\item $X'[n]\leq Y'[n], \, \forall n \geq 0$, $\pb-$almost surely.
\eenum
\end{defn}

We have the following lemma:

\begin{lemma}
\lbl{lm:stdom}
Fix any $p\in(0,1]$. If $\mbf{V}^{N,p}[0]=\mbf{V}^{N,0}[0]$, then $\{\mbf{V}_1^{N,p}[n]\}_{n\geq0} \preceq_{st} \{\mbf{V}_1^{N,0}[n]\}_{n\geq0}$.
\end{lemma}

\begin{proof}
We will first interpret the system with $p>0$ as that of an optimal scheduling policy with a time-varying channel. The result will then follow from the classical result in Theorem 3 in \cite{TASEPH93}, with a slightly modified arrival assumption, but almost identical proof steps. Recall the Secondary Motivation described in Section \ref{sec:mot2}. Here we will use a similar but modified interpretation: instead of thinking of the central server as deciding between serving a most-loaded station versus servicing a random station, imagine that the central server always serves a most-loaded station among the ones that are \emph{connected} to it. The channel between the central server and local stations, represented by a set of connected stations, evolves according to the following dynamics and is independent across different time slots:
\benum
\item With probability $p$, all $N$ stations are connected to the central server.
\item Otherwise, only one station, chosen uniformly at random from the $N$ stations, is connected to the central server.
\eenum
It is easy to see that, under the above channel dynamics, a system in which a central server always serves a most-loaded stations among \emph{connected} stations will produce the same distribution for $\mbf{V}^{N,p}[n]$ as our original system. For the case $p=0$, it is equivalent to scheduling tasks under the same channel condition just described, but with a server that servers a station chosen \emph{uniformly at random} among all connected stations. The advantage of the above interpretation is that it allows us to treat $\mbf{V}^{N,p}[n]$ and $\mbf{V}^{N,0}[n]$ as the resulting aggregate queue length processes by applying two \emph{different} scheduling policies to the \emph{same} arrival, token generation, and channel processes. In particular, $\mbf{V}^{N,p}_1[n]$ corresponds to the resulted normalized total queue length process ($\mbf{V}^{N,p}\bydef\frac{1}{N}\sum_{i=1}^N Q_i(t_n)$), when a longest-queue-first policy is applied, and $\mbf{V}^{N,0}_1[n]$ corresponds to the normalized total queue length process, when a fully random scheduling policy is applied. Theorem 3 of \cite{TASEPH93} states that when the arrival and channel processes are symmetric with respect to the identities of stations, the total queue length process under a longest-queue-first policy is stochastically dominated by all other causal policies (i.e., policies that use only information from the past). Since the arrival and channel processes are symmetric in our case, and a random scheduling policy falls under the category of causal policies, the statement of Theorem 3 of \cite{TASEPH93} implies the validity of our claim. 

There is, however, a minor difference in the assumptions of Theorem 3 of \cite{TASEPH93} and our setup that we note here. In \cite{TASEPH93}, it is possible that both arrivals and service occur during the same slot, while in our case, each event corresponds either to the an arrival to a queue or the generation of a service token, but not both. This technical difference can be overcome by discussing separately, whether the current slot corresponds to an arrival or a service. The structure of the proof for Theorem 3 in \cite{TASEPH93} remains unchanged after this modification, and is hence not repeated here.
\end{proof}

Using the discrete-time stochastic dominance result in Lemma \ref{lm:stdom}, we can now establish a similar dominance for the continuous-time processes $\mbf{V}_1^{N,p}(t)$ and $\mbf{V}_1^{N,0}(t)$. Since $\{\mbf{V}_1^{N,p}[n]\}_{n\geq0} \preceq_{st} \{\mbf{V}_1^{N,0}[n]\}_{n\geq0}$, by the definition of stochastic dominance, we can construct $\{\mbf{V}_1^{N,p}[n]\}_{n\geq0}$ and $\{\mbf{V}_1^{N,0}[n]\}_{n\geq0}$ on a common probability space $(\Omega_d, \mcal{F}_d, \pb_d)$, such that $\mbf{V}_1^{N,p}[n] \leq \mbf{V}_1^{N,0}[n]$ for all $n\geq0$, $\pb_d$-almost surely. Recall from previous sections that $W^N(t)$, the $N$th event process, is a Poisson jump process with rate $N(1+\lambda)$, defined on the probability space $(\Omega_W,\mcal{F}_W,\pb_W)$. Let $(\Omega_c,\mcal{F}_c,\pb_c)$ be the product space of $(\Omega_d, \mcal{F}_d, \pb_d)$ and $(\Omega_W,\mcal{F}_W,\pb_W)$. Define two continuous-time random processes on $(\Omega_c,\mcal{F}_c,\pb_c)$ by
\beqn
\lbl{eq:asdom}
& \hat{\mbf{V}}^{N,p}(t) \bydef \mbf{V}^{N,p}\lt[ N W^N(t) \rt], \\
& \mbox{and }\hat{\mbf{V}}^{N,0}(t) \bydef \mbf{V}^{N,0}\lt[ NW^N(t) \rt].
\eeqn
Note that $NW^N(t)$ is a Poisson jump process, and hence $NW^N(t) \in \zp$ for all $0 \leq t < \ity$ $\pb_c$-almost surely. Since $\mbf{V}_1^{N,p}[n] \leq \mbf{V}_1^{N,0}[n]$ for all $n$ almost surely by construction, this implies
\beq
\hat{\mbf{V}}^{N,p}_1(t) \leq \hat{\mbf{V}}^{N,0}_1(t), \, \forall t \geq 0, \quad \mbox{almost surely.}
\eeq
Recall the processes $\{\mbf{V}^{N,p}[n]\}_{n\geq0}$ and $\{\mbf{V}^{N,0}[n]\}_{n\geq0}$ were defined to be the embedded discrete-time processes for $\mbf{V}^{N,p}(t)$ and $\mbf{V}^{N,0}(t)$. It is also easy to check that the continuous-time Markov process $\mbf{V}^{N,p}(t)$ is \emph{uniform} for all $N$ and $p$ (i.e., the rate until next event is uniform at all states). Hence, the processes $\hat{\mbf{V}}^{N,p}(t)$ and $\hat{\mbf{V}}^{N,0}(t)$ constructed above have the \emph{same distributions} as the original processes ${\mbf{V}}^{N,p}(t)$ and ${\mbf{V}}^{N,0}(t)$, respectively. Therefore, as we work with the processes $\hat{\mbf{V}}^{N,p}(t)$ and $\hat{\mbf{V}}^{N,0}(t)$ in the rest of the proof, it is understood that any statement regarding the distributions of $\hat{\mbf{V}}^{N,p}(t)$ and $\hat{\mbf{V}}^{N,0}(t)$ automatically holds for ${\mbf{V}}^{N,p}(t)$ and ${\mbf{V}}^{N,0}(t)$, and vice versa.

We first look at the behavior of $\hat{\mbf{V}}^{N,0}(t)$. When $p=0$, only local service tokens are generated. Hence, it is easy to see that the system degenerates into $N$ individual $M/M/1$ queues with independent and identical statistics for arrivals and service token generation. In particular, for any station $i$, the arrival follows an Poisson process of rate $\lambda$ and the generation of service tokens follows a Poisson process of rate $1$. Since $\lambda<1$, it is not difficult to verify that the process $\hat{\mbf{V}}^{N,0}(t)$ is positive recurrent, and it admits a unique steady-state distribution, denoted by $\pi^{N,0}$, which satisfies:
\beq
\lbl{eq:p0dis}
\pi^{N,0}\lt(\mbf{V}_1 \leq x\rt)  = \pb\lt(\frac{1}{N}\sum_{i=1}^N E_i \leq x\rt), \quad \forall x \in \R,
\eeq
where $\{E_i\}_{i=1}^N$ is a set of i.i.d. geometrically distributed random variables, with
\beq
\pb(E_i=k) = \lambda^k(1-\lambda), \quad \forall k \in \zp,
\eeq

We now argue that $\hat{\mbf{V}}^{N,p}(t)$ is also positive-recurrent for all $p\in (0,1]$. Let 
\beq
\tau^p \bydef \inf\lt\{t > 0:\hat{\mbf{V}}_1^{N,p}(t)=0, \mbox{ and } \hat{\mbf{V}}_1^{N,p}(s)\neq0 \mbox{ for some } 0 < s < t\rt\}
\eeq
In other words, if the system starts empty (i.e., $\hat{\mbf{V}}_1^{N,p}(0)=0$), $\tau^p$ is the first time that the system becomes empty again after having visited some non-empty state. Since the process $\hat{\mbf{V}}_1^{N,p}(t)$ can be easily verified to be irreducible (i.e., all states communicate) for all $p\in (0,1]$, $\hat{\mbf{V}}_1^{N,p}(t)$ is positive-recurrent if and only if
\beq
\mathbb{E}\lt[\tau^p \middle| \mbf{V}_1^{N,p}(0)=0 \rt] < \ity.
\eeq

Since $\hat{\mbf{V}}^{N,p}(t) \leq \hat{\mbf{V}}^{N,0}(t), \, \forall t \geq 0 \,$ almost surely, it implies that $\tau^p \leq \tau^0$ almost surely. From the positive recurrence of $\hat{\mbf{V}}^{N,0}(t)$, we have
\beq
\mathbb{E}\lt[\tau^p \middle|\hat{\mbf{V}}_1^{N,p}(0)=0 \rt] \leq \mathbb{E}\lt[\tau^0 \middle|\hat{\mbf{V}}_1^{N,0}(0)=0 \rt] < \ity.
\eeq
This establishes that $\hat{\mbf{V}}^{N,p}(t)$ is positive-recurrent for all $p \in (0,1]$.

To complete the proof, we need the following standard result from the theory of Markov processes (see, e.g., \cite{JRN97}).
\begin{lemma}
\lbl{lm:clasmarkov}
If $X(t)$ is an irreducible and positive recurrent Markov process taking values in a countable set $\mcal{I}$, then there exists a unique steady-state distribution $\pi$ such that for any initial distribution of $X(0)$,
\beq
\lim_{t\rar\ity}\pb\lt(X(t)=i\rt) = \pi(i), \quad \forall i \in \mcal{I}.
\eeq
\end{lemma}

By the positive recurrence of $\hat{\mbf{V}}^{N,p}(t)$ and Lemma \ref{lm:clasmarkov}, we have that $\hat{\mbf{V}}^{N,p}(t)$ converges in distribution to a unique steady-state distribution $\pi^{N,p}$ as $t\rar\ity$. Combining this with the dominance relation in Eq.~\eqref{eq:asdom}, we have that for any initial distribution of $\hat{\mbf{V}}^{N,p}(0)$,
\beqn
\pi^{N,p}(\overline{\spv}^M) &\bydef& \pi^{N,p}(\buv_1 \leq M)  \nnb \\
&=& \lim_{t \rar \ity} \pb\lt(\hat{\mbf{V}}^{N,p}(t) \leq M\rt) \quad \quad \mbox{(by Lemma \ref{lm:clasmarkov})} \nnb\\
&\geq& \lim_{t \rar \ity} \pb\lt(\hat{\mbf{V}}^{N,0}(t) \leq M\rt)  \quad \quad \mbox{(by Eq.\ \eqref{eq:asdom})}\nnb\\
&=& \pi^{N,0}(\buv_1 \leq M)  \quad \quad \quad \quad \; \; \mbox{(by Lemma \ref{lm:clasmarkov})} \nnb \\
&=& \pb\lt(\frac{1}{N}\sum_{i=1}^N E_i \leq M\rt)  \quad \; \; \; \;  \; \quad \mbox{(by Eq.\ \eqref{eq:p0dis})}
\eeqn
Since the $E_i$s are i.i.d. geometric random variables, by Markov's inequality,
\beqn
&&\pi^{N,p}(\overline{\spv}^M)\geq 1-\pb\lt(\frac{1}{N}\sum_{i=1}^N E_i \geq M\rt) \geq 1- \frac{\mathbb{E}(E_1)}{M} = 1- \frac{\lambda}{(1-\lambda)M},
\eeqn
for all $M > \mathbb{E}(E_1)=\frac{\lambda}{1-\lambda}$, which establishes the tightness of $\{\pi^{N,p}\}_{N=1}^\ity$. This completes the proof of Proposition \ref{prop:vtinght}.
\end{proof}

\chapter{Appendix: Simulation Setup}
\lbl{app:sim}
The simulation results shown in Figure~\ref{fig:heavytraffic} for a finite system with $100$ stations were obtained by simulating the embedded discrete-time Markov chain, $\{Q[n]\}_{n \in \N}$, where the vector $Q[n]\in \zp^{100}$ records the queue lengths of all 100 queues at time step $n$. Specifically, we start with $Q[1]=0$, and,  during each time step, one of the following takes place:
\benum
\item With probability $\frac{\lambda}{1+\lambda}$, a queue is chosen uniformly at random from all queues, and one new task is added to this queue. This corresponds to an arrival to the system.
\item With probability $\frac{1-p}{1+\lambda}$, a queue is chosen uniformly at random from all queues, and one task is removed from the queue if the queue is non-empty. If the chosen queue is empty, no change is made to the queue length vector. This corresponds to the generation of a local service token.
\item  With probability $\frac{p}{1+\lambda}$, a queue is chosen uniformly at random from the \emph{longest queues}, and one task is removed from the chosen queue if the queue is non-empty. If all queues are empty, no change is made to the queue length vector. This corresponds to the generation of a central service token.
\eenum
To make the connection between the above discrete-time Markov chain $Q[n]$ and the  continuous-time Markov process $Q(t)$ considered in this thesis, one can show that $Q(t)$ is uniformized and hence the steady-state distribution of $Q(t)$ coincides with that of the embedded discrete-time chain $Q[n]$. 

To measure the steady-state queue length distribution seen by a typical task, we sampled from the chain $Q[n]$ in the following fashion: $Q[n]$ was first run for a burn-in period of $1,000,000$ time steps, after which $500,000$ samples were collected with $20$ time steps between adjacent samples, where each sample recorded the current length of a queue chosen uniformly at random from all queues. Denote by $\mathbf{S}$ the set of all samples. The average queue length, as marked by the symbol ``$\times$'' in Figure~\ref{fig:heavytraffic}, was computed by taking the average over $\mathbf{S}$. The upper (UE) and lower (LE) ends of the $95\%$ confidence intervals were computed by:
\beqn
UE &\bydef& \min\{x\in \mathbf{S}: \mbox{ there are no more than } 2.5\% \nnb \\
& & \hspace{-10pt} \mbox{of the elements of } \mathbf{S} \mbox{ that are strictly greater than } x \}, \nnb \\
LE &\bydef& \max\{x\in \mathbf{S}: \mbox{ there are no more than } 2.5\% \nnb \\
& & \mbox{of the elements of } \mathbf{S} \mbox{ that are strictly less than } x \}. \nnb
\eeqn
Note that this notion of confidence interval is meant to capture the concentration of $\mathbf{S}$ around the mean, and is somewhat different from that used in the statistics literature for parameter estimation.

A separate version of the above experiment was run for each value of $\lambda$ marked in Figure~\ref{fig:heavytraffic}, while the the level of centralization $p$ was fixed at $0.05$ across all experiments.
\clearpage
\newpage

\begin{singlespace}

\end{singlespace}

\end{document}